\documentclass[sigplan,9pt]{acmart}

\usepackage{booktabs}   
\usepackage{subcaption} 

\newif\ifprod
\prodtrue

\newif\iflong
\longtrue

\ifprod
    \definecolor{longcolor}{rgb}{0.0, 0.0, 0.0}
    \definecolor{shortcolor}{rgb}{0.0, 0.0, 0.0}
    \newcommand{\grammarr}[1]{#1}
    \newcommand{\note}[1]{}
    \newcommand{\ConsiderRemoving}[1]{#1}
    \newcommand{\FutureWork}[1]{}
    \newcommand{\todo}[1]{}
\else
    \definecolor{longcolor}{rgb}{0.3, 0.3, 0.0}
    \definecolor{shortcolor}{rgb}{0.0, 0.3, 0.3}
    \newcommand{\note}[1]{\textcolor{red}{[[{#1}]]}}
    \newcommand{\grammarr}[1]{\textcolor{red}{\underline{#1}}}
    \newcommand{\ConsiderRemoving}[1]{\textcolor{blue}{#1}}
    \newcommand{\FutureWork}[1]{\textcolor{brown}{#1}}
    
    \newcommand{\todo}[1]{\textbf{\color{red}[#1]}}
\fi

\iflong
    \renewcommand\footnotetextcopyrightpermission[1]{}
    \newcommand{\longv}[1]{\textcolor{longcolor}{#1}}
    \newcommand{\shortv}[1]{}
    \newcommand{\shortlong}[2]{\textcolor{longcolor}{#2}}
    \excludecomment{shortversion}
    \newenvironment{longversion}
    {\color{longcolor}}
    {}
\else
    \newcommand{\longv}[1]{}
    \newcommand{\shortv}[1]{\textcolor{shortcolor}{#1}}
    \newcommand{\shortlong}[2]{\textcolor{shortcolor}{#1}}
    \excludecomment{longversion}
    \newenvironment{shortversion}
    {\color{shortcolor}}
    {}
\fi

\makeatletter\if@ACM@journal\makeatother
\acmJournal{PACMPL}
\acmVolume{1}
\acmNumber{1}
\acmArticle{1}
\acmYear{2017}
\acmMonth{1}
\acmDOI{10.1145/nnnnnnn.nnnnnnn}
\startPage{1}
\else\makeatother
\copyrightyear{2018} 
\acmYear{2018} 
\setcopyright{rightsretained} 
\acmConference[LICS '18]{33rd Annual ACM/IEEE Symposium on Logic in Computer Science}{July 9--12, 2018}{Oxford, United Kingdom}
\acmDOI{10.1145/3209108.3209193}
\acmISBN{978-1-4503-5583-4/18/07}
\startPage{1}
\fi

\usepackage{amssymb}
\usepackage{url}
\usepackage{tikz}
\usepackage{tikz-cd}
\usepackage{mathpartir}
\usepackage{fancyvrb}
\usepackage{hyperref}
\usepackage{listings}

\usepackage{amsthm}

\usepackage{comment}

\definecolor{SubtleColor}{rgb}{0,0,.50}

\definecolor{green}{rgb}{0.0, 0.5, 0.0}
\definecolor{red}{rgb}{0.8, 0.0, 0.0}

\definecolor{glsdiminished}{rgb}{0.5, 0.5, 0.5}
\newcommand{\glsdiminished}[1]{\textcolor{glsdiminished}{#1}}

\usepackage{glossaries}
\glstoctrue
\glsdisablehyper

\makeglossaries

\newglossaryentry{R}
{
  name={\ensuremath{\mathbb{R}}},
  description={The space of real numbers},
}
\newcommand{\R}{\gls{R}}

\newglossaryentry{nat}
{
  name={\ensuremath{\mathbb{N}}},
  description={The type/set/discrete space of natural numbers},
}

\newglossaryentry{bool}
{
  name={\ensuremath{\mathbb{B}}},
  description={The type/set/discrete space of Boolean values},
}
\newcommand{\bool}{\gls{bool}}

\newglossaryentry{rat}
{
  name={\ensuremath{\mathbb{Q}}},
  description={The set/discrete space of rational numbers},
}
\newcommand{\rat}{\gls{rat}}

\newglossaryentry{Zero}
{
  name={\ensuremath{\varnothing}},
  description={The empty space},
}

\newglossaryentry{One}
{
  name={\ensuremath{\ast}},
  description={Either the type with one element or the space with one point},
}
\newcommand{\One}{\gls{One}}

\newglossaryentry{btrue}
{
  name={\ensuremath{\mathsf{true}}},
  description={Boolean true, $\mathsf{true} : \One \to_c \bool$},
}
\newcommand{\btrue}{\gls{btrue}}

\newglossaryentry{bfalse}
{
  name={\ensuremath{\mathsf{false}}},
  description={Boolean false, $\mathsf{false} : \One \to_c \bool$},
}
\newcommand{\bfalse}{\gls{bfalse}}

\newglossaryentry{bor}
{
  name={\ensuremath{\mathbin{||}}},
  description={Boolean or, $\mathbin{||} : \bool \times \bool \to_c \bool$},
}
\newcommand{\bor}{\gls{bor}
}

\newglossaryentry{band}
{
  name={\ensuremath{\mathbin{\&\&}}},
  description={Boolean and, $\mathbin{\&\&} : \bool \times \bool \to_c \bool$},
}
\newcommand{\band}{\gls{band}}

\newglossaryentry{bneg}
{
  name={\ensuremath{\mathop{!}}},
  description={Boolean negation, $\mathop{!} : \bool \to_c \bool$},
}
\newcommand{\bneg}{\gls{bneg}}

\newglossaryentry{Type}
{
  name={\ensuremath{\mathcal{U}}},
  description={The type of types, $\mathcal{U} : \mathcal{U}$},
}
\newcommand{\Type}{\gls{Type}}

\newglossaryentry{Prop}
{
  name={\ensuremath{\Omega}},
  description={The (large) set of propositions},
}
\newcommand{\Prop}{\gls{Prop}}

\newglossaryentry{PLower}
{
  name={\ensuremath{\mathcal{P}_\lozenge}},
  description={The partial \& nondeterministic powerspace monad, $\mathcal{P}_\lozenge : \cat{FSpc} \to \cat{FSpc}$},
}
\newcommand{\PLower}{\gls{PLower}}

\newglossaryentry{PLowerP}
{
  name={\ensuremath{\mathcal{P}_\lozenge^+}},
  description={The nondeterministic powerspace monad, $\mathcal{P}_\lozenge^+ : \cat{FSpc} \to \cat{FSpc}$},
}
\newcommand{\PLowerP}{\gls{PLowerP}}

\newglossaryentry{PUpper}
{
  name={\ensuremath{\mathcal{P}_\square}},
  description={The compact powerspace monad, $\mathcal{P}_\square : \cat{FSpc} \to \cat{FSpc}$},
}
\newcommand{\PUpper}{\gls{PUpper}}

\newglossaryentry{Viet}
{
  name={\ensuremath{\mathcal{P}_{\square\lozenge}}},
  description={The compact/overt powerspace monad, $\mathcal{P}_{\square\lozenge} : \cat{FSpc} \to \cat{FSpc}$},
}
\newcommand{\Viet}{\gls{Viet}}

\newglossaryentry{iimg}
{
  name={$\glsdiminished{f}^*$},
  description={The inverse image map $f^* : \Open{B} \to \Open{A}$ of a map $f : A \ndpto B$},
}

\newglossaryentry{dimg}
{
  name={$\glsdiminished{f}_!$},
  description={The direct image map $f_! : \Open{A} \to \Open{B}$ of an open map $f : A \to_o B$},
}

\newglossaryentry{BOpen}
{
  name={$\mathcal{O}_\mathsf{B}\glsdiminished{(A)}$},
  text={\mathcal{O}_\mathsf{B}},
  description={The set of basic opens of a space $A$},
}
\newcommand{\BOpen}{\gls{BOpen}}

\newglossaryentry{Open}
{
  name={$\mathcal{O}\glsdiminished{(A)}$},
  text={\mathcal{O}},
  description={The (large) set of opens of a space $A$},
}

\newcommand{\Open}[1]{\gls{Open}\left({#1}\right)}

\newglossaryentry{up}
{
  name={\ensuremath{\mathsf{up}}},
  description={The open embedding of a space in its lifting, $\mathsf{up} : A \hookto A_\bot$},
}
\newcommand{\up}{\gls{up}}

\newglossaryentry{parallel}
{
  name={$\glsdiminished{f} \otimes \glsdiminished{g}$},
  text = {\otimes},
  description={The parallel composition $f \otimes g : A \times X \to_c B \times Y$ of $f : A \to_c B$ and $g : X \to_c Y$},
}
\newcommand{\parmap}{\gls{parallel}}

\newglossaryentry{liesin}
{
  name={$\glsdiminished{x} \models \glsdiminished{U}$},
  text = {\models},
  description={A point $x$ in a space $A$ lies in an open $U$ of $A$},
}
\newcommand{\liesin}{\gls{liesin}}

\newglossaryentry{iso}
{
  name={$\glsdiminished{A} \cong \glsdiminished{B}$},
  text = {\cong},
  description={Objects $A$ and $B$ in a category are isomorphic (usually, either two sets are isomorphic or two spaces are homeomorphic)},
}
\newcommand{\iso}{\gls{iso}}

\newglossaryentry{oinclf}
{
  name={$\iota[\glsdiminished{U}]$},
  text = {\iota},
  description={The open embedding $\iota[U] : \{ A \mid U \} \hookto A$ of an open subspace $U$ of $A$ into the entire space},
}
\newcommand{\oinclf}[1]{\gls{oinclf} [{#1}]}
\newcommand{\oincl}[2]{\oinclf{#1} \left({#2}\right)}

\newglossaryentry{omeet}
{
  name={$\glsdiminished{a} \downarrow \glsdiminished{b}$},
  text = {\downarrow},
  description={The meet (intersection) of two basic opens $a$ and $b$ of a space $A$ (which is an open of $A$)},
}
\newcommand{\omeet}{\gls{omeet}}

\newglossaryentry{cov}
{
  name={$\glsdiminished{a} \triangleleft \glsdiminished{U}$},
  text = {\triangleleft},
  description={The basic open $a$ of a space $A$ is covered by the open $U$ of $A$},
}
\newcommand{\cov}{\glsadd{cov} \triangleleft}

\newglossaryentry{inteq}
{
  name={$\glsdiminished{a} \equiv \glsdiminished{b}$},
  text = {\equiv},
  description={Intensional equality of $a$ and $b$ (which must have the same type)},
}

\newcommand{\inteq}{\gls{inteq}}

\newglossaryentry{iff}
{
  name={\ensuremath{\Longleftrightarrow}},
  description={If and only if},
}
\newcommand{\ifandonlyif}{\gls{iff}}

\newglossaryentry{cto}
{
  name={$\glsdiminished{A} \to_c \glsdiminished{B}$},
  text = {\to_c},
  description={The set of continuous maps from a space $A$ to a space $B$},
}
\newcommand{\cto}{\gls{cto}}

\newglossaryentry{ndto}
{
  name={$\glsdiminished{A} \to_{nd} \glsdiminished{B}$},
  text = {\to_{nd}},
  description={The set of nondeterministic maps from a space $A$ to a space $B$},
}
\newcommand{\ndto}{\gls{ndto}}

\newglossaryentry{pto}
{
  name={$\glsdiminished{A} \to_p \glsdiminished{B}$},
  text = {\to_p},
  description={The set of partial maps from a space $A$ to a space $B$},
}
\newcommand{\pto}{\gls{pto}}

\newglossaryentry{ndpto}
{
  name={$\glsdiminished{A} \, \to_{nd, p}\, \glsdiminished{B}$},
  text = {\to_{nd, p}},
  description={The set of nondeterministic and partial maps from a space $A$ to a space $B$},
}
\newcommand{\ndpto}{\gls{ndpto}}

\newglossaryentry{Sierp}
{
  name={\ensuremath{\Sigma}},
  description={The Sierpi\'nski space, $\Sigma \iso \One_\bot \iso \PLower(\One)$},
}
\newcommand{\Sierp}{\gls{Sierp}}

\newglossaryentry{strue}
{
  name={\ensuremath{\top_{\Sierp}}},
  description={The ``true'' open point of the Sierpi\'nski space, $\top_{\Sierp} : \One \to_c \Sierp$},
}
\newcommand{\strue}{\gls{strue}}

\newglossaryentry{sfalse}
{
  name={\ensuremath{\bot_{\Sierp}}},
  description={The ``false'' closed point of the Sierpi\'nski space, $\bot_{\Sierp} : \One \to_c \Sierp$},
}
\newcommand{\sfalse}{\gls{sfalse}}

\newglossaryentry{Lifted}
{
  name={$\cdot_\bot$},
  text={\bot},
  description={The lifting monad on spaces},
}
\newcommand{\Lifted}[1]{#1_{\gls{Lifted}}}

\newglossaryentry{possibly}
{
  name={$\lozenge \glsdiminished{U}$},
  text = {\lozenge},
  description={The direct image of an open $U$ of $A$ in its partial and nondeterministic powerspace $\PLower(A)$ or some subspace of it},
}
\newcommand{\possibly}{\gls{possibly}}

\newglossaryentry{Branch}
{
  name={$\glsdiminished{p} \Rightarrow \glsdiminished{e}$},
  text = {\Rightarrow},
  description={A branch in a pattern match, where $p$ is a pattern such that, if matched, the expression $e$ results},
}
\newcommand{\Branch}{\gls{Branch}}

\newglossaryentry{wildcard}
{
  name={\ensuremath{\_}},
  description={A wildcard in a pattern, which matches anything},
}
\newcommand{\wildcard}{\gls{wildcard}}

\newglossaryentry{leto}
{
  name={$\glsdiminished{f} \leqq \glsdiminished{g}$},
  text = {\leqq},
  description={The truth order implication $f$ implies $g$ for maps $f, g : \Gamma \ndpto \bool$},
}
\newcommand{\leto}{\gls{leto}}

\newcommand{\commentt}[1]{}
\newcommand{\iimg}[1]{#1^*}

\newcommand{\refsection}[1]{\S\ref{#1}}
\newcommand{\hookto}{\hookrightarrow}

\newcommand{\xto}[1]{\to_{#1}}

\newcommand{\denote}[1]{\llbracket #1 \rrbracket}

\newcommand{\cat}[1]{\textbf{#1}}
\newcommand{\functor}[1]{\mathsf{#1}}
\newcommand{\fun}[2]{\lambda {#1}.\  {#2}}
\newcommand{\irule}[1]{\textsc{#1}}
\newcommand{\restrict}[2]{{#1}_{|{#2}}}

\tikzset{
 thicker/.style={line width=#1\pgflinewidth},
 thicker/.default={2},
}

\newcommand{\rootfindingcode}{
\mathsf{roots}_f &: \One \ndto \{ \One \mid \forall x \in K.\ f(x) \ne 0 \} + 
  \{ x : K \mid | f(x) | < \varepsilon \}
\\ \mathsf{roots}_f &\triangleq
\mathsf{case}(\mathsf{tt})
\begin{cases}
\oincl{\exists x \in K.\ |f(x)| < \varepsilon}{y}
  &\Branch \mathsf{inr}(\mathsf{simulate}(y))
\\ \oincl{\forall x \in K.\ f(x) \neq 0}{n}
  &\Branch \mathsf{inl}(n)
\end{cases}
}

\newcommand{\bijectivecorrespondence}[2]{#1 \quad &\cong \quad #2}

\newcommand{\grammar}[1]{\grammarr{#1}}

\lstdefinelanguage{marshall} {
morekeywords={cut, left, right, True, False, let, in, fun, real, prop, bool},
otherkeywords={<,~>,=>,->,>,/\,\/,||,;;,=,+,*,-,:},
sensitive=false,
morecomment=[l]{!},
morestring=[b]", }

\usepackage{letltxmacro}
\newcommand*{\SavedLstInline}{}
\LetLtxMacro\SavedLstInline\lstinline
\DeclareRobustCommand*{\lstinline}{%
  \ifmmode
    \let\SavedBGroup\bgroup
    \def\bgroup{%
      \let\bgroup\SavedBGroup
      \hbox\bgroup
    }%
  \fi
  \SavedLstInline
}

\begin{document}

\title{Computable decision making on the reals and other spaces}       
\subtitle{via partiality and nondeterminism}                     
\iflong
\subtitlenote{This is an extended version of a paper with the same name in the proceedings of the Thirty-Third Annual ACM/IEEE Symposium on Logic in Computer Science (LICS) in July 2018.}       %
\fi


\author{Benjamin Sherman}
\orcid{0000-0002-9569-3393}             
\affiliation{
  \institution{MIT CSAIL}            
}
\email{sherman@csail.mit.edu}          

\author{Luke Sciarappa}
\affiliation{
  \institution{MIT CSAIL}            
}
\email{lukesci@mit.edu}          

\author{Adam Chlipala}
\affiliation{
  \institution{MIT CSAIL}            
}
\email{adamc@csail.mit.edu}          

\author{Michael Carbin}
\affiliation{
  \institution{MIT CSAIL}            
}
\email{mcarbin@csail.mit.edu}          


\begin{abstract}
Though many safety-critical software systems use floating point to represent real-world input and output, programmers usually have idealized versions in mind that compute with real numbers. Significant deviations from the ideal can cause errors and jeopardize safety.
Some programming systems implement exact real arithmetic, which resolves this matter but complicates others, such as decision making.
In these systems, it is impossible to compute 
(total and deterministic) discrete decisions based on connected spaces such as $\R$.
We present programming-language semantics based on constructive topology
with variants allowing nondeterminism and/or partiality.
Either nondeterminism or partiality suffices to allow computable decision making on connected spaces such as $\R$.
We then introduce \emph{pattern matching on spaces}, a language construct for creating programs on spaces, generalizing pattern matching in functional programming, where patterns need not represent decidable predicates and also may overlap or be inexhaustive, giving rise to nondeterminism or partiality, respectively.
Nondeterminism and/or partiality also yield formal \emph{logics for constructing approximate decision procedures}.
We extended the Marshall language for exact real arithmetic with these constructs.
\end{abstract}



\maketitle

\section{Introduction}

Ensuring the safety of software controlling cyber-physical systems can
be challenging, at least in part due to the need to compute with values
that are continuous in nature, such as space, time, magnitude, and probability.
When these values are represented unsoundly, such as with finite precision,
failure can result from numerical error alone.
Verification necessitates both guaranteeing accuracy of computations with
continuous values and idealized reasoning about a system's
behavior.

Programming systems implementing exact real arithmetic \cite{oconnor2008, lamcra, marshall} do guarantee
accuracy and have been used to develop verified cyber-physical systems \cite{roscoq}.
While these programming systems do ease
development of traditionally continuous computations on the reals,
there has been little investigation of how to soundly
incorporate decision-making: computations from the reals ($\R$) to the
Booleans ($\bool$).  Classic results prove that it is impossible to compute
(total and deterministic) discrete decisions on connected spaces such as $\R$
\cite{weihrauchintro}\note{cite Brouwer per review \#1}.

\commentt{
\todo{ideally we should say, "We show that the only way to compute discrete
decision on connected spaces is"} 
}

However, we show that by allowing partiality or nondeterminism into the
computational model, we can enable decision-making while retaining fairly
strong computational abilities. We present programming-language semantics based
on \emph{constructive topology} with variants allowing nondeterminism and/or
partiality. 

Constructive topology, in the form of \emph{locale theory}, provides a single
programming language in which it is possible to build and execute programs that
compute with continuous values and to reason about these programs in terms of
their mathematical descriptions. In this programming language (category)
\cat{FSpc}, types (objects) are \emph{spaces} and programs (morphisms) are
\emph{continuous maps}.

\paragraph{Types are spaces.}
Spaces are defined as theories of \emph{geometric logic} \cite{vickersloctopspaces}: propositional
symbols describe the core observable properties of the space, and axioms
describe which properties imply others. Points of a space are models of its
theory.

For example, the theory for $\R$ has as its propositional symbols the open balls with rational centers, $q - \varepsilon < \cdot < q + \varepsilon$ (for each $q : \rat$, $\varepsilon : \rat^+$), and an
example axiom is the one $ \top \le \bigvee_{q : \rat} \left( q - \varepsilon <
\cdot < q + \varepsilon \right)$ that says (for any $\varepsilon : \rat^+$)
that without assumptions ($\top$), a point must be within $\varepsilon$ of
\emph{some} rational $q$.
Axioms with disjunctions on the right, like this one, are called \emph{open covers}.
They can be read computationally. For instance, since the point $\pi$ lies in $\top$ (as every point does),
it must lie in some ball of radius 1 with a rational center,
and indeed it should be able to \emph{compute} any such rational:
3 would be one possible choice, since $3 - 1 < \pi < 3 + 1$.

\paragraph{Programs are continuous maps.} One
defines a continuous map $f : A \cto B$ by programming how it reduces open
covers of $B$ to open covers of $A$.  Accordingly, computation is pull-based,
where a composition of functions successively reduces open covers of the output
to open covers of the input, at which point the input computes which open of
the cover it lies in, which corresponds to a particular open that the output
lies in. Constructive topology has surprisingly strong computational abilities
\cite{lamcra, escardoinfinite, simpson1998lazy}, such as the ability to compute
the maximum that a real-valued function attains over a compact-overt
space (see Definition \ref{compact-overt}). 

\subsection{Contributions}

Making nontrivial total and deterministic decisions based on connected spaces
is impossible: any continuous map $f : C \cto D$ from a connected space $C$ to
a discrete space $D$ must be constant.  We demonstrate that decisions
\emph{can} be made, however, by permitting either partiality or nondeterminism,
and we continue to then present the following contributions:

\commentt{
 \refsection{parnondet} describe an
integrated characterization of these concepts for continuous maps.

Guided by
this characterization, we generalize two tools from the discrete world, pattern
matching and decision procedures, yielding two programming-language constructs
for constructive topology that facilitate decision-making with continuous
values}

\paragraph{Partial and/or nondeterministic maps.} \refsection{parnondet} defines
partiality and nondeterminism as they relate to continuous maps and
\refsection{openmaps} characterizes the open maps and open embeddings as those
continuous maps having partial and/or nondeterministic inverses. While each of
these subjects has been studied individually in the context of constructive
topology, we contribute the first integrated characterization relating them.

\paragraph{Pattern matching on spaces.} \refsection{patterns} generalizes
pattern matching on inductive types in functional programming to spaces. It
differs in that patterns need not correspond to decidable predicates, and
patterns are allowed to overlap or fail to be injective, yielding nondeterminism, or be inexhaustive,
yielding partiality.

\paragraph{Formal logics for approximate decision procedures.}
\refsection{bcover} generalizes the decidable predicates of functional
programming to \emph{approximately} decidable predicates on spaces,
which may either be partial or nondeterministic.
Because spaces often have few decidable predicates, this
relaxation is essential for decision-making on spaces.
The Boolean
algebra of decidable predicates generalizes to a quasi-Boolean algebra, and
quantification over finite sets is generalized to quantification over
\emph{compact-overt} spaces (see Definition \ref{compact-overt}). The partial
logic and nondeterministic logic are observed to be duals.

\paragraph{Case study.} We have extended the Marshall language~\cite{marshall}
for exact real arithmetic with versions of these constructs and in
\refsection{applications} present two example programs that make critical
use of those constructs.

\commentt{
 We present a semantics for the construct that enables it to
+operate in either of is partial and/or nondeterministic configurations, thereby
+enabling it to be instantiated in a variety of different programming
+environments.
}

With these additional tools, constructive topology could conceivably be used
for applications such as formally verified cyber-physical systems or 
model checking of continuous systems.

\commentt{

 We implemented versions of these
constructs in the Marshall language for exact real arithmetic, and in
\refsection{applications} describe two example applications that make critical
use of those constructs:

\begin{itemize}
\item \emph{Approximate root finding:} Suppose we have an arbitrary continuous function $f : K \cto \R$, where $K$ is a compact-overt space (such as the unit interval $[0, 1]$), and want to determine if it has any roots. In general, it is not decidable whether $f$ has any roots, but an approximation is possible: fix a tolerance $\varepsilon > 0$. Then at least one of the following statements must hold:
\begin{itemize}
\item There is some $x \in K$ such that $|f(x)| < \varepsilon$.
\item For every $x \in K$, $f(x) \neq 0$.
\end{itemize}
Though each of the above statements is in general undecidable, we use the constructs devised in this paper to define a function that nondeterministically verifies that one of the above two statements holds, in the former case \emph{computing} such an $x$ that is ``almost a root.''

\item \emph{Car approaching a yellow light:}
Consider an autonomous car that is approaching a traffic light that has just turned yellow. To ensure safety, the car must be outside of the intersection when the light turns red. We model the problem with additional concrete detail, show that it is impossible to solve deterministically, and derive a nondeterministic program that is proven safe.
\end{itemize}
}

\begin{shortversion}
A longer version of this paper \cite{arXiv} provides further technical detail and proofs.
\end{shortversion}

\section{Constructive topology}

This section reviews locale theory, a constructive theory of topology that provides a semantic and computational foundation for programming with spaces. Readers interested in a more thorough introduction may wish to consult \citeauthor{topologyvialogic}'s \emph{Topology via Logic} \citeyear{topologyvialogic}.

\paragraph{Preliminaries.}
We intend mathematical statements to be interpreted within a constructive metatheory with a universe of impredicative propositions $\Prop$, potentially formalizable within, for instance, the Calculus of Constructions\footnote{
It is possible to formulate a predicative analogue of locale theory known as \emph{formal topology}, which makes more clear the computational content of constructive topology.
\grammarr{This} does impose some difficulties that require some changes. For instance,
the construction of product spaces is generally impredicative, but it is possible
to instead use \emph{inductively generated formal spaces} \cite{coquand2003}, which has products even in a predicative setting. All spaces used in this paper are inductively generated \cite{coquand2003, vickersmetric, vickersdoublepowerlocale, sublocFT, vickersconnected}.
Palmgren \citeyear{palmgren2003} offers a more careful treatment of predicativity and universes in formal topology.
}.
We use the term ``type'' to refer to a type and the term ``set'' to refer to what is often called a setoid or a Bishop set \cite{bishop}: a type together with a distinguished equivalence relation on it, which we denote by $=$. 
If $A$ and $B$ are both sets, then the notation $f : A \to B$ means that $f$ is a morphism of sets (i.e., it maps equivalent elements of $A$ to equivalent elements of $B$).
For objects $A, B$ of a category, let the notation $A \cong B$ indicate that they are isomorphic.

\begin{definition}
A \emph{space}\footnote{
Since all topological notions in this article are pointfree, we coopt terminology from classical topology without fear of confusion. For instance, we say ``space'' rather than ``locale'' when describing the pointfree analogue of spaces.
} $A$ is a distributive lattice $\Open{A}$ that has top and bottom elements, $\top$ and $\bot$, respectively, and that has all joins such that binary meets distribute over all joins:
\[
U \wedge \bigvee_{i : I} V_i = \bigvee_{i : I} U \wedge V_i.
\]
\end{definition}

We call the lattice $\Open{A}$ the \emph{opens} of $A$. This lattice describes the observable or ``affirmable'' properties of $A$ \cite{topologyvialogic}. If $U \le \bigvee_{i : I} V_i$, we call the family $(V_i)_{i : I}$ an \emph{open cover} of $U$.

\begin{definition}
\label{def:point}
A \emph{point} $x$ of a space $A$ is a subset $(x \models \cdot) : \Open{A} \to \Prop$ (read ``$x$ lies in'') such that
\begin{mathpar}
\inferrule*[right=join]
  {x \models U \\ U \le \bigvee_{i : I} V_i}
  {\exists i : I.\ x \models V_i}

\inferrule*[right=meet-0]
  { }
  {x \models \top}

\inferrule*[right=meet-2]
  {x \models U \\ x \models V}
  {x \models U \wedge V}.
\end{mathpar}
\end{definition}

The formal proof that a point satisfies the above three rules both justifies the consistency of its definition and provides its computational content.
Intuitively, $x \models U$ means we have some knowledge $U$ about $x$. \irule{join} says that it is possible to refine existing knowledge about $x$ to get an even sharper estimate of where $x$ lies. When a point $x$ that lies in $U$ is presented with an open cover $U \le \bigvee_{i : I} V_i$, it uses the \emph{proof} of the covering relationship to \emph{compute} some open $V_i$ that $x$ also lies in. The index $i$ is a concrete answer that indicates where the point lies. 
\irule{meet-0} says that we know \emph{something} about $x$ (which we can then refine with \irule{join}), and \irule{meet-2} says that we can assimilate two pieces of knowledge about $x$ into one, which assures that they are mutually consistent.

\begin{definition}
\label{def:cmap}
A \emph{continuous map} $f : A \cto B$ between spaces is a map $\iimg{f} : \Open{B} \to \Open{A}$, called an \emph{inverse image map}, that preserves all joins, $\top$, and binary meets, i.e., it satisfies
\begin{mathpar}
\inferrule*[right=join]
  {U \le \bigvee_{i : I} V_i}
  {f^*(U) \le \bigvee_{i : I} f^*(V_i)}

\inferrule*[right=meet-0]
  { }
  {\top \le f^*(\top)}

\inferrule*[right=meet-2]
  { }
  {f^*(U) \wedge f^*(V) \le f^*(U \wedge V)}.
\end{mathpar}
\end{definition}

A continuous map $f : A \cto B$ transforms covers on $B$ into covers on $A$. Spaces and continuous maps form a cartesian monoidal category we call \cat{FSpc} (for \emph{formal spaces})\footnote{
\cat{FSpc} is often called \cat{Loc}, for \emph{locales}.}.
The terminal object is the one-point space $\One$, whose lattice of opens $\Open{\One}$ is $\Omega$, where $U \le V$ if $U$ implies $V$. Points of a space $A$ can be identified with continuous maps $\One \cto A$, and in particular the \irule{join} and \irule{meet} rules for continuous maps reduce to the corresponding rules for points.
Two continuous maps are equal if they have the same inverse image maps. One can think of the inverse image map as a behavioral \emph{specification} and the formal proof that the continuous map preserves meets and finitary joins as an \emph{implementation} of that specification.

Given a space $A$ and an open $U : \Open{A}$, we can form the open subspace $\{ A \mid U \}$ of $A$ by making $\Open{\{A \mid U \}}$ a quotient of $\Open{A}$, identifying opens $P, Q : \Open{A}$ in $\{ A \mid U \}$ when $P \wedge U = Q \wedge U$.

\section{Decision making with partiality and nondeterminism}
\label{parnondet}

The real line $\R$ is \emph{connected}, meaning that any continuous map $f : \R
\cto D$ to a discrete set $D$ must be a constant map. In particular, every map
$f : \R \cto \bool$ is constant.  The practical implications of connectedness
are severe: it is impossible to (continuously) make (nontrivial) discrete
decisions over variables that come from connected spaces such as $\R$.

\begin{proposition}
Continuous maps $f : A \cto \bool$ are in bijective correspondence with pairs of opens $(P, Q)$ of $A$
that are \emph{covering}, i.e., $\top \le P \vee Q$, and \emph{disjoint}, i.e., $P \wedge Q \le \bot$.
\label{bool-pairs}
\end{proposition}
\begin{proof}[Proof sketch]
Since $f^*$ preserves joins, it is specified entirely by its behavior on the two basic opens, $P \triangleq f^*(\cdot = \btrue)$ and $Q \triangleq f^*(\cdot = \bfalse)$. Since $f^*$ preserves $\top$ (\irule{meet-0}), $P$ and $Q$ are covering, and since $f^*$ preserves binary meets (\irule{meet-2}), $P$ and $Q$ are disjoint.
\end{proof}

While it is impossible to make discrete decisions on connected spaces $A$ that are total and
deterministic, we \emph{can} make decisions that are either partial (only
defined on some open subspace of the input space) or nondeterministic (could
potentially give different answers even when given the exact same input).
Partiality relaxes the requirement that the inverse image map preserves $\top$ (\irule{meet-0}),
while nondeterminism relaxes the requirement that the inverse image map preserves binary meets (\irule{meet-2}).
Accordingly, partial $\bool$-valued maps correspond to pairs of opens that are not necessarily covering, and nondeterministic $\bool$-valued maps correspond to pairs of opens that are not necessarily disjoint.

In this section, we present categories whose objects are spaces and whose
morphisms are like continuous maps, but the inverse image maps need not
necessarily preserve $\top$ or binary meets (see Fig. \ref{fig:parnondet-diamond}).
The remainder of this
section characterizes these partial and/or nondeterministic maps and the monads
that represent them.

\subsection{Partiality}

The \irule{meet-0} rule enforces totality: viewing $\top$ as the predicate
representing the entire space, \irule{meet-0} says that a point must lie in the
entire space.  Eliminating \irule{meet-0} permits definition of a continuous
map that is only defined on an open subspace of the domain. 

\begin{definition}
A \emph{partial map} $f$ from $A$ to $B$, written $f : A \pto B$, is a map $f^* : \Open{B} \to \Open{A}$ that preserves joins and binary meets but not necessarily $\top$.
These maps form a category $\cat{FSpc}_p$.
\end{definition}

\begin{example}
Consider the task of approximately comparing a real number with 0. We can define a partial comparison
$
\mathsf{cmp} : \R \pto \bool
$
by only defining a continuous map on the open subspace $\{ \R \mid \cdot \neq 0 \}$ of $\R$. We specify its observable behavior with its inverse image map
\begin{align*}
\mathsf{cmp}^*(\cdot = \btrue) &\triangleq \cdot > 0
& \mathsf{cmp}^*(\cdot = \bfalse) &\triangleq \cdot < 0.
\end{align*}

The inverse image map $\mathsf{cmp}^*$ in fact defines a partial map, as $\mathsf{cmp}^*$ preserves joins and binary meets, but it is not total, since it fails to preserve $\top$.
\end{example}
\begin{proof}
To confirm that $\mathsf{cmp}^*$ preserves binary meets, it suffices to check binary meets of distinct basic opens, so we confirm
\begin{align*}
\mathsf{cmp}^*((\cdot = \btrue) \wedge (\cdot = \bfalse))
&= \mathsf{cmp}^*(\bot) 
\\ &= \bot
  = (\cdot > 0) \wedge (\cdot < 0)
\\ &= \mathsf{cmp}^*(\cdot = \btrue) \wedge \mathsf{cmp}^*(\cdot = \bfalse).
\end{align*}

However, $\mathsf{cmp}^*$ does not preserve $\top$, since
$
\top \nleq \mathsf{cmp}^*(\top) = (\cdot < 0) \vee (\cdot > 0).
$
\end{proof}

There is a bijective correspondence between partial maps and continuous maps defined on some open subspace of the domain.
\commentt{
As a result, there is a bijective correspondence between partial values of $A$, $\One \pto A$, and points of $A$ ``guarded'' by some arbitrary proposition $P$, $\sum_{P : \Prop} (P \to (\One \cto A))$ (using the correspondence $\Open{\One} \cong \Prop$). \grammarr{This} tells us that proving that a partial point is actually total could require proving an arbitrary proposition. Assuming classical logic, the proposition $P$ is either true, in which case we indeed have a point of $A$, or it is false, in which case no point of $A$ was defined.
}

\subsection{Nondeterminism}

The \irule{meet-2} rule enforces determinism.  Spatially, the rule says that if
a point lies in two opens, it must lie in their intersection. Computationally,
it says that it should be possible to consistently reconcile different answers
given by different refinements computed by use of the \irule{join} rule.
Eliminating \irule{meet-2} allows the definition of programs whose observable
behavior might depend on the exact implementation of their inputs
(specifically, the formal proofs that their inputs preserve joins and finitary
meets). Rather than viewing such behavior as breaking the abstraction provided
by the equivalence relation on points (since points that lie in the same opens
may be treated differently), we can instead choose to maintain this abstraction
and view such behavior as fundamentally nondeterministic.

\begin{definition}
A \emph{nondeterministic map} $f$ from $A$ to $B$, written $f : A \ndto B$, is a map $f^* : \Open{B} \to \Open{A}$ that preserves joins and $\top$ but not necessarily binary meets.
These maps form a category $\cat{FSpc}_{nd}$.
\end{definition}

For instance, we can perform a nondeterministic approximate comparison of a real number with 0.
\begin{example}
\label{ex:cmp_nd}
Fix some error tolerance parameter $\varepsilon > 0$. We may define a total but nondeterministic approximate comparison with 0,
$
\mathsf{cmp} : \R \ndto \bool,
$
allowing error up to $\varepsilon$, by specifying its observable behavior with the inverse image map
\begin{align*}
\mathsf{cmp}^*(\cdot = \btrue) &\triangleq \cdot > - \varepsilon
& \mathsf{cmp}^*(\cdot = \bfalse) &\triangleq \cdot < \varepsilon.
\end{align*}

We can confirm that $\mathsf{cmp}^*$ in fact defines a nondeterministic map, as it preserves joins and $\top$ but fails to preserve binary meets.
\end{example}
\begin{proof}
Since $\mathsf{cmp}$'s codomain is discrete, it trivially satisfies \irule{join}.
We confirm it preserves $\top$:
\begin{align*}
\mathsf{cmp}^*(\top)
&= \mathsf{cmp}^*(\cdot = \btrue) \vee \mathsf{cmp}^*(\cdot = \bfalse) 
  = (\cdot > - \varepsilon) \vee (\cdot < \varepsilon) 
  = \top.
\end{align*}
However, it fails to preserve binary meets, since
\[
\mathsf{cmp}^*((\cdot = \btrue) \wedge (\cdot = \bfalse))
= \mathsf{cmp}^*(\bot) = \bot
\]
but
\begin{align*}
\mathsf{cmp}^*(\cdot = \btrue) \wedge \mathsf{cmp}^*(\cdot = \bfalse)
&= (\cdot > - \varepsilon) \wedge (\cdot < \varepsilon)
  = - \varepsilon < \cdot < \varepsilon,
\end{align*}
which is not $\bot$.
\end{proof}

\subsection{Both partiality and nondeterminism}

\begin{definition}
A \emph{nondeterministic and partial map} $f$ from $A$ to $B$, written $f : A \ndpto B$, is a map $f^* : \Open{B} \to \Open{A}$ that preserves joins but not necessarily any meets.
These maps form a category $\cat{FSpc}_{nd, p}$ (equivalent to the category of suplattices).
\end{definition}

\begin{figure}[!h]
\begin{equation*}
\begin{tikzcd}[column sep=tiny]
& \cat{FSpc}
   \arrow[dl, tail, "\text{forget \irule{meet-2}}"']
   \arrow[dr, tail, "\text{forget \irule{meet-0}}"]
&
\\ \cat{FSpc}_{nd}
    \arrow[dr, tail, "\text{forget \irule{meet-0}}"']
& & \cat{FSpc}_p
       \arrow[dl, tail, "\text{forget \irule{meet-2}}"]
\\ & \cat{FSpc}_{nd, p} &
\end{tikzcd}
\end{equation*}
\caption{The lattice of categories representing potentially nondeterministic ($nd$) or partial ($p$) maps on spaces.
}
\label{fig:parnondet-diamond}
\end{figure}
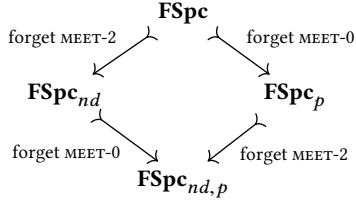

\subsection{Monads and summary}

Potentially allowing partiality and nondeterminism yields
a lattice of categories that represent nondeterministic and partial maps,
depicted in Fig. \ref{fig:parnondet-diamond}, where each arrow denotes a
faithful (``forgetful'') functor where a particular rule is no longer required
for inverse image maps. These forgetful functors have right adjoints, such that
they induce a family of (strong) monads on $\cat{FSpc}$:
$\cdot_\bot$ for representing partiality, 
$\PLowerP$ for nondeterminism,
and $\PLower$ for both\footnote{ Each of these strong monads preserves
inductive generation of spaces \cite{topologyvialogic,
vickersdoublepowerlocale}.  }.
Their adjunctions give the correspondences
\begin{align*}
\bijectivecorrespondence{A \ndpto B}{A \cto \PLower(B)}
\\ \bijectivecorrespondence{A \ndto B}{A \cto \PLowerP(B)}
\\ \bijectivecorrespondence{A \pto B}{A \cto B_\bot}.
\end{align*}

Accordingly, it is possible to use these monads to have access to partiality and/or nondeterminism within the language of continuous maps.

\commentt{
\subsection{Catch all (Delete)}

To understand how partiality and nondeterminism arise by forgetting
\irule{meet-0} and \irule{meet-2}, respectively, it is instructive to observe
how the categories of partial and nondeterministic values restrict to the
discrete spaces. For discrete spaces, an inverse image map $f^* :
\Open{\mathsf{Discrete}(B)} \to \Open{\mathsf{Discrete}(A)}$ corresponds to a
relation $R \subset A \times B$. \irule{join} says only that $R$ must respect
the equivalence relations on $A$ and $B$. \irule{meet-0} says that the $R$ must
be left-total, and \irule{meet-2} says that $R$ must be functional.
Accordingly, $\cat{FSpc}_{nd, p}$ restricted to discrete spaces corresponds to
relations on sets, $\cat{FSpc}_{nd}$ to multi-valued functions (i.e.,
left-total relations), and $\cat{FSpc}_{p}$ to partial functions (i.e.,
functional relations).
}

\section{Pattern matching}
\label{patterns}

\commentt{
\todo{Fix this paragraph}
Our development thus far demonstrates a direct extension of existing
programming models to programming on $\R$ (as is the goal of languages and
systems such as Rosa \cite{soundreals} and Marshall \cite{marshall}) and -- more generally -- programming
on spaces, that would enable developers to construct computations from $\R$ to
$\bool$. Specifically, by the observation that each such computation is
specified by a pair of opens $(P, Q)$ that each correspond to $\btrue$ and
$\bfalse$, respectively (Proposition~\ref{bool-pairs}), we could extend a
given language with syntax that enables a developer to specify their opens of
interest. For example, a programmer could specify our running example of
comparison in an extended language as follows:

\begin{verbatim}
let cmp = fun r : R => (r > -e, r < e);; 
\end{verbatim} where $P \triangleq r > -e$ and $Q \triangleq r < e$.
}

Often, programmers would like to compose a decision
on $\R$ with other computations that depend on the decision and
thus associate each condition with a corresponding computation. This
programming pattern resembles pattern matching in traditional functional
programming, and therefore, in this section, we identify and present general
pattern matching for spaces.  Our constructions admit partial and/or
nondeterministic pattern matches, where in the former case the collection of
patterns may not be exhaustive, and in the latter they may overlap. While the
syntax of a pattern match determines a unique map that is potentially partial
and nondeterministic, there are simple conditions that ensure that a map is
total or deterministic:
\begin{enumerate}
\item \emph{Totality}: together, the cases cover the entire input space.
\item \emph{Determinism}: patterns are disjoint and injective.
\end{enumerate}

\subsection{Pattern families}
\label{s:patternfamilies}

This section characterizes those families of patterns that may be used to match on a scrutinee that comes from a space $A$ (if nothing is to be assumed about the branches). The idea is that we compose a function $f : A \xto{\mathcal{C}} B$ by factoring through a disjoint sum over a collection of spaces representing the possible patterns and branches, $\sum_{i : I} U_i$, i.e., a composition
\begin{equation*}
\begin{tikzcd}
A
\arrow[r, "\mathsf{inv}"]
& \sum_{i : I} U_i
\arrow[r, "e"]
& B
\end{tikzcd}
\end{equation*}
of a ``pattern matching'' part $\mathsf{inv}$ followed by the ``branch execution'' part $e$. The collection of branches $(e_i : U_i \xto{\mathcal{C}} B)_{i : I}$ exactly correspond to the branch-execution function $e$, but the pattern-matching part $\mathsf{inv}$ is more interesting; this section will address those families of patterns that may yield valid functions of this sort.

Semantically, we think of a single pattern as representing a space $U$ together with a map $p : U \cto A$ that represents the possibility that the scrutinee can be represented as a point in the image of $p$. For a single pattern $p : U \cto A$ to be implementable, it must have a well-behaved inverse $p^{-1} : A \ndpto U$ that is partial and also may be nondeterministic.
If we are building a program that is deterministic, then $p^{-1}$ should be deterministic.
If we are building a program that is total, then we do not need each $p^{-1}$ to be total, but we do need the collection of them to cover $A$. 
We will find that \emph{open maps} are exactly those with well-behaved (nondeterministic and partial) inverses, and \emph{open embeddings} are open maps whose inverses are deterministic.

\subsection{Open maps and open embeddings}
\label{openmaps}

In pattern matching for functional programming, one may pattern match on an inductive type by checking whether it has the form of a particular constructor applied to some argument (i.e., it is in the image of the map defined by a particular constructor). The analogues of constructors for pattern matching on spaces are the open maps and the open embeddings.

\subsubsection{Open maps}

\begin{definition}[\citealt{elephant}]
A continuous map $f : A \cto B$ is an \emph{open map} (which we may denote by $f : A \to_o B$) if the inverse image map $f^* : \Open{B} \to \Open{A}$ has a left adjoint $f_! : \Open{A} \to \Open{B}$, called the \emph{direct image map}, that satisfies the Frobenius law,
$
f_!(U \wedge f^*(V)) = f_!(U) \wedge V.
$
\end{definition}

That is, $f$ is an open map if the \emph{image} of any open in $A$ is open in $B$; the direct image map provides this mapping.

\begin{longversion}
\begin{example}
Identity maps $\mathsf{id} : A \cto A$ are open maps, and the composition of open maps is an open map.
\end{example}
\begin{proof}
We have $\mathsf{id}_!(U) = U$, which clearly satisfies $\mathsf{id}_! \dashv \mathsf{id}^*$ since they are all identity maps. It is also immediate that the Frobenius law holds for $\mathsf{id}_!$. Given open maps $f : A \cto B$ and $g : B \cto C$, we claim that the direct image map is $(g \circ f)_! = g_! \circ f_!$. We first confirm the adjunction:
\begin{align*}
(g \circ f)_!(U) \le V
&\ifandonlyif g_! (f_!(U)) \le V
\\ &\ifandonlyif f_!(U) \le g^*(V)  \tag{$g_! \dashv g^*$}
\\ &\ifandonlyif U \le f^*(g^*(V))  \tag{$f_! \dashv f^*$}
\\ &\ifandonlyif U \le (g \circ f)^*V.
\end{align*}
It only remains to confirm the Frobenius law for the composition $g_! \circ f_! : \Open{A} \to \Open{C}$:
\begin{align*}
(g \circ f)_!(U \wedge (g \circ f)^*(V))
  &= g_! (f_! (U \wedge f^*(g^*(V))))
\\ &= g_! (f_! (U) \wedge g^*(V)) \tag{Frobenius law for $f$}
\\ &= g_! (f_! (U)) \wedge V 
\\ &= (g \circ f)_! (U) \wedge V  \tag{Frobenius law for $g$}.
\end{align*}
\end{proof}

\begin{proposition}
In the pullback square where $p$ and $q$ are open maps,
\begin{equation*}
\begin{tikzcd}
\\A \times_X B \arrow[r, "\theta"]
   \arrow[d, "\varphi"]
   \arrow[dr, phantom, "\lrcorner", very near start]
  & A \arrow[d, "p"]
\\ B \arrow[r, "q"]
& X
\end{tikzcd},
\end{equation*}
define $f : A \times_X B \cto X$ by $f = p \circ \theta = q \circ \varphi$. Then
\[
f_!(\top) = p_!(\top) \wedge q_!(\top).
\]
\end{proposition}
\begin{proof}
Note that for any $U : \Open{B}$, $p^*(q_!(U)) = \theta_!(\varphi^*(U))$ (proof in section 5.2 of \cite{pedicchio}). Then
\begin{align*}
p_!(\top) \wedge q_!(\top)
   &= p_!(\top \wedge p^*(q_!(\top))) \tag{Frobenius reciprocity}
\\ &= p_!(p^*(q_!(\top)))
\\ &= p_!(\theta_!(\varphi^*(\top)))
\tag{the earlier equality}
\\ &= p_!(\theta_!(\top)) \tag{$\varphi^*$ preserves $\top$}
\\ &= f_!(\top).
\end{align*}
\end{proof}

\begin{proposition}
Parallel composition of open maps yields an open map. That is, given $f : A \cto B$ and $g : X \cto Y$ open,
$f \parmap g : A \times X \cto B \times Y$ is open.
\end{proposition}
\begin{proof}
We claim that the direct image map operates on basic opens (which are open rectangles of $A \times X$) by $(f \parmap g)_!(a \times x) = f_!(a) \times g_!(x)$ (this extends to all opens by taking joins). We confirm the adjunction (using the fact that every open is a join of basic opens):
\begin{align*}
&(f \parmap g)_!\left(\bigvee_{i : I} a_i \times x_i \right) \le \bigvee_{j : J} (b_j \times y_j)
\\ &\ifandonlyif \bigvee_{i : I} f_!(a_i) \times g_!(x_i) \le \bigvee_{j : J} (b_j \times y_j)
\\ &\ifandonlyif \forall i : I.\ \exists j : J.\ f_!(a_i) \times g_!(x_i) \le b_j \times y_j
\\ &\ifandonlyif \forall i : I.\ \exists j : J.\ (f_!(a_i) \le b_j) \text{ and } g_!(x_i) \le y_j
\\ &\ifandonlyif \forall i : I.\ \exists j : J.\ (a_i \le f^*(b_j)) \text{ and } x_i \le g^*(y_j)
\\ &\ifandonlyif \forall i : I.\ \exists j : J.\ a_i \times x_i \le f^*(b_j) \times g^*(y_j)
\\ &\ifandonlyif \bigvee_{i : I} a_i \times x_i \le \bigvee_{j : J} f^*(b_j) \times g^*(y_j)
\\ &\ifandonlyif \bigvee_{i : I} a_i \times x_i \le (f \parmap g)^* \left( \bigvee_{j : J} b_j \times y_j \right).
\end{align*}

It suffices to confirm the Frobenius law holds on basic opens \cite{spitterslocatedness}:
\begin{align*}
&(f \parmap g)_!((a \times x) \wedge (f \parmap g)^*(b \times y))
\\ &= (f \parmap g)_!((a \times x) \wedge (f^*(b) \times g^*(y)))
\\ &= (f \parmap g)_!((a \wedge f^*(b)) \times (x \wedge g^*(y))
\\ &= f_!(a \wedge f^*(b)) \times g_!(x \wedge g^*(y))
\\ &= (f_!(a) \wedge b) \times (g_!(x) \wedge y) \tag{Frobenius law for $f$ and $g$}
\\ &= (f_!(a) \times g_!(x)) \wedge (b \times y)
\\ &= (f \parmap g)_!(a \times x) \wedge (b \times y).
\end{align*}
\end{proof}

We will now describe some facts that relate the open maps to nondeterministic maps.
\end{longversion}

\begin{proposition}
For any open map $p : A \cto B$, there is a (potentially) partial and nondeterministic inverse map $p^{-1} : B \ndpto A$ whose inverse image map is $p_!$.
\end{proposition}
\begin{proof}
We only must prove that $p_!$ preserves joins: it does, since $p_!$ is a left adjoint (to $p^*$).
\end{proof}

An example of an open map is the ``return'' function of the nondeterminism monad $\{ \cdot \} : A \cto \PLower^+(A)$.
\begin{longversion}
Its direct image map takes opens $U$ of $A$ to $\lozenge U$ in $\PLower(A)$.
\end{longversion}

\subsubsection{Open embeddings}
\label{s:open-embedding}

Given an open $U$ of a space $A$, let $\iota[U] : \{A \mid U \} \cto A$ denote the inclusion of the open subspace $\{ A \mid U \}$ into $A$.

\begin{lemma}
\label{open-factor}
An open map $f : A \cto B$ factors through its direct image $f_!(\top)$, i.e., there is an $\tilde{f}$ such that the following diagram commutes:
\begin{equation*}
\begin{tikzcd}
A \arrow[r, "\tilde{f}"]
   \arrow[dr,swap,"f", shift right, shift right]
& \{ B \mid f_!(\top) \}
   \arrow[d, "\oinclf{f_!(\top)}"]
\\
{} & B
\end{tikzcd}
\end{equation*}
\end{lemma}
\begin{longversion}
\begin{proof}
This statement is equivalent to that for all $U : \Open{B}$, $f^*(U) \le f^*(U \wedge f_!(\top))$. This is indeed the case:
\begin{align*}
 f^*(U \wedge f_!(\top))
   &= f^*(U) \wedge f^*(f_!(\top)) \tag{$f^*$ preserves meets}
 \\ &\ge f^*(U) \wedge \top \tag{$f^* \circ f_!$ inflationary}
 \\ &= f^*(U).
\end{align*}
\end{proof}
\end{longversion}

\begin{definition}
A map $f : A \cto B$ is an \emph{open embedding} (or \emph{open inclusion}), denoted $f : A \hookto B$, if $A$ is isomorphic to its image under $f$ in $B$, i.e., if there is an open $U : \Open{B}$ and isomorphism $\tilde{f} : A \cto \{B \mid U \}$ such that the following diagram commutes:
\begin{equation*}
\begin{tikzcd}
A \arrow[r, "\tilde{f}", shift left]
   \arrow[r, <-, shift right, swap, "\tilde{f}^{-1}"]
   \arrow[dr,swap,"f", shift right, shift right]
& \{ B \mid U \}
   \arrow[d, "\oinclf{U}"]
\\
{} & B
\end{tikzcd}
\end{equation*}
\end{definition}

\begin{theorem}
A map $f : A \cto B$ is an open embedding if and only if it is an open map and its direct image map $f_!$ preserves binary meets.
\end{theorem}
\begin{shortversion}
\begin{proof}[Proof sketch]
Given an open embedding $f : A \cto B$ that factors through $\{ B \mid U\}$, letting $\tilde{f}$ and $\tilde{f}^{-1}$ be the maps as in the above diagram, its direct image map is given by
\begin{align*}
f_! &: \Open{A} \to \Open{B}
\\ f_!(V) &\triangleq \tilde{f}^{-1*}(V) \wedge U.
\end{align*}

Conversely, given an open map $f : A \cto B$ with a meet-preserving direct image map $f_!$, we claim that $A \iso \{ B \mid f_!(\top) \}$.
\end{proof}
\end{shortversion}
\begin{longversion}
\begin{proof}
Suppose $f : A \cto B$ is an open embedding that factors through $\{ B \mid U\}$, and let $\tilde{f}$ and $\tilde{f}^{-1}$ be the maps as in the above diagram. Then its direct image map is given by
\begin{align*}
f_! &: \Open{A} \to \Open{B}
\\ f_!(V) &\triangleq \tilde{f}^{-1*}(V) \wedge U.
\end{align*}
We now confirm that $f_! \dashv f^*$. We have
\begin{align*}
f_!(V) \le W  
&\ifandonlyif  \tilde{f}^{-1*}(V) \wedge U \le W  
\\ &\ifandonlyif  \tilde{f}^{-1*}(V) \wedge \tilde{f}^{-1*}(\tilde{f}^*(U)) \le W
\\ &\ifandonlyif  \tilde{f}^{-1*}(V \wedge \tilde{f}^*(U)) \le W
\\ &\ifandonlyif  \tilde{f}^{-1*}(V \wedge \top) \le W
\\ &\ifandonlyif  \tilde{f}^{-1*}(V) \le W
\\ &\ifandonlyif V \le \tilde{f}^*(W)
\end{align*}

Moreover, $f_!$ preserves meets:
\begin{align*}
f_!(V \wedge W) 
  &= \tilde{f}^{-1*}(V \wedge W) \wedge U
\\ &= \tilde{f}^{-1*}(V) \wedge \tilde{f}^{-1*}(W) \wedge U
\\ &= (\tilde{f}^{-1*}(V) \wedge U) \wedge (\tilde{f}^{-1*}(W) \wedge U)
\\ &= f_!(V) \wedge f_!(W).
\end{align*}

We also confirm $f_!$ satisfies the Frobenius law:
\begin{align*}
f_!(V \wedge f^*(W)) 
&= f_!(V) \wedge f_!(f^*(W)) \tag{$f_!$ preserves meets}
\\ &=  f_!(V) \wedge \tilde{f}^{-1*}(f^*(W)) \wedge U
\\ &=  f_!(V) \wedge \tilde{f}^{-1*}(f^*(W)) \tag{$f_!(V)$ already at most $U$}
\\ &=  f_!(V) \wedge \tilde{f}^{-1*}(\tilde{f}^*(W \wedge U))
\\ &=  f_!(V) \wedge W \wedge U
\\ &=  f_!(V) \wedge W \tag{$f_!(V)$ already at most $U$}
\end{align*}

Thus an open embedding is an open map with a meet-preserving direct image map. We now prove the converse. Given an open map $f : A \cto B$ with a meet-preserving direct image map $f_!$, we claim that $A \iso \{ B \mid f_!(\top) \}$. By Proposition \ref{open-factor}, we already have a continuous map $\tilde{f} : A \cto \{ B \mid f_!(\top) \}$. We define $g : \{ B \mid f_!(\top) \} \cto A$ by
\begin{align*}
g^* &:  \Open{A} \to \Open{\{ B \mid f_!(\top) \}}
\\ g^*(V) &\triangleq f_!(V) \wedge f_!(\top).
\end{align*}
We claim $g^*$ indeed defines a continuous map. It preserves joins and binary meets since it is the composition of $f_!$ and $\cdot \wedge f_!(\top)$, both of which preserve joins and binary meets, so it suffices to show that $g^*(\top) = \top$, which is indeed the case, as
\[
g^*(\top) = f_!(\top) \wedge f_!(\top) = f_!(\top)
\]
which is equivalent to $\top$ in $\{ B \mid f_!(\top) \}$.
\end{proof}
\end{longversion}

\begin{longversion}
\begin{example}
Identity maps $\mathsf{id} : A \hookto A$ are open embeddings, and the composition of open embeddings is an open embedding.
\end{example}
\begin{proof}
We have $\mathsf{id}_!(U) = U$, which preserves binary meets. Given open maps $f : A \cto B$ and $g : B \cto C$, the composition $g_! \circ f_! : \Open{A} \to \Open{C}$ preserves binary meets since $f_!$ and $g_!$ both do.
\end{proof}

\begin{proposition}
Parallel composition of open embeddings yields an open embedding. That is, given $f : A \hookto B$ and $g : X \hookto Y$ open embeddings,
$f \parmap g : A \times X \hookto B \times Y$ is an open embedding.
\end{proposition}
\begin{proof}
It suffices to confirm that the direct image preserves binary meets, and it suffices to check this by only checking the basic opens, which in this case are open rectangles:
\begin{align*}
&(f \parmap g)_! \left( (a_{1} \times x_{1}) \wedge (a_{2} \times x_{2}) \right)
\\ &=(f \parmap g)_!  \left( (a_{1} \wedge a_{2}) \times (x_{1} \wedge x_{2}) \right)
\\ &= f_!(a_{1} \wedge a_{2}) \times g_!(x_{1} \wedge x_{2})
\\ &= (f_!(a_{1}) \wedge f_!(a_{2})) \times (g_!(x_{1}) \wedge g_!(x_{2}))
  \tag{$f_!$ and $g_!$ preserve binary meets}
\\ &= (f_!(a_{1}) \times g_!(x_{1})) \wedge (f_!(a_{2}) \times g_!(x_{2}))
\\ &= (f \parmap g)_!(a_{1} \times x_{1}) \wedge (f \parmap g)_!(a_{2} \times x_{2})
\end{align*}
\end{proof}

\begin{proposition}
\label{pullbackopenembedding}
The pullback of an open embedding is an open embedding.
\end{proposition}
\begin{proof}
Since (up to homeomorphism) an open embedding is just the inclusion of an open subspace, it suffices to prove that the pullback of the inclusion of an open subspace is an open embedding. Given the open inclusion $\oinclf{U} : \{ X \mid U \} \hookto X$, for any map $f : A \cto X$ we have the pullback square
\begin{equation*}
\begin{tikzcd}
\\\{ A \mid f^*(U) \} \arrow[r, hook, "\oinclf{f^*(U)}"]
   \arrow[d, "f_{| U}"]
   \arrow[dr, phantom, "\lrcorner", very near start]
  & A \arrow[d, "f"]
\\ \{ X \mid U \} \arrow[r, hook, "\oinclf{U}"]
& X
\end{tikzcd},
\end{equation*}
which expresses the fact that the preimage of opens are open. It is possible to confirm that the diagram commutes and that $\{ A \mid f^*(U) \}$ satisfies the universal property of pullbacks.
\end{proof}

The open embeddings are closely related to the partial maps:
\end{longversion}
\begin{proposition}
Given an open embedding $f : A \hookto B$, the ``inverse'' map $f^{-1} : B \ndpto A$ that any open map has is in fact deterministic, i.e., $f^{-1} : B \pto A$.
\end{proposition}
\begin{proof}
Follows directly from the fact that $f_!$ preserves joins and binary meets.
\end{proof}

An example of an open embedding is the ``return`` function of the partiality monad $\up : A \cto \Lifted{A}$. \grammarr{This} allows us to view $A$ as an open subspace of $\Lifted{A}$.

\subsection{Pattern families: definition and properties}

In general, we have an entire family of patterns $(p_i : U_i \to_o A)_{i : I}$ where $I$ is some index type. We can use this pattern family to construct the partial and nondeterministic inverse
\begin{align*}
\mathsf{inv} &: A \ndpto \sum_{i : I} U_i
\\ \mathsf{inv}(x) &\triangleq \bigsqcup_{i : I} \mathsf{inj}_i (p_i^{-1}(x)),
\end{align*}
where $\sqcup$ denotes the nondeterministic join, $\sqcup : \prod_{i : I} X \ndpto X$ for any space $X$.

Therefore, for any index type $I$, any collection of open maps $p_i : U_i \to_o A$ is a collection of patterns for defining a nondeterministic and partial map using a pattern match. Given an arbitrary collection of branches $e_i : U_i \ndpto B$, or equivalently, $e : \sum_{i : I} U_i \ndpto B$, the pattern match is just the composition
$e \circ \mathsf{inv} : A \ndpto B$.
As expected, the pattern match is a nondeterministic union of its branches.

For those categories/languages that require totality or determinism, we would like to characterize the families of patterns that are suitable in those cases. These will be subcollections of the collection of pattern families in the nondeterministic and partial case.

\begin{definition}
A \emph{pattern family} for a subcategory $\mathcal{C}$ of $\cat{FSpc}_{nd, p}$ is a family of open maps $\left(p_i : U_i \to_o A\right)_{i : I}$ such that $\bigsqcup_{i : I} \mathsf{inj}_i \circ p_i^{-1} : A \ndpto \sum_{i : I} U_i$ is in $\mathcal{C}$. Let $\mathcal{J}_\mathcal{C}$ denote the collection of pattern families for $\mathcal{C}$.
\end{definition}

\begin{theorem}
We can characterize the pattern families for the following subcategories of $\cat{FSpc}_{nd, p}$ as exactly those families of open maps $\left(p_i : U_i \to_o A\right)_{i : I}$ satisfying certain additional properties:
\begin{description}
\item[$\cat{FSpc}_{nd}$] (\emph{Totality}) The patterns cover the whole input space, i.e., 
$
\top \le \bigvee_{i : I} {p_i}_!(\top).
$
\item[$\cat{FSpc}_{p}$] (\emph{Determinism}) Each $p_i : U_i \to_o A$ is an open embedding, and the patterns are \emph{pairwise disjoint}, meaning that whenever ${p_i}_!(\top) \wedge {p_j}_!(\top)$ is positive\footnote{
An open $U$ is called \emph{positive} if whenever $U \le \bigvee_{i : I} V_i$, $I$ is inhabited,
i.e., every open cover of $U$ itself must have at least one open.
} in $A$, then (intensionally) $i = j$.
\item[$\cat{FSpc}$] (\emph{Totality} and \emph{determinism}) The above conditions for totality and determinism must both hold.
\end{description}
\end{theorem}
\begin{longversion}
\begin{proof}
\begin{description}
\item[$\cat{FSpc}_{nd}$] (\emph{Totality})
If we want to ensure that a pattern match is total, then it suffices to require that 
\[
\left(\bigsqcup_{i : I} \mathsf{inj}_i \circ p_i^{-1}\right)^*(\top) = \top,
\]
or equivalently,
\[
\top \le \bigvee_{i : I} {p_i}_!(\top),
\]
meaning that the patterns cover the whole input space.

\item[$\cat{FSpc}_{p}$] (\emph{Determinism}):
To ensure determinism, we require that $\bigsqcup_{i : I} \mathsf{inj}_i \circ p_i^{-1}$ preserves binary meets. We claim that this map preserves binary meets if and only if both of the following conditions hold:
\begin{enumerate}
\item Each pattern $p_i : U_i \to_o A$ is in fact an open embedding.
\item The patterns are pairwise disjoint.
\end{enumerate}

First, we prove the two conditions hold if $\mathsf{inv}^*$ preserves binary meets.
\begin{enumerate}
\item It suffices to show that each ${p_i}_!$ preserves binary meets, which follows from the calculation
\begin{align*}
{p_i}_!(a \wedge b) 
&= \mathsf{inv}^*((i, a) \wedge (i, b))
\\ &= \mathsf{inv}^*(i, a) \wedge \mathsf{inv}^*(i, b)
\\ &= {p_i}_!(a) \wedge {p_i}_!(b).
\end{align*}
\item We must show that if ${p_i}_!(\top) \wedge {p_j}_!(\top)$ is positive in $A$, then $i \inteq j$.
\begin{align*}
{p_i}_!(\top) \wedge {p_j}_!(\top)
&= \mathsf{inv}^*(i, \top) \wedge \mathsf{inv}^*(j, \top)
\\ &= \mathsf{inv}^*((i, \top) \wedge (j, \top)).
\end{align*}
Since $\mathsf{inv}^*$ preserves joins, by Proposition \ref{thm:join-pos} if ${p_i}_!(\top) \wedge {p_j}_!(\top)$ is positive, so is $(i, \top) \wedge (j, \top)$, which implies that $i \inteq j$.
\end{enumerate}
Now we prove the converse: if a family $p_i : U_i \hookto A$ of open embeddings is pairwise disjoint, its $\mathsf{inv}$ map is in fact deterministic.
It suffices to prove that $\mathsf{inv}^*$ preserves binary meets of basic opens. For any open $a : \Open{A}$ and proposition $Q$, let $\chi_Q(a) \triangleq \bigvee_{q : Q} a$. If $Q$ is true, then $\chi_Q(a) = a$ and if $Q$ is false, then $\chi_Q(a) = \bot$.
\begin{align*}
&\mathsf{inv}^*((k, a) \wedge (\ell, b))
\\ &= \bigvee_{i : I} {p_i}_!(\chi_{i \inteq k}(a) \wedge \chi_{i \inteq \ell}(b))
\\ &= \bigvee_{i : I} {p_i}_!(\chi_{i \inteq k}(a)) \wedge {p_i}_!(\chi_{i \inteq \ell}(b))
  \tag{${p_i}_!$ preserves meets}
\\ &= \bigvee_{i : I, j : I} {p_i}_!(\chi_{i \inteq k}(a)) \wedge {p_j}_!(\chi_{j \inteq \ell}(b))
      \tag{$p_i$s pairwise disjoint}
\\ &= \left(\bigvee_{i : I} {p_i}_!(\chi_{i \inteq k}(a)) \right)
      \wedge  \left(\bigvee_{j : I} {p_j}_!(\chi_{j \inteq \ell}(b)) \right)
\\ &= \mathsf{inv}^*(k, a)
      \wedge  \mathsf{inv}^*(\ell, b).
\end{align*}

\item[$\cat{FSpc}$] (\emph{Totality} and \emph{determinism}): Since $\cat{FSpc} = \cat{FSpc}_{nd} \cap \cat{FSpc}_{p}$, we just require that the conditions for totality and determinism must both hold.
\end{description}

\end{proof}
\end{longversion}

For any subcategory $\mathcal{C}$ of $\cat{FSpc}_{nd, p}$, one can construct a pattern match by composing a pattern family $\left(p_i : U_i \to_o A\right)_{i : I}$ for $\mathcal{C}$ with a collection of branches $e : \sum_{i : I} U_i \to_\mathcal{C} B$ that is in $\mathcal{C}$.

When determinism is required, we recover the familiar condition required of pattern matching in functional programming: disjointness of patterns (i.e., patterns are not allowed to overlap). When both determinism and totality are required, we further recover the familiar condition that pattern membership is decidable:

\begin{longversion}
\begin{definition}
An open $U$ is \emph{clopen} (for ``closed'' and ``open'') if it has a Boolean complement, i.e., there is another open $V$ such that $U \vee V = \top$ and $U \wedge V = \bot$.
\end{definition}
\end{longversion}

\begin{proposition}
For a pattern family $( p_i : U_i \hookto A )_{i : I}$ for \cat{FSpc}, if the index type $I$ has decidable equality, then for each $i : I$, there is a map $\chi_{{p_i}_!(\top)} : U_i \cto \bool$ satisfying $\chi_{{p_i}_!(\top)}^*(\cdot = \btrue) = {p_i}_!(\top)$ (i.e., ${p_i}_!(\top)$ is clopen).
\end{proposition}
\begin{longversion}
\begin{proof}
Fix some $i : I$. We claim that the Boolean complement of ${p_i}_!(\top)$ is $\bigvee_{j : I \mid j \not\inteq i} {p_j}_!(\top)$. Since $I$ has decidable equality, their join is $\top$:
\[
{p_i}_!(\top) \vee \bigvee_{j : I \mid j \not\inteq i} {p_j}_!(\top)
= \bigvee_{i : I} {p_i}_!(\top)
= \top.
\]
Pairwise disjointness implies that their meet is $\bot$:
\begin{align*}
{p_i}_!(\top) \wedge \bigvee_{j : I \mid j \not\inteq i} {p_j}_!(\top)
&= \bigvee_{j : I \mid j \not\inteq i} {p_i}_!(\top) \wedge {p_j}_!(\top)
= \bigvee_{j : I \mid j \not\inteq i} \bot
= \bot.
\end{align*}
\end{proof}
This means that if the index type $I$ has decidable equality, the predicate corresponding to any given pattern is decidable. This recovers the usual understanding of pattern matching in functional programming, where patterns are disjoint and correspond to decidable predicates.
\end{longversion}

We will now observe that the pattern families for the various subcategories of $\cat{FSpc}_{nd, p}$ form a lattice of Grothendieck pretopologies. This structure is useful: it tells us that there are certain techniques that we can always use to form pattern families, and that pattern families will have important structural properties. For instance, the transitivity axiom corresponds to the ability to flatten nested pattern matches into a single one.
The stability axiom allows us to use ``pulled-back covers'': it is possible to pattern match on an input $x : A$ by doing a case analysis on $f(x) : B$, such that in each branch it is known that $x$ lies in a particular open subspace of $A$ (rather than only knowing that $f(x)$ lies in a particular open subspace of $B$). The root-finding example in \refsection{s:bcov:root} uses a pulled-back cover in this way.

\begin{definition}[\citealt{sheaves}]
A \emph{Grothendieck pretopology} is an assignment to each space $A$ of a collection of families $(U_i \cto A)_{i : I}$ of continuous maps, called \emph{covering families}, such that
\begin{enumerate}
\item \emph{isomorphisms cover} -- every family consisting of a single isomorphism $
U \stackrel{\iso}{\cto}A$ is a covering family;
\item \emph{stability axiom} -- the collection of covering families is stable under pullback: if $(U_i \cto A )_{i : I}$ is a covering family and $f : V \cto A$ is any continuous map, then the family of pullbacks $(f^* U_i \cto V)_{i : I}$ is a covering family;
\item \emph{transitivity axiom} -- if $(U_i \cto A)_{i : I}$ is a covering family and for each $i$ also $(U_{i,j} \cto U_i)_{j : J_i}$ is a covering family, then also the family of composites $(U_{i,j} \cto U_i \cto A)_{i : I, j : J_i}$ is a covering family.
\end{enumerate}
\end{definition}

\begin{proposition}[Product axiom]
In any Grothendieck pretopology, given covering families 
$( p_i : U_i \cto A )_{i : I}$
and $( q_j : V_j \cto B )_{j : J}$,
there is a product covering family
$ ( p_i \parmap q_j : U_i \times V_j \cto A \times B )_{(i, j) : I \times J}$.
\end{proposition}
\begin{longversion}
\begin{proof}
First, we use pullback stability along $\mathsf{fst} : A \times B \cto A$ to produce the covering family
$( \mathsf{fst}^*p_i : U_i \times B \cto A \times B )_{i : I}$. Then, for each $i : I$, we use pullback stability along $\mathsf{snd} : U_i \times B \cto B$ to produce the covering family
$( \mathsf{snd}^*q_j: U_i \times V_j \cto U_i \times B )_{j : J}$.
Finally, by the transitivity axiom, we get the covering family
\[
( \mathsf{fst}^*p_i \circ \mathsf{snd}^*q_j : U_i \times V_j \cto A\times B )_{i :I, j : J}.
\]
Then one can confirm that
\[
\mathsf{fst}^*p_i \circ \mathsf{snd}^*q_j  = p_i \parmap q_j.
\]
\end{proof}
\end{longversion}

\begin{theorem}
For each $\mathcal{C} \in \{ \cat{FSpc}, \cat{FSpc}_{nd}, \cat{FSpc}_{p}, \cat{FSpc}_{nd, p} \}$, the collection of pattern families for $\mathcal{C}$, $\mathcal{J}_\mathcal{C}$, forms a Grothendieck pretopology.
\end{theorem}
\begin{longversion}
\begin{proof}
\begin{description}
\item[$\cat{FSpc}_{nd, p}$]

\begin{enumerate}
\item Homeomorphisms cover, since any homeomorphism is (an open embedding and thus) an open map.
\item Stability follows from the fact that the pullback of an open map against any continuous map is open (see Proposition C3.1.11 (i) in \cite{elephant}).
\item Transitivity follows from the fact that open maps are closed under composition.
\end{enumerate}

\item[$\cat{FSpc}_{nd}$] (\emph{Totality})
\begin{enumerate}
\item An isomorphism $p : V \to_o A$ satisfies $p_!(\top) = \top$ and so it alone is a pattern family on $A$.
\item Given a pattern family $p_i : U_i \to_o A$ and a continuous map $f : V \cto A$, the pullback object $f^*U_i$ is homeomorphic to $\{ V \mid f^*({p_i}_!(\top)) \}$. Using this definition of the pullback object, we confirm
\[
\bigvee_{i : I} (f^*p_i)_!(\top) 
  = \bigvee_{i : I} f^*({p_i}_!(\top))
  = f^* \left( \bigvee_{i : I} {p_i}_!(\top) \right)
  = f^*(\top)
  = \top.
\]
\item The family of composites indeed covers, since
\begin{align*}
\bigvee_{i : I, j : J_i} (p_i  \circ p_{i,j})_!(\top)
&= \bigvee_{i : I, j : J_i} {p_i}_!({p_{i,j}}_!(\top))
\\ &= \bigvee_{i : I} {p_i}_! \left( \bigvee_{j : J_i} {p_{i,j}}_!(\top) \right)
   \tag{direct images preserve joins}
\\ &= \bigvee_{i : I} {p_i}_! (\top)
\\ &= \top.
\end{align*}
\end{enumerate}

\item[$\cat{FSpc}_{p}$] (\emph{Determinism})
\begin{enumerate}
\item Any cover with a single pattern is trivially pairwise disjoint.
\item By Proposition \ref{pullbackopenembedding}, the pullback of an open embedding is an open embedding. 
Thus it remains to confirm that pullback preserves disjointness. From the computation
\begin{align*}
(f^*p_i)_!(\top) \wedge (f^*p_j)_!(\top)
&= f^*({p_i}_!(\top)) \wedge f^*({p_j}_!(\top))
\\ &= f^*({p_i}_!(\top) \wedge {p_j}_!(\top)),
\end{align*}
and since $f^*$ preserves joins, by Proposition \ref{thm:join-pos}
if $(f^*p_i)_!(\top) \wedge (f^*p_j)_!(\top)$ is positive, then so is ${p_i}_!(\top) \wedge {p_j}_!(\top)$, which implies that $i \inteq j$. Therefore, the pullback of the cover is still pairwise disjoint.

\item
Since the composition of open embeddings is an open embedding, it only remains to confirm that transitivity preserves disjointness.
Suppose we consider two composites $p_i \circ p_{i, j}$ and $p_{i'} \circ p_{i', j'}$ such that
\begin{align*}
&(p_i \circ p_{i, j})_!(\top) \wedge (p_{i'} \circ p_{i', j'})_!(\top)
\\ &= {p_i}_!({p_{i, j}}_!(\top)) \wedge {p_{i'}}_!({p_{i', j'}}_!(\top))
\end{align*}
is positive. This implies that the larger open
\[
{p_i}_!(\top) \wedge {p_{i'}}_!(\top)
\]
is also positive, and therefore, $i \inteq i'$. We will now prove that $j \inteq j'$. Since $i \inteq i'$, by equality induction we can consider them both $i$, so that we know that
\begin{align*}
&(p_i \circ p_{i, j})_!(\top) \wedge (p_i \circ p_{i, j'})_!(\top)
\\ &= {p_i}_!({p_{i, j}}_!(\top)) \wedge {p_i}_!({p_{i, j'}}_!(\top))
\\ &= {p_i}_!({p_{i, j}}_!(\top) \wedge {p_{i, j'}}_!(\top))
  \tag{${p_i}_!$ preserves meets}
\end{align*}
is positive. Note that any direct image map $f_!$ preserves joins, and accordingly, if $f_!(U)$ is positive then $U$ is positive. Therefore,
\[
{p_{i, j}}_!(\top) \wedge {p_{i, j'}}_!(\top)
\]
is positive, and by pairwise disjointness of the covering family, we know $j \inteq j'$.
Therefore, $(i, j) \inteq (i', j')$.

\end{enumerate}

\item[$\cat{FSpc}$] (\emph{Totality} and \emph{determinism}) 
This is just the intersection $\mathcal{J}_{nd} \cap \mathcal{J}_p$, and it is straightforward from the definition of a Grothendieck pretopology that they are closed under arbitrary intersection.
\end{description}
\end{proof}
\end{longversion}

\subsection{Syntax of pattern matching}
\label{s:contpl}

We now describe syntax for a programming language with spaces with pattern matching as guided by the above semantics.

Each of $\mathcal{C} \in \{ \cat{FSpc}, \cat{FSpc}_{nd}, \cat{FSpc}_{p}, \cat{FSpc}_{nd, p} \}$ are cartesian monoidal categories, meaning that they admit a restricted (first-order) $\lambda$-calculus syntax \cite{escardo2004}, with the typing rules

\begin{mathpar}
\inferrule*[right=var]
  {(x : A) \in \Gamma}
  {\Gamma \vdash_\mathcal{C} x : A}

\inferrule*[right=app]
  {f : A_1 \times \cdots \times A_n \to_\mathcal{C} B
  \\ \Gamma \vdash_\mathcal{C} e_1 : A_1 \\ \cdots \\ \Gamma \vdash_\mathcal{C} e_n : A_n}
  {\Gamma \vdash_\mathcal{C} f(e_1, \ldots, e_n) : B},
\end{mathpar}
where contexts are lists of spaces and in the \irule{var} rule $(x : A) \in \Gamma$ denotes a witness that $A$ is a member of the list $\Gamma$. Let $\mathsf{Prod} : \mathsf{list}(\cat{FSpc}) \to \cat{FSpc}$ denote the product of a list of types.
We use an unadorned turnstile $\vdash$ for continuous maps, $\mathcal{C} = \cat{FSpc}$.

\begin{proposition}
Given any expression $\Gamma \vdash_\mathcal{C} e : A$, we can construct a term $\denote{e} : \mathsf{Prod}(\Gamma) \to_\mathcal{C} A$.
\end{proposition}

\newcommand{\gor}{\ \ \mid\ \ }

\ConsiderRemoving{
\begin{figure}
\small
\newcommand{\shole}[1]{\langle \text{#1} \rangle}
\begin{align*}
\text{expression}\ e ::=\ &x
\gor f(e, \ldots, e)
\gor \mathsf{case}(e) \begin{cases}
p \Branch e
\\ \vdots
\\ p \Branch e
\end{cases}
\\
\text{pattern}\ p ::=\ &f
 \gor x
  \gor \_
  \gor f(p)
  \gor  p, p
\\
\text{function}\ f \phantom{::=\ }&
\\ \text{variable}\ x \phantom{::=\ }&
\end{align*}
\caption{Syntax for our simple language with pattern matching.}
\end{figure}
}

In Fig. \ref{fig:pattern-syntax} we define syntax and typing rules for patterns,
where $p : A \dashv \Gamma$ intuitively means that the pattern $p$ provides a context of pattern matching variables $\Gamma$ by pattern matching on a space $A$.
For instance, we could have the pattern
\[
(\up(x) , y)  : \Lifted{A} \times B \dashv x : A, y : B
\]
on a space $\Lifted{A} \times B$ that provides variables $x : A$ and $y : B$ to be used in the branch corresponding to that pattern.

\begin{figure}[htbp]
   \centering
\begin{mathpar}
\inferrule*[right=constant]
  {f : \One \to_o A}
  {f : A \dashv \cdot}

\inferrule*[right=var]
  { }
  {v : A \dashv v : A}
  
\inferrule*[right=wildcard]
  { }
  {\wildcard : A \dashv \cdot}
  
\inferrule*[right=compose]
  {p : U \dashv \Gamma \\ f : U \to_o A}
  {f(p) : A \dashv \Gamma}
  
\inferrule*[right=product]
  {p : A \dashv \Gamma \\ q : B \dashv \Delta}
  {p, q : A \times B \dashv \Gamma, \Delta}
\end{mathpar}
   \caption{Syntax and type rules for patterns.}
   \label{fig:pattern-syntax}
\end{figure}

\begin{theorem}
Given any pattern derivation $p : A \dashv \Gamma$, there is a map $\denote{p} : \mathsf{Prod}(\Gamma) \times \Delta \to_o A$ for some space $\Delta$ that collects the ``discarded variables'' from the wildcards\footnote{
Discarding variables is unnecessary if we only consider overt spaces $A$; a space $A$ is overt if the map $A \cto \One$ is an open map. Classically, all spaces are overt.
}.
\end{theorem}
\begin{shortversion}
\begin{proof}[Proof sketch]
Follows from the fact that open maps include the identity (\irule{var}, \irule{wildcard}) and are closed under composition (\irule{compose}) and parallel composition (\irule{product}).
\end{proof}
\end{shortversion}
\begin{longversion}
\begin{proof}
By induction on the derivation of the pattern:
\begin{description}
\item[\irule{constant}] By assumption, we have a constant map of the right kind.
\item[\irule{var}] We use $\mathsf{id} : A \to_{o} A$.
\item[\irule{wildcard}] We \emph{also} use $\mathsf{id} : A \to_{o} A$, but since we require $\Gamma \iso \One$, we set the ``garbage space'' $\Delta$ to be $A$.
\item[\irule{compose}] By induction, we have a map $\denote{p} : \Gamma \times \Delta \to_o U$ for some $\Delta$. Then we use the composition $(f \parmap \mathsf{id}_\Delta) \circ \denote{p} : \Gamma \times \Delta \to_o A$, threading through the ``garbage'' space $\Delta$.
\item[\irule{product}] By induction, we have maps $\denote{p} : \Gamma_1 \times \Delta_1 \to_o A$ and $\denote{q} : \Gamma_2 \times \Delta_2 \to_ B$. Their parallel composition $\denote{p} \parmap \denote{q} : \Gamma_1 \times \Gamma_2 \to_o A \times B$ is also a map of the right kind, which we can use, together with some homeomorphisms to rearrange the ``garbage'' spaces $\Delta_1$ and $\Delta_2$ and to pass them through.
\end{description}
\end{proof}
\end{longversion}

Note that the same pattern syntax would work with open embeddings instead of open maps, as they too form a category and are closed under parallel composition; patterns of open embeddings are necessary for constructing deterministic programs.

\begin{longversion}
Note that the definition of patterns we give here does \emph{not} admit some structural rules that may be expected, such as 
\begin{description}
\item[\emph{Strengthening}] Given $p : A \dashv \Gamma$, we cannot necessarily derive $p : A \dashv \Delta$ where $\Delta$ is a sub-list of $\Gamma$.
\item[\emph{Exchange}] Given $p : A \dashv \Gamma, \Delta$, we cannot necessarily derive $p : A \dashv \Delta, \Gamma$.
\end{description}
These structural properties aren't necessary since patterns are \emph{providing} a context to be \emph{used} by the expression language, which \emph{does} have the corresponding rules weakening and exchange. Regardless, the \irule{wildcard} rule still gives a sort of ``explicit'' strengthening, allowing the user to discard some variables.
\end{longversion}

We can now define a general (mostly) syntactic rule for interpreting the pattern matches described in this section.

\begin{theorem}
We can interpret the syntax\footnote{
As mnemonics, $I$ stands for \emph{index} type, $s$ stands for \emph{scrutinee} of a case expression, $p$ stands for \emph{pattern}, and $e$ for \emph{expression}.}
\begin{mathpar}
\inferrule*[right=case-$\mathcal{C}$]
{\Gamma \vdash_\mathcal{C} s : A
 \\ \prod_{i : I} p_i : A \dashv_o A_i
 \\ \prod_{i : I} \Gamma, A_i \vdash_\mathcal{C} e_i : B
 \\ (\denote{p_i})_{i : I} \in \mathcal{J}_\mathcal{C}}
{\Gamma \vdash_\mathcal{C} \left(\mathsf{case}(s) 
    \begin{cases} [i : I] \quad p_i &\Branch e_i
    \end{cases}\right) : B}
\end{mathpar}
(with the one non-syntactic side condition $(\denote{p_i})_{i : I} \in \mathcal{J}_\mathcal{C}$),
where $\mathcal{C} \in \{ \cat{FSpc}, \cat{FSpc}_{nd}, \cat{FSpc}_{p}, \cat{FSpc}_{nd, p} \}$.
\end{theorem}
\begin{proof}
The syntactic constructions give us maps $\denote{s} : \Gamma \to_\mathcal{C} A$, $\denote{p_i} : A_i \times \Delta_i \to_o A$ (for some spaces $\Delta_i$ representing discarded variables in the pattern $p_i$), and $\denote{e_i} : \Gamma \times A_i \to_\mathcal{C} B$. The condition $(\denote{p_i})_{i : I} \in \mathcal{J}_\mathcal{C}$ means that these maps appropriately cover $A$, so that we get an appropriately behaved map
$\mathsf{inv} : A \to_\mathcal{C} \sum_{i : I} A_i \times \Delta_i$.

We must produce a map $f : \Gamma \to_\mathcal{C} B$. We can do so by defining\footnote{
While the ``let-in'' syntax for using the universal property (itself a sort of pattern matching) of sums has not been formally described, hopefully it is clear how it can be implemented via categorical semantics.}
\begin{align*}
f : \Gamma &\to_\mathcal{C} B
\\ f(\gamma) &\triangleq 
\text{  let  } \langle i , x \rangle \triangleq \mathsf{inv}(\denote{s}(\gamma))
  \text{  in  } \denote{e_i}(\gamma, \mathsf{fst}(x)).
 \qedhere
\end{align*}
\end{proof}

Note that the condition $(\denote{p_i})_{i : I} \in \mathcal{J}_\mathcal{C}$ that the patterns lie in the appropriate Grothendieck pretopology is trivial when $\mathcal{C} = \cat{FSpc}_{nd, p}$, making the rule purely syntactic.

\section{Formal logics for approximate decision procedures}

\label{bcover}

In this section, we develop a formal logic for constructing decision procedures on spaces that may be either partial or nondeterministic, by simply considering partial and/or nondeterministic Boolean values. \grammarr{This} is useful where there are not total and deterministic decision procedures: for instance, there are no nontrivial maps $\R \cto \bool$ (since $\R$ is connected) but plenty of nondeterministic maps $\R \ndto \bool$ or partial maps $\R \pto \bool$.

Conventional decision procedures in general functional programming correspond to decidable predicates, or functions returning Boolean values. Decidable predicates are closed under conjunction, disjunction, and negation since those operations are computable on $\bool$, and additionally they are closed under universal and existential quantification over finite\footnote{
By ``finite,'' we always mean Kuratowski-finite \cite{johnstonetopos}.}
sets.

Analogously, $\bool$-valued continuous maps are also closed under conjunction, disjunction, and negation, and they admit quantification over \emph{compact-overt} spaces (which generalize finite sets). We will show that these operations still work when partiality or nondeterminism is admitted, where Boolean logic is generalized to many-valued logic with the structure of a quasi-Boolean algebra.

\begin{proposition}
\label{bool-ndp}
Maps $A \ndpto \bool$ are in bijective correspondence with pairs $(P, Q)$ of opens of $A$.
For maps $A \ndto \bool$ the opens are \emph{covering}, i.e.,  $\top \le P \vee Q$, and for maps $A \pto \bool$ they are \emph{disjoint}, i.e., $P \wedge Q \le \bot$.
\end{proposition}

This correspondence between pairs of opens and Boolean-valued maps establishes two alternative perspectives on approximate decision-making, one spatial and one more algorithmic in flavor.
We think of $P$ as the ``true'' region and $Q$ the ``false'' region. We can also think of $Q$ as representing the closed subspace complement $\overline{Q}$ of the open subspace $Q$, in which case  $T \le P \vee Q$ corresponds to the subspace inclusion $\overline{Q} \subseteq P$, and $P \wedge Q \le \bot$ corresponds to $P \subseteq \overline{Q}$.

Adding nondeterminism or partiality changes the behavior in comparison to deterministic decision procedures. We can characterize \grammarr{this} algebraically. In $\cat{Set}$, $\bool$ forms a Boolean algebra. Since the functor $\mathsf{Discrete} : \cat{Set} \to \cat{FSpc}$ is full, faithful, and preserves binary products and the terminal object, \grammarr{this} lifts to an internal Boolean algebra (within $\cat{FSpc}$) on $\bool$ the space. We can again lift these operations by the forgetful functor $\cat{FSpc} \to \cat{FSpc}_{\mathcal{C}}$; for instance, the lifted version of the Boolean ``and'' operation ($\band : \bool \times \bool \ndpto \bool$) is potentially $\btrue$ if its argument potentially takes on values whose conjunction is equal to $\btrue$. However, the operations on $\bool$ no longer form a Boolean algebra within either $\cat{FSpc}_{nd}$ or $\cat{Fspc}_{p}$:

\begin{proposition}
The space $\bool$ does not form a Boolean algebra (with its usual operations) within either $\cat{FSpc}_{nd}$ or $\cat{Fspc}_{p}$.
\end{proposition}
\begin{proof}
In particular, there is the nondeterministic value $\mathsf{both} : \One \ndto \bool$ and the partial value 
$\mathsf{neither} : \One \pto \bool$ satisfying 
\begin{align*}
\mathsf{both}^*(\cdot = \btrue) = \mathsf{both}^*(\cdot = \bfalse) &= \top
\\ \mathsf{neither}^*(\cdot = \btrue) = \mathsf{neither}^*(\cdot = \bfalse) &= \bot,
\end{align*}
which implies that
\begin{align*}
\mathsf{both} \bor \bneg \mathsf{both} = \mathsf{both} &\neq \btrue
\\ \mathsf{neither} \bor \bneg \mathsf{neither} = \mathsf{neither} &\neq \btrue,
\end{align*}
whereas in a Boolean algebra there is the identity $x \bor \bneg x = \btrue$.
\end{proof}

Though $\bool$ does not form a Boolean algebra in $\cat{FSpc}_{nd,p}$, it comes close:
\begin{proposition}
The space $\bool$ forms a quasi-Boolean algebra (or De Morgan algebra) in $\cat{FSpc}_{nd,p}$, meaning that $\bool$ with $\band, \bor, \btrue$, and $\bfalse$ forms a bounded distributive lattice, and $\bneg$ is a De Morgan involution, in that it satisfies $\bneg \bneg x = x$ and $\bneg (x \band y) = \bneg x \bor \bneg y$.
\end{proposition}
\begin{proof}[Proof sketch]
It is instructive to observe how the operations act on generalized points $\Gamma \ndpto \bool$; we will use their equivalent representation as pairs of opens of $\Gamma$. Observe that
\begin{align*}
\begin{aligned}
(P_1, Q_1) \band (P_2, Q_2) &= (P_1 \wedge P_2, Q_1 \vee Q_2)
\\ (P_1, Q_1) \bor (P_2, Q_2) &= (P_1 \vee P_2, Q_1 \wedge Q_2)
\\ \bneg (P, Q) &= (Q, P)
\end{aligned}
\\
\begin{aligned}
\btrue &= (\top, \bot) \quad
& \quad \bfalse &= (\bot, \top).
\end{aligned}
\end{align*}
\begin{shortversion}
We can use these equations to confirm the various laws, for instance,
$
\bneg \bneg (P, Q) = \bneg (Q, P) = (P, Q).
$
\end{shortversion}
\begin{longversion}
We can use this to confirm the various laws, for instance, that $\btrue$ is the identity for $\band$,
\[
\btrue \band (P, Q) = (\top \wedge P, \bot \vee Q) = (P, Q).
\]
We similarly can confirm $\bfalse$ is the identity for $\bor$, and it is easy to observe that $\band$ and $\bor$ are commutative and associative. Absorption of $\band$ and $\bor$ follows from the similar absorption properties of opens, and likewise for their distributivity properties.

It remains to confirm that $\bneg$ is a deMorgan involution. We have
\[
\bneg \bneg (P, Q) = \bneg (Q, P) = (P, Q)
\]
and
\begin{align*}
\bneg ((P_1, Q_1) \band (P_2, Q_2))
  &= \bneg (P_1 \wedge P_2, Q_1 \vee Q_2)
  = (Q_1 \vee Q_2, P_1 \wedge P_2)
\\ &= (Q_1, P_1) \bor (Q_2, P_2)
  = \bneg (P_1, Q_2) \bor \bneg (P_2, Q_2).
\end{align*}
\end{longversion}
\end{proof}

The argument also shows that $\bool$ is a quasi-Boolean algebra in $\cat{FSpc}_{nd}$ and $\cat{FSpc}_{p}$ as well. We will show that in each variant, it is possible to quantify these approximate decision procedures over compact-overt spaces.

\subsection{Quantification over compact-overt spaces}

When working with sets, if a predicate $P$ on a set $A$ is decidable and if $A$ is finite, then $\forall a : A.\ P(a)$ and $\exists a : A.\ P(a)$ are decidable as well. The spatial analogue of the finite sets is the compact-overt spaces.

A space $\Sierp$, called the Sierpi\'nski space, is useful in describing the logic of opens:
there is a correspondence $\Open{A} \iso A \cto \Sierp$ between opens of $A$ and $\Sierp$-valued continuous maps on $A$ for any space $A$. We can use \grammarr{this} to describe opens via $\Sierp$-valued continuous maps. We  use the notation $\{ x : A \mid e \}$ where $e$ is a $\Sierp$-valued term that may mention $x$, i.e., $x : A \vdash e : \Sierp$, to denote the open subspace $\{ A \mid \denote{e} \}$. For instance, we can define the open subspace $\{ x : \R \mid x \times x < 2 \}$, where $(<) : \R \times \R \cto \Sierp$. We will readily conflate opens and $\Sierp$-valued continuous maps, implicitly converting between the two.

\subsubsection{On compact-overt spaces}

\begin{definition}[\citealt{vickerspowerpoints}]
\label{compact-overt}
A space $K$ is \emph{compact} if for every space $\Gamma$, the functor $- \times \top_K : \Open{\Gamma} \to \Open{\Gamma \times K}$ has a right adjoint $\forall_K : \Open{\Gamma \times A} \to \Open{\Gamma}$\footnote{
This definition of compactness is equivalent to the more common one, that every open cover has a finite subcover.
}. Similarly, a space $A$ is \emph{overt} if for every space $\Gamma$, $- \times \top_A$ has a left adjoint $\exists_A : \Open{\Gamma \times A} \to \Open{\Gamma}$. 
A space is \emph{compact-overt} if it is compact and overt.
\end{definition}

These conditions are the definitions of universal and existential quantification in terms of adjoints, viewing $\Gamma$ as some context and opens as truth values in a context.
These adjunctions allow us to define syntax for quantification of $\Sierp$-valued continuous maps on compact-overt spaces:
\begin{mathpar}
\inferrule*
  {\Gamma, x : K \vdash e : \Sierp
  \\ K \text{ compact}}
  {\Gamma \vdash \forall x \in K.\ e : \Sierp}
  
\inferrule*
  {\Gamma, x : A \vdash e : \Sierp
  \\ A \text{ overt}}
  {\Gamma \vdash \exists x \in A.\ e : \Sierp}
\end{mathpar}

For any compact-overt space $K$, for any $\Gamma$ and $P, Q : \Open{\Gamma \times K}$, we have in $\Gamma$ \cite{vickerspowerpoints}
\begin{align*}
\forall_K(P \vee Q) \le \forall_K P \vee \exists_K Q
\quad \text{and} \quad
  \forall_K P \wedge \exists_K Q \le \exists_K (P \wedge Q).
\end{align*}
These properties allow us to quantify over compact-overt spaces, too. That is, we can add some syntax
\begin{mathpar}
\inferrule*
  {\Gamma, x : K \vdash_{\mathcal{C}} e : \bool
  \\ K \text{ compact-overt}}
  {\Gamma \vdash_{\mathcal{C}} \forall x \in K.\ e : \bool}
  
\inferrule*
  {\Gamma, x : K \vdash_{\mathcal{C}} e : \bool
  \\ K \text{ compact-overt}}
  {\Gamma \vdash_{\mathcal{C}} \exists x \in K.\ e : \bool}
\end{mathpar}
that behaves as we would expect (for a quasi-Boolean algebra, at least). We interpret this syntax by defining quantification functionals of the type $(\Gamma \times K \xto{\mathcal{C}} \bool) \to (\Gamma \xto{\mathcal{C}} \bool)$.
For a compact-overt space $K$, we define a universal-quantification functional
\begin{align*}
\forall_K &: (\Gamma \times K \xto{\mathcal{C}} \bool) \to (\Gamma \xto{\mathcal{C}} \bool)
\\ \forall_K (P, Q) &\triangleq (\forall_K P, \exists_K Q).
\end{align*}
We confirm this definition works for $\mathcal{C} = nd$ because it preserves covering: if $\top \le P \vee Q$, then
\begin{shortversion}
\begin{align*}
\top_\Gamma
\le \forall_K (\top_{\Gamma \times K}) 
\le \forall_K (P \vee Q) 
\le \forall_K P \vee \exists_K Q. 
\end{align*}
\end{shortversion}
\begin{longversion}
\begin{align*}
\top_\Gamma 
  &\le \forall_A (\top_{\Gamma \times A}) \tag{$\forall_A$ adjointness}
\\ &\le \forall_A (P \vee Q) \tag{$(P, Q)$ cover $\Gamma \times A$, $\forall_A$ monotone}
\\ &\le \forall_A P \vee \exists_A Q. \tag{$A$ compact-overt}
\end{align*}
\end{longversion}
Dually, it works for $\mathcal{C} = p$ since it preserves disjointness: if $P \wedge Q \le \bot$, then
\begin{shortversion}
\begin{align*}
\forall_K P \wedge \exists_K Q
\le
\exists_K (P \wedge Q)
\le 
\exists_K (\bot_{\Gamma \times K})
\le \bot_\Gamma.
\end{align*}
\end{shortversion}
\begin{longversion}
\begin{align*}
\forall_A P \wedge \exists_A Q
   &\le \exists_A (P \wedge Q) \tag{$A$ compact-overt}
\\ &\le \exists_A (\bot_{\Gamma \times A}) \tag{$(P, Q)$ disjoint, $\exists_A$ monotone}
\\ &\le \bot_\Gamma. \tag{$\exists_A$ adjointness}
\end{align*}
\end{longversion}
We can similarly define the existential-quantification functional by $\exists_K (P, Q) \triangleq (\exists_K P, \forall_K Q)$.
\begin{longversion}
By inspection, the two quantifiers are related by the law
\[
\bneg (\forall_A f) = \exists_A (\bneg f).
\]
\end{longversion}

We make the notion that these operations are quantifiers precise by showing that these quantification functionals are adjoints to a weakening functional. The quasi-Boolean algebra on $\bool$ determines the preorder which we call \emph{truth order} on maps $A \ndpto \bool$: representing maps as pairs of opens, we define $(P_1, Q_1) \leto (P_2, Q_2)$ if and only if both $P_1 \le P_2$ and $Q_2 \le Q_1$.

For any spaces $\Gamma$ and $A$, weakening 
($(- \circ \mathsf{fst}) : (\Gamma \ndpto \bool) \to (\Gamma \times A \ndpto \bool)$)
is monotone with respect to truth order. The quantifiers deserved to be called such:

\begin{theorem}
The existential- and universal-quantification functionals are left and right adjoints to weakening, respectively, with respect to truth order, i.e.,
$
\exists_K \dashv (- \circ \mathsf{fst}) \dashv \forall_K.
$
\end{theorem}
\begin{longversion}
\begin{proof}
First we prove that $\exists_A$ is left adjoint to weakening: given a map $\Gamma \times A \ndpto \bool$ represented by the opens $(P, Q)$ a map $\Gamma \ndpto \bool$ represented by opens $(U, V)$, we must show
\begin{align*}
(P, Q) \leto (- \circ \mathsf{fst})(U, V)
  \quad &\iff \quad
  \exists_A(P, Q) \leto (U, V)
\\
(P, Q) \leto (U \times \top, V \times \top)
  \quad &\iff \quad
  (\exists_A P, \forall_A Q) \leto (U, V)
\\
P \le U \times \top \text{ and } V \times \top \le Q
  \quad &\iff \quad
  \exists_A P \le U \text{ and } V \le \forall_A Q.
\end{align*}
In the final form, we see that this follows from adjoint properties of the $\exists_A$ and $\forall_A$ operations that act on opens of $\Gamma \times A$.

The proof that $\forall_A$ is a right adjoint is a mirror image of the proof regarding $\exists_A$.
\end{proof}
\end{longversion}

\subsubsection{On compact-overt subspaces}

Sometimes, the space that we might want to quantify over could depend on some continuous variables in the context. For instance, we may want to quantify a predicate $f : \R \times \R \xto{\mathcal{C}} \bool$ over the triangle in $\R \times \R$ bounded by $(0, 0)$, $(1, 0)$, and $(0, 1)$. We will describe a formalism whereby it will be possible to write \grammarr{this} as
$
\forall x \in [0, 1].\ \forall y \in [0, 1 - x].\ f(x, y).
$
We handle this situation by considering spaces whose points represent compact-overt subspaces of some space.

There is a connection between overt spaces and the partiality-and-nondeterminism monad $\PLower$: Every point of $\PLower(A)$ corresponds to an overt subspace of $A$\footnote{
Specifically, the points of $\PLower(A)$ are in bijective correspondence with the \emph{weakly closed} overt subspaces of $A$ \cite[Theorem 32]{sublocFT}.}. Similarly, for each space $A$ there is a powerspace $\PUpper(A)$ whose points correspond with compact subspaces of $A$ \cite{vickersconnected, vickersdoublepowerlocale}. We summarize its salient characteristics. There is a ``necessity'' modality $\square : \Open{A} \to \Open{\PUpper(A)}$ that distributes over meets and directed joins (analogous to the ``possibility'' modality $\possibly : \Open{A} \to \Open{\PLower(A)}$ for the lower powerspace). Continuous maps $\Gamma \cto \PUpper(A)$ are in bijective correspondence with inverse image maps $\Open{A} \to \Open{\Gamma}$ that preserve meets and directed joins. Like $\PLower$, $\PUpper$ is a strong monad.

The powerspace analogue of the compact-overt spaces is called the \emph{Vietoris powerspace} $\Viet$ \cite{topologyvialogic}. Points of $\Viet(A)$ correspond to compact-overt subspaces of $A$. The space $\Viet$ has both the possibility and necessity modalities that interact exactly as with the compact-overt spaces.
\begin{longversion}
That is, for any opens $P, Q : \Open{A}$, the following laws hold:
\begin{align}
\label{boxdiamond}
\square (P \vee Q) &\le \square P \vee \possibly Q
\\  \square P \wedge \possibly Q &\le \possibly (P \wedge Q).
\end{align}
Like for the lower powerspaces, there is the ``positive'' subspace $\Viet^+(A)$ of $\Viet(A)$ that additionally satisfies $\top \le \possibly \top$ and $\square \bot \le \bot$.
\end{longversion}

We can add some additional syntax to make it easier to describe opens with these modalities:
\begin{mathpar}
\inferrule*
  {\Gamma \vdash s : \PLower(A)
  \\ \Gamma, x : A \vdash e : \Sierp}
  {\Gamma \vdash \exists x \in s.\ e : \Sierp}
  
\inferrule*
  {\Gamma \vdash s : \PUpper(A)
  \\ \Gamma, x : A \vdash e : \Sierp}
  {\Gamma \vdash \forall x \in s.\ e : \Sierp}.
\end{mathpar}
This syntax is interpreted using the correspondence
$
\Sierp \iso \PLower(\One) \iso \PUpper(\One)
$
\cite{townsend} via the strong monadic ``bind'' operations of $\PLower$ and $\PUpper$.
\begin{longversion}
Specifically, replacing $\Sierp$ with $\PLower(\One)$ in \irule{$\exists$-$\Sierp$-sub-form} and
with $\PUpper(\One)$ in \irule{$\forall$-$\Sierp$-sub-form}, these rules just correspond to the bind operation of a strong monad $M$, which composes maps $\Gamma \to M(A)$ and $\Gamma \times A \to M(B)$ to produce $\Gamma \to M(B)$.
\end{longversion}

Accordingly, we can quantify our Booleans over compact-overt subspaces as well in the same way, implementing the syntax

\begin{mathpar}
\mprset {sep=1.5em}
\inferrule*
  {\Gamma \vdash s : \Viet(A) \\ \Gamma, x : A \vdash_{\mathcal{C}} e : \bool}
  {\Gamma \vdash_{\mathcal{C}}\mathcal{Q} x \in s.\ e : \bool}
\end{mathpar}
where $\mathcal{Q}$ is either $\forall$ or $\exists$.
Just as in the case of compact-overt spaces, these definitions preserves both covering and disjointness, so $\mathcal{C}$ can have any combination of partiality and nondeterminism.

\begin{longversion}
This syntax is implemented in the following manner: Given any map $f : \Gamma \times A \to_\mathcal{C} \bool$ of $A$, we define its universal quantification
\begin{align*}
\forall_{[\cdot]} f &: \Gamma \times \Viet(A) \to_\mathcal{C} \bool
\\ \forall_{[\cdot]} f &\triangleq \mathsf{case}(p)
\begin{cases}
\oincl{\fun{(\gamma, s)}{\forall x \in s.\ f(\gamma, x) = \btrue}}{\wildcard} &\Branch \btrue
\\ \oincl{\fun{(\gamma, s)}{\exists x \in s.\ f(\gamma, x) = \bfalse}}{\wildcard} &\Branch \bfalse
\end{cases}.
\end{align*}
We can confirm $\forall_{[\cdot]} f$ is total, i.e.,
\[
\top \le (\fun{(\gamma, s)}{\forall x \in s.\ f(\gamma, x) = \btrue}) \vee
            (\fun{(\gamma, s)}{\exists x \in s.\ f(\gamma, x) = \bfalse})
\]
if $f$ is total
with the derivation\footnote{
The algebraic manipulations in this derivation are justified by the corresponding ones on $\Viet(A)$. For instance $\top \le \square \top$ implies 
\[
\fun{(\gamma, s)}{\top} \quad \le \quad \fun{(\gamma, s)}{\forall x \in s.\ \top}.
\]
}
\begin{align*}
\top 
  &\le \fun{(\gamma, s)}{\forall x \in s.\ \top}
    \tag{$\square$ preserves $\top$}
\\ &\le \fun{(\gamma, s)}{\forall x \in s.\ f(x) = \btrue \vee f(x) = \bfalse}
\\ &\le \fun{(\gamma, s)}{(\forall x \in s.\ f(x) = \btrue) \vee \exists x \in s.\ f(x) = \bfalse}
   \tag{law \ref{boxdiamond}}.
\end{align*}

Similarly, $\forall_{[\cdot]} f$ is deterministic if $f$ is deterministic.

We can define $\exists_{[\cdot]} f : \Viet(A) \ndpto \bool$ by composing $\forall_{[\cdot]} f$ with Boolean negation.
\end{longversion}

Compact-overt subspaces form a convenient class of spaces over which exhaustive reasoning is possible. The continuous image of a compact-overt space is compact-overt (just as the image of a finite set under any map is finite). Like finite subsets, compact-overt subspaces are closed under finitary union but not necessarily intersection. Naturally, a finite set viewed as a discrete space is compact-overt.

\section{Implementation and case studies}

\subsection{Implementation in Marshall}
We implemented a pattern-matching construct as well as a library for partial and/or nondeterministic decision procedures within the Marshall programming language for exact real arithmetic \cite{marshall}, which is based on Abstract Stone Duality, a related, though different, theory of constructive topology. 

Marshall's type system includes \lstinline!real! ($\R$), \lstinline!prop! ($\Sierp$), finitary products, and function types. Notably, it lacks discrete types such as $\bool$ and has no support for subspace types. However, we used \lstinline!prop * prop! to simulate $\bool$, using the correspondence with pairs of opens described in Proposition \ref{bool-ndp}.

Partiality is intrinsic to Marshall, whether in evaluation of terms of type \lstinline!prop!, or evaluation of real numbers defined by Dedekind cuts where there is a gap between the left and right cuts. In the course of adding a pattern-match construct and computational support for it, we added support for nondeterminism in this manner, which was not previously available. Accordingly, Marshall effectively allows programming in $\cat{FSpc}_{nd, p}$.

Our pattern-match construct in Marshall has syntax and is typed as follows:
\begin{mathpar}
\inferrule{
\forall i \in \{1, \ldots, n \}, p_i : \texttt{prop}
\\
\forall i \in \{1, \ldots, n \}, e_i : t
}
{
\texttt{(} p_1 \texttt{\textasciitilde>} e_1 \texttt{||} \ldots \texttt{||} p_n \texttt{\textasciitilde>} e_n \texttt{)} : t
}
\end{mathpar}
\grammarr{This} is different, and substantially less general, than pattern matching described in \refsection{patterns}: there is no variable binding, and only finitely many cases are permitted. Regardless, this construct suffices to enable implementation of the approximate decision procedures of \refsection{bcover}, including quantification over those compact-overt spaces that are available in Marshall, which are closed intervals with rational endpoints (and implicitly, their finitary products), as well as to implement simplified versions of examples that follow. The modified version of Marshall and examples are available at \url{https://github.com/psg-mit/marshall-lics/}.
\commentt{\todo{are you making the cad examples available?}}
We describe the semantics of Marshall and its extension in further detail in
\shortlong{\cite{arXiv}}{Appendix \ref{a:marshall}}.

\subsection{Case studies}
\label{applications}

\subsubsection{Autonomous car approaching a yellow light}

Sometimes, partiality is unacceptable.
Consider an autonomous car that is approaching a traffic light that has just turned yellow. To ensure safety, the car must be outside of the intersection when the light turns red. \grammarr{This} requires the discrete decision to be made of whether or not to proceed through the intersection. We will model this problem with some additional concrete detail and demonstrate that it is impossible to do so deterministically, but with nondeterminism it is possible to write a program with a formal safety guarantee.

We model the car's state when the light turns yellow as a position and velocity, 
$
\mathsf{CarState} \triangleq \{ (x, v) : \R \times \R \mid 0 < v < v_\text{max} \},
$
where the velocity is positive but bounded. At the moment the light turns yellow, the car may choose a constant acceleration in the range
$
\mathsf{Accel} \triangleq \{ \R \mid a_\text{min} < \cdot < a_\text{max} \},
$
limited by the car's physical capabilities (with $a_\text{min} < 0 < a_\text{max}$).
The continuous map
$
\mathsf{pos} : \mathsf{CarState} \times \mathsf{Accel} \cto \R
$
computes where the car will lie when the light turns red, given the car's state when the light turns yellow and the chosen acceleration in the intervening period\footnote{
However, we enforce that velocity remains nonnegative, so if the car decelerates to zero velocity before the light turns red, it remains stopped rather than going backwards.}.
We define the position 0 to be where the intersection begins and let $w > 0$ mark the end of the intersection. Thus, we define $\mathsf{safe} : \Open{\R}$ by $\mathsf{safe} \triangleq (\cdot < 0) \vee (w < \cdot)$.

The problem of choosing an acceleration to safely navigate the intersection is that of finding a function $f : \mathsf{CarState} \cto \{ y : \mathsf{CarState} \times \mathsf{Accel} \mid \mathsf{safe}(\mathsf{pos}(y)) \}$ such that $\mathsf{fst} \circ f = \mathsf{id}$.

\begin{proposition}
It is impossible to continuously choose an acceleration to safely navigate the intersection (i.e., there is no continuous map $f$ as described above).
\end{proposition}
\begin{proof}
Note that $\{ y : \mathsf{CarState} \times \mathsf{Accel} \mid \mathsf{safe}(\mathsf{pos}(y)) \}$ has two connected components, corresponding to whether the car is before the intersection ($\mathsf{pos}(y) < 0$) or past the intersection ($w < \mathsf{pos}(y)$) when the light turns red.
If there were such an $f$, then since $\mathsf{CarState}$ is connected, so would be the image of $f$, meaning that $f$ must, regardless of the initial car state, always make the same decision of whether to stop for the intersection or proceed through. But if the car's initial state is sufficiently far back from the intersection, it could not choose an acceleration that ensures it is past the intersection when the light turns red. Conversely, if the car is already past the intersection, it cannot go backwards and thus cannot ensure it is before the intersection when the light turns red.
\end{proof}

\begin{proposition}
However, if we permit nondeterminism, we \emph{can} produce a map $f : \mathsf{CarState} \ndto \{ y : \mathsf{CarState} \times \mathsf{Accel} \mid \mathsf{safe}(\mathsf{pos}(y)) \}$ that nondeterministically chooses an acceleration that is always safe, assuming some conditions on the constants $a_\text{min}$, $a_\text{max}$, $v_\text{max}$, $w$, and the time $T$ between when the light turns yellow and when it turns red.
\end{proposition}

\begin{proof}
We outline one possible solution.
Let $\varepsilon > 0$ be some buffer distance. We can compute as a function of the initial car state the necessary acceleration to proceed through the light and be at position $w + \varepsilon$ when the light turns red, as well as the necessary deceleration to stop before the light at position $-\varepsilon$,
\begin{align*}
a_\text{go} &: \mathsf{CarState} \cto \R
\\ a_\text{go}(x, v) &\triangleq \max\left(0, 2(w + \varepsilon - x - vT) / T^2\right)
\\ a_\text{stop} &: \mathsf{CarState} \pto \R
\\ a_\text{stop}(x, v) &\triangleq \frac{v^2}{2(x + \varepsilon)}.
\end{align*}
Note that $a_\text{go}(x, v)$ is always nonnegative, and $a_\text{stop}(x, v)$ is always nonpositive.  We assemble the final solution as
\begin{align*}
f &: \mathsf{CarState} \ndto \{ y : \mathsf{CarState} \times \mathsf{Accel} \mid \mathsf{safe}(\mathsf{pos}(y)) \}
\\ f(s) &\triangleq \mathsf{case}(s)
\begin{cases}
\oincl{\fun{c}{a_\text{go}(c) < a_\text{max}}}{c'}
  &\Branch (c', a_\text{go}(c'))
\\ \oincl{\fun{c}{a_\text{min} < a_\text{stop}(c)}}{c'}
  &\Branch (c', a_\text{stop}(c'))
\end{cases}.
\end{align*}
Formal proof would be required to show that the output of each branch is indeed within the required subspaces indicated by the output type.

With sufficient conditions on the constants, we can prove that the cases of $f$ are covering,
i.e., that it is always the case that either the $\mathsf{go}$ or $\mathsf{stop}$ strategy is applicable\footnote{
The cases cover so long as
$w + 2 \varepsilon \le \frac{1}{2}T^2a_\text{max}$,
i.e., it is possible to speed up from a standstill to cross the intersection, and
\begin{align*}
\label{eqn:vmax}
v_\text{max} < -a_\text{min}
\left(
  T + \sqrt{T^2 - 2\left(\text{\textonehalf} T^2 a_\text{max} - w - 2\varepsilon \right)/a_\text{min}}
\right)
\end{align*}
which guarantees that the car never goes so fast that it is too close to stop and also cannot speed up to pass through the intersection in time.
}.
\end{proof}

\begin{figure}
\begin{lstlisting}[basicstyle=\small\ttfamily]
let a_go = fun x : real => fun v : real =>
    max 0 (2 * (w + eps - x - v * T) / (T * T));;
let a_stop = fun x : real => fun v : real =>
    v * v / (2 * (x + eps));;
let accel = fun x : real => fun v : real =>
  (  a_go x v   < a_max  ~>  a_go x v
  || a_stop x v > a_min  ~>  a_stop x v );;
\end{lstlisting}
\caption{A Marshall program that uses an overlapping pattern match to nondeterministically compute the desired acceleration of an autonomous car approaching a traffic light.}
\label{marshall-car-example}
\end{figure}

\grammarr{This} translates to the Marshall program in Fig. \ref{marshall-car-example}.

\subsubsection{Approximate root-finding}
\label{s:bcov:root}

Given an arbitrary continuous function $f : K \cto \R$, where $K$ is compact-overt, then least one of the following statements must hold:
\begin{itemize}
\item There is some $x \in K$ such that $|f(x)| < \varepsilon$.
\item For every $x \in K$, $f(x) \neq 0$.
\end{itemize}

The following root-finding program nondeterministically computes which, in the former case computing some $x \in K$ that is almost a root:
\begin{align*}
\rootfindingcode.
\end{align*}

That the two cases cover follows from the logic in \refsection{bcover}: the opens $| \cdot | < \varepsilon$ and $\cdot \neq 0$ cover $\R$, and covering opens are stable under pullback by continuous maps $f$ and quantification over compact-overt spaces $K$.
\begin{longversion}
Specifically, we can define a map
\begin{align*}
[ \cdot \neq 0 ]_\varepsilon &: \R \ndto \bool
\\ [ x \neq 0 ]_\varepsilon &\triangleq \mathsf{case}(x)
\begin{cases}
\oincl{\cdot \neq 0}{\wildcard} &\Branch \btrue
\\ \oincl{\fun{x}{|x| < \varepsilon}}{\wildcard} &\Branch \bfalse
\end{cases}
\end{align*}
that approximately determines whether a real number is nonzero. Then the map
\begin{align*}
\mathsf{no\_roots}_\varepsilon' &: \One \ndto \bool
\\ \mathsf{no\_roots}_\varepsilon' &\triangleq \forall x \in K.\ [x \neq 0]_\varepsilon
\end{align*}
is defined by the same pair of opens as in the definition of $\mathsf{roots}_f$, and hence those opens cover $\One$.
\end{longversion}

It remains to define $\mathsf{simulate}$, which in general has the type
\[
\mathsf{simulate} : \{ \One \mid \exists x \in A.\ U(x) \}
  \ndto \{ x : A \mid U(x) \}
\]
for any overt space $A$ and open $U$ of $A$ (in $\mathsf{roots}_f$, $A$ is $K$ and $U$ is $\fun{x}{|f(x)| < \varepsilon}$). Given the existence of some values that satisfy a property $U$ of $A$, $\mathsf{simulate}$ can nondeterministically simulate those values. It is defined by the inverse image map
\begin{align*}
\mathsf{simulate}^* &: \Open{\{ A \mid U \}} \to \Open{\{ \One \mid \exists_A U \}}
\\ \mathsf{simulate}^*(V) &\triangleq \exists_A(V \wedge U).
\end{align*}

\begin{longversion}
\begin{proposition}
The inverse image map $\mathsf{simulate}^*$ preserves joins and $\top$.
\end{proposition}
\begin{proof}
First, we confirm that the output lies in the open subspace $U$ of $A$, that is, that $\mathsf{simulate}^*$ is gives equivalent results on inputs $V$ and $W$ satisfying $V \wedge U = W \wedge U$. This is straightforward because $\mathsf{simulate}^*$ immediately applies $\cdot \wedge U$ to its argument.

The map preserves joins since it is the composition $\exists_A \circ (\cdot \wedge U)$, both of which preserve joins. It preserves $\top$ (or equivalently, $U$), since
\[
\mathsf{simulate}^*(\top) = \exists_A(U),
\]
which is equal to $\top$ in $\{ \One \mid \exists_A U \}$.
\end{proof}
\end{longversion}

The approximate root-finding program $\mathsf{roots}_f$ accomplishes the task of approximate root-finding over a very general class of functions with a very short definition that works by composing constructs from \refsection{patterns} and \refsection{bcover}. That $K$ is compact-overt means that it implements a general computational interface for exhaustive search.

\begin{figure}
\begin{lstlisting}[basicstyle=\small\ttfamily]
let exists_bool_interval = fun pred : real -> bool =>
     (exists x : [0,1], is_true  (pred x)) ~> tt
  || (forall x : [0,1], is_false (pred x)) ~> ff ;;
let is_0_eps = fun x : real =>
        x < 0 \/ x > 0   ~> ff
  || -eps < x /\ x < eps ~> tt ;;
let roots_interval = fun f : real -> real =>
  exists_bool_interval (fun x : real => is_0_eps (f x));;
\end{lstlisting}
\caption{A Marshall program that approximately computes whether a continuous function $f$ has roots on $[0,1]$.}
\label{marshall-roots}
\end{figure}

The Marshall functional in Fig. \ref{marshall-roots} approximately decides whether a real-valued function $f : \R \cto \R$ has roots on the interval $[0,1]$.

\section{Related work}

Several works exploit the ability to quantify over compact-overt spaces \cite{lamcra, escardo2004, escardoinfinite, simpson1998lazy} to compute values in $\R$, $\bool$, or $\Sierp$, but not partial or nondeterministic Boolean values.

$\cat{FSpc}_{nd, p}$ is equivalent to the category of suplattices (also known as complete join semilattices), which is well-studied and its relation to $\cat{FSpc}$ well-documented \cite{townsend, vickersdoublepowerlocale, elephant, convergence}. \citeauthor{topologyvialogic} defines $\cdot_\bot$ and $\PLowerP$ \citeyear{topologyvialogic}, but we are not aware of any previous explicit characterizations of these constructions as strong monads induced by adjoints to the forgetful functors to $\cat{FSpc}_p$ and $\cat{FSpc}_{nd}$, respectively.

Several works describe related programming formalisms involving continuity, partiality, and nondeterminism.
\citet{marcial-romero} define a language with real-number computation that admits nontermination and has a foundational family of functions $\mathrm{rtest}_{a,b}$ which map real numbers to nondeterministic Boolean values, with a domain-theoretic semantics that uses Hoare powerdomains (which roughly corresponds to $\PLower$). Establishing totality requires reasoning within their operational model, in contrast to our framework, which optionally has denotational semantics for total functions. 
\citet{escardorealpcf} defines a language ``Real PCF'' with a denotational semantics in terms of cpo's, in which there is an operation known as a ``parallel conditional,'' which corresponds to the internal ``or'' operation on the Sierpi\'nski space $\vee_{\Sierp} : \Sierp \times \Sierp \cto \Sierp$ in our formalism. Parallel conditionals are applied to construct deterministic functions on the real numbers, which differs from our examples, whose computations are total but nondeterministic.
Similarly, \citeauthor{tsuiki}'s work on computation with Gray-code-based real numbers \citeyear{tsuiki} is based on ``indeterministic'' computation, where potentially nonterminating computations must be interleaved, and those that terminate must agree in their answers.

We are unaware of any other notion of pattern matching that permits patterns where determining membership is undecidable, without jeopardizing totality.
\citet{mueller} describes a system for exact real arithmetic that has a datatype of ``lazy Booleans'' analogous to our partial Booleans, as well as a partial and nondeterministic $n$-ary \texttt{choose} operation on lazy Booleans.
\emph{dReal} is a tool that allows computation of approximate truth values over $\R$ \cite{dReal}, allowing order comparisons and bounded quantifiers. Our calculus of nondeterministic $\bool$-valued maps, when restricted to $\R$, provides similar computational abilities but with a different foundational framework.

\section{Conclusion}

We presented a semantic framework for principled computation with continuous values with partiality and/or nondeterminism. In each variant, pattern-matching constructs facilitate construction of programs, and the Booleans yield a formal logic for approximate decision procedures.

The programs we describe are executable, thanks to their use of constructive topology, as demonstrated by their implementation in our modified version of Marshall.

\begin{acks}                            
  We thank Tej Chajed, Gabriel Scherer, and Muralidaran Vijayaraghavan for their feedback on drafts.
  This research was supported in part by the United States Department of Energy (Grants DE-SC0014204, DE-SC0008923) and the Office of Naval Research (Grant N00014-17-1-2699).
\end{acks}

\bibliography{LICS}

\begin{longversion}
\appendix

\section{Constructive topology}
\label{c:topology}

This appendix gives further detail about constructive topology. In particular, it introduces formal topology, which is a constructive and predicative version of locale theory that is useful for describing explicit constructions on locales, and which also makes its computational content more clear.

\subsection{Formal topology}
\label{s:formal-topology}

In this section, we formally describe formal spaces, and then the inductively generated formal spaces. Each of these are locales whose structures are presented in a particular (predicative) way.

\begin{definition}
A \emph{formal space} $A$ is a preorder $\BOpen(A)$ together with a relation
$\cov : \BOpen(A) \to (\BOpen(A) \to \Prop) \to \Prop$ that satisfies the rules in Figure \ref{fig:fspc}.
\end{definition}

\begin{figure}[htbp]
   \centering
\begin{mathpar}
\inferrule*[left=refl]
  {a \in U}
  {a \cov U}

\inferrule*[left=trans]
  {a \cov U\\ U \cov V}
  {a \cov V}
  
\inferrule*[left=$\le$-left]
  {a \le b \\ b \cov U}
  {a \cov U}
  
\inferrule*[left=$\le$-right]
  {a \cov U \\ a \cov V}
  {a \cov U \omeet V}
  
\end{mathpar}
   \caption{Rules that a formal cover relation must satisfy.}
   \label{fig:fspc}
\end{figure}

The intuitive meaning of this structure is that $\BOpen(A)$ is a set of basic opens on $A$, where $a \le_{\BOpen(A)} b$ if the basic open $a$ is included in the basic open $b$. A subset $U : \BOpen(A) \to \Prop$ represents an arbitrary union of the basic opens in $U$. This means that every open of $A$ can be represented as a subset $U : \BOpen(A) \to \Prop$. Then $a \cov U$ represents the proposition that the basic open $a$ is covered by the open $U$ (that is, the extent of $a$ is included in the extent of $U$, though it need not be the case that directly $a \in U$).

\grammar{This} can be extended to specify the entire covering relation, that is, when an open $U$ is covered by another open $V$. We define (overloading the notation $\cov$ to have either basic opens or opens on the left side)\commentt{
\ConsiderRemoving{For a subset $U : X \to \Prop$, the notation $\forall x \in U.\ P(x)$ is shorthand for $\forall x : X.\ x \in U \to P(x)$.}}
\[
U \cov V \triangleq \forall a \in U.\ a \cov V,
\]
which says that $U$ is covered by $V$ if all of the basic opens comprising $U$ are covered by $V$. With this definition, we can define the (large) set of opens on $A$, $\Open{A}$, as the type $\BOpen(A) \to \Prop$ where two opens $U, V : \BOpen(A) \to \Prop$ are equivalent if they each cover each other, i.e., $U \cov V$ and $V \cov U$.

The intersection of two basic opens $a$ and $b$ need not necessarily be a basic open itself; however, it must necessarily be the union of basic opens $c$ which are included in both $a$ and $b$. Therefore, we define the operator\footnote{
For any type $A$ and $x : A$, the notation $\{ x \}$ denotes the singleton subset, $\fun{y}{x \inteq y}$.}
\footnote{
There are many analogies between the construction of the topos of sheaves on a site and the construction of the lattice of opens of a formal space. The definition of $\omeet$ operator is just a special case of the construction of the product of representable presheaves, where the presheaves are over the category $\BOpen(A)$ (considering a preorder as a category). We will refrain from pointing out the other various correspondences.
}
\begin{align*}
\omeet &: \BOpen(A) \to \BOpen(A) \to (\BOpen(A) \to \Prop)
\\ a \omeet b &\triangleq \bigcup \{c : A \mid c \le a \text{ and } c \le b \}.
\end{align*}

The $\omeet$ operator can be extended to apply to opens just by taking unions, i.e.,
\[
U \omeet V \triangleq \bigcup_{a \in U} \bigcup_{b \in V} a \omeet b.
\]

Having defined $\omeet$, we can return to understand the rules on a formal cover described in Figure \ref{fig:fspc}. \irule{refl} says that a basic open is covered by an open that includes that basic open. \irule{$\le$-left} is similar, saying that if $U$ covers a basic open $b$, then it covers a basic open $a$ that is included in $b$. \irule{trans} describes transitivity of the cover relation: if $a$ is covered by $U$, and $U$ is covered by $V$, then $a$ is covered by $V$. \irule{$\le$-right} says that if a basic open $a$ is covered by two opens $U$ and $V$, it must be covered by their intersection.

We should emphasize that these definitions and rules are purely formal: basic opens are just ``symbols'' in some sense, rather than subsets as they are in classical topology. The rules should match spatial intuition but are not justified by it. The rules can be understood from a \emph{computational} perspective as well, which is most apparent when considering the definition of a point in a formal space.

\begin{definition}
A \emph{point} $x$ of a formal space $A$ is a subset $(x \liesin \cdot) : \BOpen(A) \to \Prop$ (read ``$x$ lies in'') satisfying the rules in Figure \ref{fig:fspc-pt}. The points of $A$ form a set $\mathsf{Pt}(A)$ where two points $x$ and $y$ are considered equivalent if for all $a : \BOpen(A)$, $x \liesin a$ if and only if $y \liesin a$.
\end{definition}

\begin{figure}
\centering
\begin{mathpar}
\inferrule*[right=join]
  {x \liesin a \\ a \cov U}
  {x \liesin U}

\inferrule*[right=meet-0]
  { }
  {x \liesin \top}
  
\inferrule*[right=meet-2]
  {x \liesin a \\ x \liesin b}
  {x \liesin a \omeet b}
\end{mathpar}
   \caption{Rules that a point of a formal space must satisfy.}
   \label{fig:fspc-pt}
\end{figure}

As with the $\omeet$ operator, we extend $x \liesin \cdot$ to operate on opens, rather than just basic opens, by taking a union,\commentt{
\ConsiderRemoving{For a subset $U : X \to \Prop$, the notation $\exists x \in U.\ P(x)$ is shorthand for $\exists x : X.\ x \in U \wedge P(x)$.}}
\footnote{The existential quantifier in this definition should be viewed as a dependent pair, where it is possible to \emph{compute} the basic open that exists. This definition agrees with formal topology as formalized in Martin-L\"of type theory but differs from Maietti and Sambin \cite{whypointfree}, who use an existential quantifier that does not guarantee a computational interpretation.
}
\[
x \liesin U \triangleq \exists a \in U.\ x \liesin a.
\]
The symbol $\top : \Open{A}$ denotes the trivially true propositional function (i.e., subset that is the whole set), and denotes the open that represents the entire space.

Intuitively, $x \liesin a$ means that the point $x$ lies in the basic open $a$. Points are described by which opens they lie in, but not every collection of basic opens actually specifies a point. The \irule{join} rule says that if $x$ lies in $a$ and if $a$ is covered by $U$, the $x$ lies in $U$. \grammar{This} can also be read computationally. If $x \liesin a$, and we provide a cover $U$ of $a$ (imagine $U$ is very fine), then $x$ \emph{computes} some basic open $b \in U$ such that $x \liesin b$. We call this computation \emph{splitting} a point with an open cover (following Palmgren \cite{palmgren2007}). \irule{meet-0} says that a point must lie in the entire space, and \irule{meet-2} says that if $x$ lies in both $U$ and $V$, then it must lie in their intersection. Another intuitive interpretation is that \irule{join} allows one to refine information about a point, \irule{meet-0} gives us \emph{some} information about the point (so there is at least something to refine), and \irule{meet-2} allows us to combine two pieces of information into one.

\begin{definition}
The one-point space $\One$ is defined by taking $\BOpen(\One) \triangleq \One$, where $\One$ also denotes the preorder with one element $\mathsf{tt}$, and defining the cover relation
\[
\mathsf{tt} \cov U \triangleq \mathsf{tt} \in U.
\]
\end{definition}
One can confirm that this covering relation satisfies the necessary rules. Then $\Open{\One} \iso \Prop$, and therefore we can interpret the subset $(x \liesin \cdot) : \BOpen(A) \to \Prop$ corresponding to a point $x$ of a space $A$ also as an inverse image map $x^* : \BOpen(A) \to \Open{\One}$. In $\Open{\One}$, the $\omeet$ operator corresponds to conjunction of propositions, and $\cov$ to implication. We can use this correspondence to generalize the definition of points of a space to continuous maps from one space to another:

\begin{definition}
A \emph{continuous map} $f$ from a formal space $\Gamma$ to a formal space $A$, written $f : \Gamma \cto A$, is a map $f^* : \BOpen(A) \to \Open{\Gamma}$ satisfying the rules in Figure \ref{fig:fspc-cmap}. The continuous maps from $\Gamma$ to $A$ form a set $\Gamma \cto A$ where two maps $f$ and $g$ are considered equivalent if for all $a : \BOpen(A)$, $f^*(a) = g^*(a)$.
\end{definition}

\begin{figure}
\centering
\begin{mathpar}
\inferrule*[right=split]
  {a \cov U}
  {f^*(a) \cov f^*(U)}

\inferrule*[right=meet-0]
  { }
  {\top \cov f^*(\top)}

\inferrule*[right=meet-2]
  { }
  {f^*(a) \omeet f^*(b) \cov f^*(a \omeet b)}.
\end{mathpar}
   \caption{Rules that a continuous map of formal spaces must satisfy.}
   \label{fig:fspc-cmap}
\end{figure}

\subsubsection{Inductively generated formal spaces}

Unfortunately, in the general case, it seems impossible to form product spaces for those formal spaces that have no structure beyond their satisfaction of the laws in Figure \ref{fig:fspc}.
\grammar{This} motivates the definition of \emph{inductively generated formal spaces} by Coquand et al. \cite{coquand2003}. As a category, the inductively generated formal spaces have products and form a full subcategory of the more general class of formal spaces.

For inductively generated formal spaces, the cover relation takes a particular form: it is generated by an indexed family of axioms of the form $a \cov U$ for concrete $a$s and $U$s.

\begin{definition}
A formal space $A$ is \emph{inductively generated} if for each $a : \BOpen(A)$ there is an index type $I_a : \Type$ indexing a family of opens $C_a : I_a \to (\BOpen(A) \to \Prop)$ such that the covering relation is equivalent to the one inductively generated by the following constructors (which always satisfy the rules in Figure \ref{fig:fspc}):
\begin{mathpar}
\inferrule*[right=refl]
  {a \in U}
  {a \cov U}

\inferrule*[right=$\le$-left]
  {a \le b
  \\ b \cov U}
  {a \cov U}

\inferrule*[right=cov-axiom]
  {i : I_b
  \\ a \le b
  \\ a \omeet C_b(i) \cov U}
  {a \cov U}.
\end{mathpar}
\end{definition}

The covering relation generated by an axiom set $(I, C)$ is the \emph{least} covering relation satisfying $a \cov C_a(i)$ for all $a : \BOpen(A)$ and $i : I_a$ \cite{coquand2003}. The \irule{cov-axiom} rule at once ensures that the required covers are present and that \irule{trans} and \irule{$\le$-right} hold.

Most importantly, inductively generated formal spaces have products \cite{coquand2003} and pullbacks \cite{pullback} (whereas formal spaces in general may not). Coquand et al. and Vickers demonstrate inductive generation of all spaces used in this report \cite{coquand2003, vickersmetric, vickersdoublepowerlocale, sublocFT, vickersconnected}.
These constructions are critical in enabling computation over these spaces.

\subsubsection{Locales}

Formal topology describes spaces with a particular base of opens, but often it is easier to express certain constructions instead with \emph{locale theory}, where a space is described without reference to a particular base. Every formal space determines a locale, but (predicatively) one cannot in general construct a formal space from a locale.

There are really two key reasons (which are related) why we deal with formal topology at all, rather than exclusively using locales:
\begin{itemize}
\item Formal topology is more concrete and, in practice, is useful for giving concrete methods of constructing spaces (similar to the use of generators and relations for presenting locales).
\item Predicatively, a general construction of product spaces and pullbacks does not appear to be possible for locales \cite{coquand2003}. These can only be constructed for (inductively generated) formal spaces.
\end{itemize}

Once a space has been constructed using formal topology, it is sometimes convenient to shift to the language of locale theory, where no fundamental distinction is drawn between the basic opens and opens in general.

\begin{definition}
A \emph{locale} $A$ is a distributive lattice $\Open{A}$ that has top and bottom elements, $\top$ and $\bot$, respectively, and that has small joins such that meets distribute over joins:
\footnote{Having small joins means having $I$-indexed joins for any small type $I : \Type$. Any lattice of opens $\mathcal{O}(A)$ where $\top \neq \bot$ is necessarily large.
Since locale theory is generally impredicative, care must be taken with defining the predicative analogue presented here. Palmgren \cite{palmgren2003} offers a more careful treatment of predicativity and universes in formal topology.
}
\[
a \wedge \bigvee_{i : I} b_i = \bigvee_{i : I} a \wedge b_i.
\]
\end{definition}

\begin{theorem}
For any formal space $A$, the preorder $\Open{A}$ of opens (defined in \refsection{s:formal-topology}) yields a locale.
\end{theorem}
\begin{proof}[Proof sketch]
We need to show $\Open{A}$ has the requisite operations. We define
\begin{align*}
   \top &\triangleq \fun{\_}{\top}
\\ \bot &\triangleq \fun{\_}{\bot}
\\  U \wedge V &\triangleq U \omeet V
\\ \bigvee_{i : I} U_i &\triangleq \bigcup_{i : I} U_i.
\end{align*}
One can confirm that these operations are well-defined (in fact, monotone) with respect to the preorder on $\Open{A}$ and that they satisfy the requisite algebraic laws.
\end{proof}

We call the lattice $\Open{A}$ the \emph{opens} of $A$, which describes the observable or ``affirmable'' properties of $A$ \cite{topologyvialogic}. If $U \le \bigvee_{i : I} V_i$, we call the family $(V_i)_{i : I}$ an \emph{open cover} of $U$.

We next define the points and continuous maps, which are closely analogous to the version for formal spaces.

\begin{definition}
A \emph{point} $x$ of a space $A$ is a subset $x \liesin \cdot : \Open{A} \to \Type$ (read ``$x$ lies in'') such that
\begin{mathpar}
\inferrule*[right=join]
  {x \liesin U \\ U \le \bigvee_{i : I} V_i}
  {\exists i : I.\ x \liesin V_i}

\inferrule*[right=meet-0]
  { }
  {x \liesin \top}

\inferrule*[right=meet-2]
  {x \liesin U \\ x \liesin V}
  {x \liesin U \wedge V}.
\end{mathpar}
\end{definition}

By the same reasoning as with formal spaces, we can rephrase these rules so that they suggest the rules for continuous maps:
\begin{mathpar}
\inferrule*[right=join]
  {U \le \bigvee_{i : I} V_i}
  {x^*(U) \le x^* \left( \bigvee_{i : I} V_i \right)}
\\
\inferrule*[right=meet-0]
  { }
  {\top \le x^*(\top)}

\inferrule*[right=meet-2]
  { }
  {x^*(U) \wedge x^*(V) \le x^*(U \wedge V)}.
\end{mathpar}

However, since \irule{join} implies that $x^*$ preserves small joins and is monotone, this means that the reversed inequalities in the \irule{meet} rules must hold, and so we can simplify these rules to a simpler and more algebraic definition:

\begin{definition}
A \emph{continuous map} $f : A \cto B$ between locales is a map $\iimg{f} : \Open{B} \to \Open{A}$ that preserves $\top$, binary meets, and small joins. The map $\iimg{f}$ \glsadd{iimg} is known as the \emph{inverse image map}.
\end{definition}

A continuous map $f : A \cto B$ transforms covers on $B$ into covers on $A$.
Locales and continuous maps form a category.

\begin{definition}
Two spaces $A$ and $B$ are \emph{homeomorphic}, written $A \iso B$, if they are isomorphic in the category of continuous maps, that is, if there are continuous maps $f : A \cto B$ and $g : B \cto A$ such that $g \circ f = \mathsf{id}_A$ as well as $f \circ g = \mathsf{id}_B$.
\end{definition}

\begin{definition}[Specialization order]
For two maps $f, g : A \cto B$, define $f \le g$ if for every $U : \Open{B}$, $\iimg{f}(U) \le \iimg{g}(U)$.
\end{definition}

Whenever it is the case that $f \le g$,$g$ can always behave as $f$. For instance, if $f$ and $g$ are points, then given an open cover $\top \le \bigvee_{i : I} U_i$, there is some $i^* : I$ such that $f \liesin U_{i^*}$. If $f \le g$, then $g \liesin U_{i^*}$ as well.

\begin{definition}
An open $U : \Open{A}$ is \emph{positive}, written $\mathsf{Pos}(U)$ if every open cover of $U$ is nonempty, i.e., if for every small family $(V_i)_{i : I}$ such that $U \le \bigvee_{i : I} V_i$, $I$ is inhabited.
\end{definition}

Positivity encodes the notion of being ``strictly bigger than $\bot$,'' and assuming classical logic, $\mathsf{Pos}(U)$ is equivalent to $\neg (U \le \bot)$.

\begin{lemma}
\label{thm:join-pos}
If $f : \Open{A} \to \Open{B}$ preserves joins, then for $U : \Open{B}$, $\mathsf{Pos}(f(U))$ implies $\mathsf{Pos}(U)$.
\end{lemma}
\begin{proof}
Suppose $U \le \bigvee_{i : I} V_i$. Since $f$ preserves joins, it is monotone, so
\[
f(U) \le f \left( \bigvee_{i : I} V_i \right)
= \bigvee_{i : I} f(V_i).
\]
Since $f(U)$ is positive, $I$ is inhabited by the above cover, and thus $U$ is positive.
\end{proof}

\subsubsection{The computational content of formal topology}

A continuous map $f : \Gamma \cto A$ (with a point of a space a special case, where $\Gamma \iso \One$) is defined by two pieces of data: its inverse image map $f^* : \Open{A} \to \Open{\Gamma}$ and a formal proof that $f^*$ preserves small joins, $\top$, and binary meets. It is reasonable to consider the inverse image map $f^*$ a \emph{specification}, as it describes the observable behavior of the output in terms of observable properties of the input. \grammar{This} accords with the fact that any two continuous maps with the same inverse image maps are considered equal; the formal proofs of structure preservation can be ignored when reasoning, except for the requirement that those formal proofs must exist. However, the formal proofs are necessary for computing concrete results, and so we can consider the formal proof that $f^*$ preserves joins and finitary meets as the \emph{implementation} of the behavior specified by $f^*$. It is remarkable, then, that the specification is ``complete,'' uniquely specifying a continuous map (by definition) if it exists.

Suppose we define a function $f : \Gamma \cto A$. The map $f$ can be applied to an input $x : \mathsf{Pt}(\Gamma)$ to produce an outpoint point $f(x) : \mathsf{Pt}(A)$. How then does one inspect $f(x)$ to get concrete results about where the output lies within $A$?

One computes concrete results in formal topology by splitting a point with an open cover. This open cover may be arbitrarily fine. By combining the \irule{join} and \irule{meet-0} rules, we observe that if we present a point $y$ in a space $A$ with an open cover of the whole space $\top \le \bigvee_{i : I} U_i$, then $y$ can \emph{compute} an $i : I$ such that $y$ lies in $U_i$. If a point lies in several opens in an open cover, it may return any index corresponding to an open it lies in. In the case of the point $f(x)$ where $f : \Gamma \cto A$, $f$ translates a cover of $A$ into a cover of $\Gamma$, from which $x$ can then compute some open from that cover of $\Gamma$ that it lies in.

For instance, a point $x : \mathsf{Pt}(\R)$, when given a formal proof of the open cover
\[
\top \le \bigvee_{b : \bool} \text{if } b \text{ then } (\cdot < 1) \text{ else } (\cdot  > -1),
\]
must either return a concrete Boolean $\btrue$ together with a proof $x \liesin \cdot < 1$ or return $\bfalse$ with a proof $x \liesin \cdot > - 1$. If a point lies in both opens, whether it returns $\btrue$ or $\bfalse$ depends on both the ``implementation'' of the point (the formal proofs that $x$ satisfies \irule{join} and \irule{meet-0}) as well as the formal proof of the cover.

In implementing a real number, one must provide a function which, given any tolerance $\varepsilon > 0$, produces a rational approximation within that tolerance. Imagine two implementations of the real number 0, where the first always returns 0, whereas the second returns $\varepsilon/2$. Consider a formal proof of the above cover of $\R$ that proceeds in the following manner: first, approximate to a tolerance of 1 with the cover $\top \le \bigvee_{q : \rat} B_1(q)$. Then, for a given $q : \rat$, if the lower endpoint $q - 1$ of the resulting approximating interval $B_1(q)$ is at most 0, then prove that $B_1(q) \le (\cdot < 1)$, and if not, prove $B_1(q) \le (\cdot > -1)$. With this particular formal covering proof, the first implementation of 0 will return $\btrue$, $(\cdot < 1)$, but the second will return $\bfalse$, $(\cdot > -1)$. However, if this cover is applied to the real number 2, the observation will be $\bfalse$, no matter the implementation of 2 or the particular proof of the covering. 

Note that the \irule{meet-2} rule does not figure into this notion of computation. We can think of maps $x \liesin \cdot : \Open{A} \to \Open{\One}$ that satisfy \irule{join} and \irule{meet-0} but not \irule{meet-2} as nondeterministic values (see \refsection{parnondet}). But there is a notion of incremental computation where \irule{meet-2} has a computational meaning.

Suppose we probe a point $x : \mathsf{Pt}(A)$ with a cover $\top \le \bigvee_{i : I} U_i$ and learn that $x \liesin U_{i_*}$ for some particular $i_* : I$. Then we may decide to further refine our knowledge of $x$ by probing with another open cover $U_{i_*} \le \bigvee_{j : J} V_j$, or equivalently, $U_{i_*} \le \bigvee_{j : J} U_{i_*} \wedge V_j$, to learn that $x \liesin U_{i_*} \wedge V_{j_*}$ for some $j_* : J$. \grammar{This} gives a sort of ``sequential composition'' of splitting with covers. The \irule{meet-2} rule allows two independent threads of computation to be joined. For instance, we might probe $x$ with an open cover to learn $x \liesin U$, and independently with another open cover (of the whole space, not just $U$) to learn that $x \liesin V$. By \irule{meet-2}, we can combine what we've learned to determine $x \liesin U \wedge V$, and can continue to refine where $x$ is with covers of $U \wedge V$. For nondeterministic values, where \irule{meet-2} may not hold, \grammar{this} isn't satisfied, since (informally) each independent thread of refinement may get a different nondeterministic realization.

\subsection{Spaces}

\subsubsection{Discrete spaces}

The simplest kinds of spaces are those that just represent sets. Let \cat{Set} denote the category of sets. The objects of \cat{Set} are types $A$ together with an equivalence relation $=_A : A \to A \to \Prop$ on $A$, and the arrows $f : A \to B$ are those equivalence-preserving functions, that is, functions $f : A \to B$ on the underlying types such that for all $a, a' : A$, if $a =_A a'$, then $f(a) =_B f(a')$.

We will show that it is possible to work with sets within the framework of spaces. Precisely, we can define a full and faithful functor $\functor{Discrete} : \cat{Set} \to \cat{FSpc}$ that exhibits \cat{Set} as the discrete spaces, which form a full subcategory of \cat{FSpc}. Given a set $A$ (with equivalence relation $=_A$), we construct a formal space whose type of basic opens is $A$, with the inclusion preorder given by $=_A$. We think of a basic open $a : A$ as representing the subset $(\cdot =_A a)$ of $A$, which is open in the discrete topology (and hence the preorder defining inclusion of basic opens is discrete). Since every element $a : A$ represents both an open and a point of the the discrete space, for $a : A$ we will use the notation $\cdot = a$ to refer to the open, reserving $a$ for the point of the space.

In $\functor{Discrete}(A)$, we have $a \cov U$ if and only if there is some $a' : A$ such that $a =_A a'$ and $a' \in U$. We can use this to simplify the rules that a point $x$ of a discrete space or a continuous map $f : \Gamma \cto A$ between discrete spaces must satisfy:

\begin{mathpar}
\inferrule*[right=split]
  {x \liesin a \\ a =_A b}
  {x \liesin b}

\inferrule*[right=meet-0]
  { }
  {\exists a : A.\ x \liesin a}

\inferrule*[right=meet-2]
  {x \liesin a \\ x \liesin b}
  {a =_A b}
\end{mathpar}

\begin{mathpar}
\inferrule*[right=split]
  {a =_A b}
  {f^*(a) =_{\Gamma \to \Prop} f^*(b)}

\inferrule*[right=meet-0]
  { }
  {\exists a : A.\ \gamma \in f^*(a)}

\inferrule*[right=meet-2]
  {\gamma \in f^*(a) \\ \gamma' \in f^*(b) \\ \gamma =_\Gamma \gamma'}
  {a =_A b}.
\end{mathpar}

We observe that $f^*$ identifies a relation between $\Gamma$ and $A$. We read $\gamma \in f^*(a)$ as ``$\gamma$ is in the preimage of $a$ under $f$'' (which is equivalent to saying that $f$ maps $\gamma$ to $a$). The \irule{join} rule says that this relation respects equality on $A$. \irule{meet-0} says that the relation is total (from $\Gamma$ to $A$), and \irule{meet-2} says that the relation respects equality on $\Gamma$ and also that each input maps to at most one output (up to $=_A$). Accordingly, $f^*$ is a functional relation from the set $\Gamma$ to the set $A$, so it is in bijective correspondence with $\Gamma \to_\cat{Set} A$, the collection of functions from the set $\Gamma$ to the set $A$, meaning that the functor $\functor{Discrete}$ is full and faithful: continuous maps between discrete spaces are just functions on their underlying sets.

Recall the functor $\functor{Pt} : \cat{FSpc} \to \cat{Set}$ which takes a space $A$ to its (large) set of points $\functor{Pt}(A)$, where two points are considered equal if they lie in the same basic opens. This is right adjoint to $\functor{Discrete}$, i.e., $ \functor{Discrete} \dashv \functor{Pt}$, giving a correspondence for a set $A$ and a space $B$
\begin{mathpar}
\mprset{fraction={===}}
\inferrule*{A \to_\cat{Set} \functor{Pt}(B)}
       {\functor{Discrete}(A) \cto B}.
\end{mathpar}

\subsubsection{Open subspaces}

Given a space $A$ and an open $U : \Open{A}$, we can form the open subspace $\{ A \mid U \}$ of $A$ by making $\Open{\{A \mid U \}}$ a quotient of $\Open{A}$, identifying opens $P, Q : \Open{A}$ in $\{ A \mid U \}$ when $P \wedge U = Q \wedge U$. The quotient still defines a space, since the operation $\cdot \wedge U : \Open{A} \to \Open{A}$ preserves binary meets and small joins.

\subsubsection{The Sierpi\'nski space}

The Sierpi\'nski space $\Sierp$ is fundamental in topology, defining the space of possible ``truth values.'' Just as the subsets of a set $S$ are in correspondence with functions $S \to \Prop$, the opens of a space $A$ are in correspondence with the continuous maps $A \cto \Sierp$ \cite{topologyvialogic},
\[
\Open{A} \iso (A \cto \Sierp).
\]

The basic opens of $\Sierp$ are $\BOpen(\Sierp) \triangleq \bool_\le$, 
where $\bool_\le$ is $\bool$ with the ``truth order,'' i.e., where $\bfalse$ is strictly less than $\btrue$.
The Sierpi\'nski space has no covering axioms. We have thus defined $\Sierp$.

In particular, we have 
\[
(\One \cto \Sierp) \iso \Open{\One} \iso \Prop,
\]
so the points of $\Sierp$ are just the propositions. Given any point $x$ of $\Sierp$, the corresponding proposition is $x \liesin \bfalse$. In particular, there are points $\strue$ and $\sfalse$, where $\strue \liesin \bfalse$ but not $\sfalse \liesin \bfalse$. It is convenient to denote the nontrivial open of $\Sierp$ by $\cdot = \strue$, since the only global point lying in that open is $\strue$.

\subsubsection{Disjoint unions (sums)}

We have not come across any characterization of sum spaces in $\cat{FSpc}$, so we describe them here.

From a family of spaces $(A_i)_{i : I}$ parameterized over some index type $I$\footnote{
In this section, we require the index type $I$ to satisfy uniqueness of identity proofs (UIP), which states that
\[
\forall a, b : I.\ \forall p, q : a \inteq b.\ p \inteq q.
\]}, we can form their disjoint union space $\sum_{i : I} A_i$, which is the coproduct of the $A_i$s in \cat{FSpc}. Intuitively, $\sum_{i : I} A_i$ pastes all the $A_i$s together, where points from different spaces are not considered near each other. When $I \inteq \bool$ (as types), this specializes to binary sums, which we denote with $+$. For instance, we have the homeomorphism in \cat{FSpc}
\[
\bool \iso \One + \One,
\]
and more generally, for any type $I$ (considered as a set with $\inteq$ as its equivalence relation), 
\[
\functor{Discrete}(I) \iso \sum_{i : I} \One.
\]

The preorder of basic opens of $\sum_{i : I} A_i$ is the coproduct of the basic opens of the constituent spaces,
\[
\BOpen\left(\sum_{i : I} A_i\right) \triangleq \sum_{i : I} \BOpen(A_i).
\]
The preorder relation for the coproduct preorder $\sum_{i : I} \BOpen(A_i)$ is generated as the inductive type with the single constructor
\begin{mathpar}
\inferrule{a \le_{A_i} b}
              {(i, a) \le_{\sum_{j : I} A_j} (i, b)}.
\end{mathpar}

The axioms for $\sum_{i : I} A_i$ are then just a sort of ``coproduct'' of the axioms of the constituent spaces. For a basic open $(j, a) : \BOpen(\sum_{i : I} A_i)$, for each axiom $a \cov_{A_j} U$, we add the axiom
\[
(j, a) \cov \mathsf{InDisjunct}_j(U),
\]
where $\mathsf{InDisjunct}_j : \Open{A_j} \to \Open{\sum_{i : I} A_i}$ is inductively generated by the constructor
\begin{mathpar}
\inferrule{a \in U}
              {(j, a) \in \mathsf{InDisjunct}_j(U)}.
\end{mathpar}

\begin{theorem}
\label{thm:inj-oe}
\[
A_j \iso \left\{ \sum_{i : I} A_i \mid \mathsf{InDisjunct}_{j}(\top) \right\}
\]
\end{theorem}
\begin{proof}[Proof sketch]
Define $f : A_j \cto \left\{ \sum_{i : I} A_i \mid \mathsf{InDisjunct}_{j}(\top) \right\}$ to have its inverse image map $f^* : \Open{\left\{ \sum_{i : I} A_i \mid \mathsf{InDisjunct}_{j}(\top) \right\}} \to \Open{A_j}$ as the inductive family generated by the constructor
\begin{mathpar}
\inferrule*
  {a =_{A_j} a'}
  {a' \in f^*(j, a)}.
\end{mathpar}
Note that with this definition, $f^*(i, b) = \bot$ whenever $i \not\inteq j$.

Define $g : \left\{ \sum_{i : I} A_i \mid \mathsf{InDisjunct}_{j}(\top) \right\} \cto A_j$ by
\[
g^*(a) \triangleq (j, a).
\]

It should be clear that $f^*$ and $g^*$ are inverses of each other, so we must just confirm that $f^*$ and $g^*$ indeed define continuous maps, which is mostly a calculational matter. For $f$, \irule{join} is straightforward. We only note that proving that $g$ satisfies \irule{meet-0} reduces to
\[
\top \cov \mathsf{InDisjunct}_{j}(\top),
\]
which is exactly the additional covering rule given by the open subspace of $f$'s output.
\end{proof}

Once open embeddings have been defined (\refsection{s:open-embedding}), this theorem will establish the open embedding $\mathsf{inj}_j : A_j \hookto \sum_{i : I} A_i$.

\subsubsection{Products}

We next introduce notation and relevant properties for the
construction of product spaces in formal topology\footnote{\citet{vickerstychonoff} provides a full definition and characterization.}.

For a family of spaces $(A_i)_{i : I}$ parameterized over some index type $I$\footnote{
Again we require that $I$ satisfy UIP.}, we denote their product in \cat{FSpc} by $\prod_{i : I} A_i$. The key idea characterizing the structure of the product space is that we can make an observation on the product space by choosing a component and making an observation on that component. Since we can only make finitely many observations, any open will make non-trivial observations on only finitely many components.

More precisely, the sub-basic opens of $\prod_{i : I} A_i$ are also given by the coproduct preorder $\sum_{i : I} \BOpen(A_i)$. However, we will use the notation $[i \mapsto a]$ for $i : I$ and $a : \BOpen(A_i)$ to refer to such a sub-basic open, because it represents the idea that the $i$th component lies in $a$. Similarly, we can define the open $[i \mapsto U] : \Open{\prod_{i : I} A_i}$ for an arbitrary open $U : \Open{A_i}$ as inductively generated by the constructor
\begin{mathpar}
\inferrule*
  {a \in U}
  {[i \mapsto a] \in [i \mapsto U]}.
\end{mathpar}

The universal property of product spaces says that for any space $\Gamma$ and maps $f_i : \Gamma \cto A_i$, we can construct a continuous map $g : \Gamma \cto \prod_{i : I} A_i$. The inverse image map acts according to
\footnote{It suffices to only define $g^*$ on its sub-basic opens ($\mathcal{O}_\mathsf{B}\left(\prod_{i : I} A_i\right)$) since $g^*$ will be required to preserve finitary meets, and the basic opens are just finitary meets of sub-basic opens.}
\begin{align*}
g^* &: \mathcal{O}_\mathsf{SB}\left(\prod_{i : I} A_i\right) \to \Open{\Gamma}
\\ g^*([i_k \mapsto a_k]) &\triangleq f_{i_k}^*(a_k).
\end{align*}

If we consider the index type $I \inteq \bool$, we get the binary products, which we denote $A \times B$ for spaces $A$ and $B$. We notice that every basic open of $A \times B$ is equivalent to one of the form
\[
[\btrue \mapsto U] \wedge [\bfalse \mapsto V]
\]
for $U : \Open{A}$ and $V : \Open{B}$, and so we may use the notation $U \times V : \BOpen\left(A \times B\right)$ to represent a basic open of $A \times B$ (also known as an \emph{open rectangle}). Accordingly, every open of $A \times B$ can be represented as a union of open rectangles.

Given $f : \Gamma \cto A$ and $g : \Gamma \cto B$, we denote by $\langle f , g \rangle : \Gamma \cto A \times B$ the ``pair'' given by the universal property of products.

\subsubsection{Metric spaces, including $\R$}
\label{s:space:metric}

It is possible to extend a set with a metric defined on it (such as $\rat$) to a metrically complete (i.e., Cauchy complete) formal space (such as $\R$) \cite{vickersmetric}. This section describes this construction.

Suppose we are given a \emph{metric set}, a set $X$ with a distance metric relation $d : \rat^+ \times X \times X \to \Prop$, where $d(\varepsilon, x, y)$ indicates the proposition that the distance between $x$ and $y$ is at most $\varepsilon$. The predicate $d$ must in fact define the closed-ball relation for a metric, meaning it must satisfy the following rules:\footnote{
This definition of a closed-ball relation is due to O'Connor \cite{oconnor2008}. We use it so that our Coq library is compatible with the Coq Repository at Nijmegen (CoRN), which has many constructive results regarding metric spaces.}
\begin{mathpar}
\inferrule*[right=refl]
  { }
  {d(\varepsilon, x, x)}
  
\inferrule*[right=sym]
  {d(\varepsilon, x, y)}
  {d(\varepsilon, y, x)}

\inferrule*[right=triangle]
  {d(\delta, x, y)
  \\ d(\varepsilon, y, z)}
  {d(\delta + \varepsilon, x, z)}

\inferrule*[right=closed]
  {\forall \delta : \rat^+.\ d(\varepsilon + \delta, x,  y)}
  {d(\varepsilon, x, y)}.
\end{mathpar}

The basic opens of the metric completion of $X$ (if $X \inteq \rat$, this would be $\R$) will be the formal balls, the type $\mathsf{Ball}(X)$ defined by the single constructor
\begin{mathpar}
\inferrule*
  {\varepsilon : \rat^+
  \\ x : X}
  {B_\varepsilon(x) : \mathsf{Ball}(X)}.
\end{mathpar}

We define a relation $<$ on balls that indicates when one ball contains another with ``room to spare'' all around, and then a relation $\le$ on balls indicating containment of one formal ball within another:
\begin{align*}
B_\delta(x) < B_\varepsilon(y) &\triangleq \exists \gamma : \rat^+.\ d(\gamma, x, y) \text{ and } \gamma + \delta < \varepsilon
\\ B_\delta(x) \le B_\varepsilon(y) &\triangleq \forall \gamma : \rat^+.\ B_\delta(x) < B_{\varepsilon + \gamma}(y).
\end{align*}

We then define the metric completion space $\mathcal{M}(X)$ of $X$ as having the basic opens $\mathsf{Ball}(X)$, and the following axioms:
\begin{mathpar}
\inferrule*[right=approx]
  {\varepsilon : \rat^+}
  {\top \cov \{ B_\varepsilon(q) \mid q : \rat \}}

\inferrule*[right=shrink]
  { }
  {B_q(\varepsilon) \cov \{ B_{\varepsilon'}(q') \mid B_{\varepsilon'}(q') < B_\varepsilon(q) \}}
\end{mathpar}

\begin{definition}
A function $f : X \to Y$ on metric sets $X$ and $Y$ is \emph{$k$-Lipschitz} (for $k : \rat^+$) if for all $x, y : X$ and $\varepsilon : \rat^+$, if $d(\varepsilon, x, y)$ then $d(k \varepsilon, f(x), f(y))$.
\end{definition}

We proved in Coq that $f$ can be extended to a continuous map $g : \mathcal{M}(X) \cto \mathcal{M}(Y)$ defined by the inverse image map\footnote{
This definition of the inverse image map is based on a theorem of Vickers (Theorem 4.9 and Remark 4.10 in \cite{vickersmetric}) extending 1-Lipschitz functions to their metric completions. Vickers omits the proof that the inverse image map defines a continuous map, saying it is ``routine to check.''}
\begin{align*}
g^* : \mathsf{Ball}(Y) \to \Open{\mathcal{M}(X)}
\\ g^*(b) \triangleq \{ B_\delta(x) \mid B_{k \delta}(f(x)) < b \}.
\end{align*}

The fact that $f$ is $k$-Lipschitz implies that for balls $a, a' : \mathsf{Ball}(X)$, if $a \le a'$ and $a' \in g^*(b)$, then $a \in g^*(b)$.

\begin{theorem}
$g^*$ preserves joins, $\top$, and binary meets.
\end{theorem}
\begin{proof}[Proof sketch]
\begin{description}
\item[\irule{meet-0}] For every formal ball $B_\delta(x)$ of $X$, there is a ball of $Y$, $B_{\delta + k\delta}(f(x))$, such that $B_\delta(x) \le g^*(B_{\delta + k\delta}(f(x)))$. Thus $\top \le g^*(\top)$.
\item[\irule{meet-2}] We are given a ball of $Y$, $B_{k \delta}(f(x))$, that lies strictly within two other balls of $Y$, 
$B_{k \delta}(f(x)) < B_\varepsilon(y)$ and $B_{k \delta}(f(x)) < B_{\varepsilon'}(y')$. We must prove in $X$ that
\[
B_\delta(x) \cov g^*(B_\varepsilon(y) \omeet B_{\varepsilon'}(y')).
\]
By the properties of $<$, we can shrink each of the larger balls by just a bit, $\gamma : \rat^+$ while maintaining strict containment, so we get
$B_{k \delta}(f(x)) < B_{\varepsilon - \gamma}(y)$ and
$B_{k \delta}(f(x)) < B_{\varepsilon' - \gamma}(y')$.
Then for some even smaller tolerance $\gamma' < \gamma$, we use the covering axiom $\irule{approx}(\gamma'/k)$ to derive
\[
B_\delta(x) \cov \{ B_{\gamma'/k}(x') \mid x' : X \} \omeet B_\delta(x).
\]
This reduces our problem to showing that any ball of $X$ with radius at most $\gamma' / k$ that is contained within $B_\delta(x)$ is covered by $g^*(B_\varepsilon(y) \omeet B_{\varepsilon'}(y'))$,
which in fact follows by reflexivity, which we will now show.
We must only consider the ``worst case,'' where we have a ball $B_{\gamma'/k}(z)$ such that 
$B_{\gamma'/k}(z) \le B_{\delta}(x)$. We will prove $B_{\gamma'/k}(z) \in g^*(B_\varepsilon(y) \omeet B_{\varepsilon'}(y'))$. We do so in two steps, showing that
$B_{\gamma'/k}(z) \in g^*(B_{\gamma}(f(z)))$ and that
$B_{\gamma}(f(z)) \in B_\varepsilon(y) \omeet B_{\varepsilon'}(y')$. The former is almost immediate. 
The latter is surprisingly intricate, depending on the fact that $f$ is $k$-Lipschitz as well as using the triangle inequality.
\item[\irule{join}] We first prove that $g^*$ preserves unary covers, and then proceed by induction on covering axioms. Preserving unary covers reduces to proving that, given balls $b \le b'$ in $Y$ that $g^*(b) \le g^*(b')$, which is follows from the fact that if $a < b$ and $b \le b'$, then $a < b'$.

We now proceed by induction on the covering axioms. 

In the case of $\irule{approx}(\varepsilon)$, given balls $B_\delta(x) : \mathsf{Ball}(X)$ and $B_\gamma(y) : \mathsf{Ball}(Y)$ such that $B_\delta(x) \in g^*(B_\gamma(y))$ show that
\[
B_\delta(x) \cov g^*(B_\gamma(y) \omeet \{ B_\varepsilon(y') \mid y' : Y \}).
\]
To do so, we can find some $\alpha : \rat^+$ such that $\alpha < \varepsilon / k$, and cover $B_\delta(x)$ with balls of this radius,
\[
B_\delta(x) \cov B_\delta(x) \omeet \{ B_\alpha(x') \mid x' : X \}.
\]
Then 
\[
B_\delta(x) \omeet \{ B_\alpha(x') \mid x' : X \} \cov 
  g^*(B_\gamma(y) \omeet \{ B_\varepsilon(y') \mid y' : Y \})
\]
follows by reflexivity.

In the case of $\irule{shrink}$, the covering in fact follows from reflexivity.
\end{description}
\end{proof}

The extension of Lipschitz functions from metric sets to their metric completions can be used to define addition ($+ : \R \times \R \cto \R$), for instance, in conjunction with facts relating products of spaces and of metric sets, such as \cite{vickersmetric}
\[
\mathcal{M}(X \times Y) \iso \mathcal{M}(X) \times \mathcal{M}(Y).
\]

Since $\times : \R \times \R \cto \R$ is not Lipschitz, it cannot be defined using the above construction. However, since it is locally Lipschitz, it can be defined as a ``gluing'' of many Lipschitz maps defined on open subspaces. One can define this gluing with an overlapping pattern match to produce a nondeterministic result, and then prove that the result is in fact deterministic, since any pair of overlapping maps will agree on their overlapping region.

In $\R$, open intervals and open ``rays'' (e.g., $(0 < \cdot)$) are evidently open, as they can be described as unions of open balls. For instance, we define
\begin{align*}
(0 < \cdot) &: \Open{\R}
\\ (0 < \cdot) &\triangleq \{ B_\varepsilon(x) \mid 0 \le x - \varepsilon \}.
\end{align*}

In $\R \times \R$, the relations $(<), (\neq) : \Open{\R \times \R}$ are open. We can define $<$ by
\begin{align*}
(<) \triangleq \{ (B_\varepsilon(x), B_\delta(y) \mid x + \varepsilon \le y - \delta \}
\end{align*}
and then $\neq$ in terms of that.

\subsection{As a programming language}
\label{s:contpl}

Since \cat{FSpc} is a cartesian monoidal category (i.e., it has well-behaved products), it admits a restricted $\lambda$-calculus syntax that compiles to continuous maps, which this report uses freely.

Following \citet{escardo2004}, we describe the syntax as an inductive family $\vdash : \mathsf{list}(\cat{FSpc}) \to \Type$ with constructors

\begin{mathpar}
\inferrule*[right=var]
  {(x : X) \in \Gamma}
  {\Gamma \vdash x : X}

\inferrule*[right=app]
  {f : A_1 \times \cdots \times A_n \cto B
  \\ \Gamma \vdash x_1 : A_1 \\ \cdots \\ \Gamma \vdash x_n : A_n}
  {\Gamma \vdash f(x_1, \ldots, x_n) : B},
\end{mathpar}
where in the \irule{var} rule $(x : X) \in \Gamma$ denotes a witness that $X$ is a member of the list $\Gamma$. Let $\mathsf{Prod} : \mathsf{list}(\cat{FSpc}) \to \cat{FSpc}$ compute the product of a list of types.

\begin{theorem}
Given any expression $\Gamma \vdash e : A$, we can construct a term $\denote{e} : \mathsf{Prod}(\Gamma) \cto A$, and conversely, given a function $f : \mathsf{Prod}(\Gamma) \cto B$, there is an expression $\Gamma \vdash e : B$.
\end{theorem}
\begin{proof}
First, we construct a continuous map from syntax by induction on its structure. In the \irule{var} case, we compute from the witness $(x : X) \in \Gamma$ a projection $\denote{x} : \mathsf{Prod}(\Gamma) \cto X$. In the \irule{app} case, we are given maps $\denote{x_i} : \Gamma \cto A_i$ for $i \in 1, \ldots n$. We use these maps to build a map $x : \Gamma \cto A_1 \times \cdots \times A_n$ by the universal property for products, i.e., $x \triangleq \langle \denote{x_1}, \ldots, \denote{x_n} \rangle$. Then $f \circ x : \Gamma \cto B$.

Conversely, given $f : \mathsf{Prod}(\Gamma) \cto B$, we can build the term
\[
\Gamma \vdash f(\pi_1, \ldots, \pi_n) : B,
\]
where each $\pi_i$ is a variable $\Gamma \vdash \pi_i : \Gamma[i]$.
\end{proof}

This allows us to define a continuous map by introducing variables and applying continuous maps to them. For instance, we can define a map like
\begin{align*}
f &: \bool \times \R \times \R \cto \R
\\ f(x, y, z) &\triangleq \mathsf{if}(x, y, -y) + z \times z
\end{align*}
rather than having to manually arrange it as a ``linear'' composition of continuous maps. If we don't want to give a name to a function such as the one above, we may choose to write it with an ``anonymous'' lambda,
\[
\fun{(x, y, z)}{\mathsf{if}(x, y, -y) + z \times z}.
\]

\section{Monads for partiality and nondeterminism}
\label{a:monad}

Using terminology, concepts, and notation from formal topology, which is explained the section \ref{c:topology}, we will give a much more thorough characterization of the monads on \cat{FSpc} corresponding to partiality and/or nondeterminism, explaining their construction, their adjoint correspondences, and why they are strong monads.

\subsubsection*{Lifted spaces}

\newglossaryentry{Lift}
{
  name={\ensuremath{\mathord{\Downarrow}}},
  description={The direct image of an open in its lifting, $\mathord{\Downarrow} : A \to A_\bot$},
}
\newcommand{\Lift}{\gls{Lift}}

As mentioned previously, the ``forgetful'' functor from $\cat{FSpc}$ to $\cat{FSpc}_{p}$ has a right adjoint, such that there is a (strong) monad $\cdot_\bot : \cat{FSpc} \to \cat{FSpc}$ giving a correspondence
\begin{mathpar}
\mprset{fraction={===}}
\inferrule*{A \pto B}
  {A \cto B_\bot}
\end{mathpar}
between continuous maps and their partial counterparts. By this correspondence, we can think of the space $B_\bot$ as the space of partial values. We will call spaces of the form 
$B_\bot$ \emph{lifted spaces}. Vickers describes them briefly in \emph{Topology via Logic} \cite{topologyvialogic}.

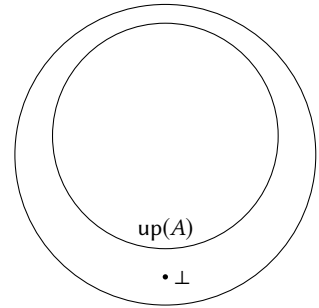
\begin{figure}[h]
   \centering
\begin{tikzpicture}[scale=0.5]
\coordinate (center) at (2,2);
\coordinate (acenter) at (2,2.5);
\coordinate (bot) at (2, -1.25);
\draw (center) circle (4);
\draw (acenter) circle (3);
\path (acenter) -- ++(270:3) node[above] {$\mathsf{up}(A)$};
\fill (bot) circle [radius=2pt];
\draw (bot) node[right] {$\bot$};
\end{tikzpicture}
   \caption{A schematic diagram of the lifted space $A_\bot$.}
   \label{fig:lifted}
\end{figure}

Figure \ref{fig:lifted} depicts a picture that in some sense represents what $A_\bot$ looks like: it has an open subspace $\mathsf{up}(A)$ (defined in \refsection{s:open-embedding}) that looks exactly like $A$, and a single point $\bot : \mathsf{Pt}(A_\bot)$ outside of that open subspace, which represents an undefined or nonterminating value.

We will now describe the construction of lifted spaces and their properties. Algebraically, we can think of the construction of lifted spaces as a free construction on the lattice of opens. Suppose we want to construct from a space $A$ its lifted space $A_\bot$. Then we want to have an operator $\Lift : \Open{A} \to \Open{A_\bot}$ which should preserve the structure from $A$ that we want to keep: joins and binary meets. We will build such an operator using the structure of $A$ as a formal space.

Since we want $\Lift$ to preserve joins, it suffices to define it for basic opens. We will add in a new top element $\mathsf{T}$ as well. Accordingly, 
the lifted space $A_\bot$ has basic opens generated by the constructors
\begin{mathpar}
\inferrule*
  {a : \BOpen(A)}
  {\Lift{a} : \BOpen(A_\bot)}

\inferrule*
  { }
  {\mathsf{T} : \BOpen(A_\bot)}
\end{mathpar}
and the preorder on the basic opens is inductively generated by the constructors
\begin{mathpar}
\inferrule*
  {a \le_A b}
  {\Lift{a} \le_{A_\bot} \Lift{b}}

\inferrule*
  { }
  {a \le_{A_\bot} \mathsf{T}}.
\end{mathpar}

We extend $\Lift$ to operate on opens of $A$ so that it preserves joins: for $U : \Open{A}$, $\Lift U$ is generated by the constructor
\begin{mathpar}
\inferrule*
  {a \in U}
  {\Lift{a} \in \Lift{U}}.
\end{mathpar}

Now we describe the covering axioms in $A_\bot$. We simply copy over the covering axioms from $A$. For each axiom $a \cov U$ in $A$, we add the axiom $\Lift{a} \cov \Lift{U}$ in $A_\bot$.

\begin{proposition}
Indeed $\mathsf{T}$ is the top element of $A_\bot$, $\top \le \mathsf{T}$.
\end{proposition}
\begin{proof}
It suffices to show that for any basic open $a : \BOpen(A_\bot)$, $a \le \mathsf{T}$, which follows directly from the definition of $\le_{A_\bot}$.
\end{proof}

Lifted spaces satisfy a property much stronger than compactness\footnote{
A space is compact if whenever $\top \cov U$, then there is a (Kuratowski-) finite subset $K$ of $U$ such that $\top \cov K$.}:
\begin{proposition}
\label{topcov}
If $\top \cov U$, then $\mathsf{T} \in U$.
\end{proposition}
\begin{proof}
It is equivalent to prove that if $\mathsf{T} \cov U$ then $\mathsf{T} \in U$. We can proceed by induction on the proof of covering. The root node of the proof could not have been the use of an axiom, because there are no axioms for covering $\mathsf{T}$. Therefore, there must be some basic open $u \in U$ such that $\mathsf{T} \le u$. Since $\mathsf{T}$ is the largest basic open, it must be that $u \inteq \mathsf{T}$, so $\mathsf{T} \in U$.
\end{proof}

\grammar{This} means that any covering of $A_\bot$ is necessarily trivial. What \grammar{this} means is that one cannot get (nontrivial) information about a point of $A_\bot$ by splitting the point with a cover. \grammar{This} makes sense, because we require refinement by splitting to occur in finite time, and we want to allow $A_\bot$ to represent points that may be partial. To get useful information from a point in a lifted space, one should \emph{prove} that the point actually lies in $\Lift{\top}$.

Every space $A_\bot$ has a point $\bot$ which has no interesting information: it lies only in the entire space. This corresponds to the undefined or nonterminating value. We describe the basic opens that $\bot$ lies in by the inverse image map
\begin{align*}
\bot \liesin \cdot &: \BOpen(A) \to \Open{1}
\\ \bot \liesin \mathsf{T} &\triangleq \top
\\ \bot \liesin \Lift{a} &\triangleq \bot.
\end{align*}

\begin{proposition}
$\bot$ indeed describes a point of $A_\bot$.
\end{proposition}
\begin{proof}
\begin{description}
\item[\irule{meet-0}] Follows from $\bot \liesin \mathsf{T}$.
\item[\irule{meet-2}] Trivially satisfied, since there are no nontrivial intersections of basic opens that $\bot$ lies in.
\item[\irule{join}] Follows from the previous proposition.
\end{description}
\end{proof}

The idea that $\bot$ is the ``least defined'' point can be made formal, in that it is minimal in terms of specialization order.
\begin{proposition}
For any (generalized) point $x : \Gamma \cto A_\bot$, $\bot \le x$ (in terms of specialization order).
\end{proposition}
\begin{proof}
Since $\bot$ lies only in $\mathsf{T}$, it suffices to show that $\top \le x^*(\mathsf{T})$. However, since $\mathsf{T} = \top$, this is equivalent to $\top \le x^*(\top)$, which is true of every continuous map.
\end{proof}

Having established some intuition about what lifted spaces represent, we can confirm they fulfill their original purpose:
\begin{theorem}
There is the bijective correspondence
\begin{mathpar}
\mprset{fraction={===}}
\inferrule*{A \pto B}
  {A \cto B_\bot}.
\end{mathpar}
\end{theorem}
\begin{proof}
Given $f : A \pto B$, we define $g : A \cto B_\bot$ by
\begin{align*}
g^* &: \BOpen(B_\bot) \to \Open{A}
\\ g^*(\mathsf{T}) &\triangleq \top
\\ g^*(\Lift{b}) &\triangleq f^*(b).
\end{align*}

The map $g^*$ preserves $\top$ since $g^*(\mathsf{T}) = \top$ and preserves joins and binary meets since $f^*$ does.

Conversely, given $g : A \cto B_\bot$, we define $f : A \pto B$ by
\begin{align*}
f^* &: \Open{B} \to \Open{A}
\\ f^*(U) &\triangleq g^*(\Lift{U}).
\end{align*}
The map $f^*$ preserves joins and binary meets since $g^*$ and $\Lift$ each do.
\end{proof}

This correspondence in natural in $A$ and $B$, giving an adjunction and thus a monad.
It remains to show that $\cdot_{\gls{Lifted}}$ in fact defines a \emph{strong} monad. Its tensorial strength
$s : A \times \Lifted{B} \pto A \times B$ is defined by the inverse image map
\begin{align*}
s^* &: \BOpen(A \times B) \to \Open{A \times \Lifted{B}}
\\  s^*(a \times b) &\triangleq a \times \Lift{b}.
\end{align*}
The inverse image map $s^*$ preserves joins and binary meets since $\Lift$ does. We claim (but do not prove) that $s$ satisfies the strong monad laws.

Lifted spaces give us a good example of a nontrivial open embedding: the map $\up : A \cto \Lifted{A}$ that is the ``return'' operation of the lifting monad $\Lifted{\cdot}$ is an open embedding. This allows us to view $A$ as an open subspace of $\Lifted{A}$. We can define $\up$ via the inverse and direct image maps
\begin{align*}
\up^* &: \BOpen(\Lifted{A}) \to \Open{A}
\\ \up^*(\mathsf{T}) &\triangleq \top
\\ \up^*(\Lift{a}) &\triangleq a
\\
\\ \up_! &: \Open{A} \to \Open{\Lifted{A}}
\\ \up_!(U) &\triangleq \Lift{U}.
\end{align*}

\begin{proposition}
\grammar{This} indeed defines an open embedding $\up : A \hookto \Lifted{A}$.
\end{proposition}
\begin{proof}
We first confirm $\up$ defines a continuous map:
\begin{description}
\item[\irule{meet-0}] 
\[
\up^*(\top) \ge \up^*(\mathsf{T}) = \top.
\]
\item[\irule{meet-2}] Follows from the facts
\begin{align*}
\mathsf{T} \omeet \Lift{a} &= \Lift{a}
\\ \Lift{a} \omeet \Lift{b} &= \Lift(a \omeet b).
\end{align*}
\item[\irule{join}] It suffices to check every axiom of $\Lifted{A}$: every axiom is of the form $\Lift{a} \cov_{\Lifted{A}} \Lift{U}$, where $a \cov_A U$. So we must confirm
\[
\up^*(\Lift{a}) \cov \up^*(\Lift{U}),
\]
which is equivalent to $a \cov U$, which we know by assumption.
\end{description}

We will now confirm that $\up_! = \Lift$ preserves joins and binary meets. We observe that it preserves binary meets since
\[
\Lift{a} \omeet \Lift{b} = \Lift{(a \omeet b)}.
\]
It preserves joins since if $a \cov_A U$ then $\Lift{a} \cov_{\Lifted{A}} \Lift{U}$.

Finally, it is straightforward to confirm that $\up_! \dashv \up^*$ and that the Frobenius law holds.
\end{proof}

\subsubsection*{Nondeterministic powerspaces}

As mentioned previously, the ``forgetful'' functor from $\cat{FSpc}$ to $\cat{FSpc}_{nd}$ has a right adjoint, such that there is a monad $\PLowerP : \cat{FSpc} \to \cat{FSpc}$ giving a correspondence
\begin{mathpar}
\mprset{fraction={===}}
\inferrule*{A \ndto B}
  {A \cto \PLowerP(B)}
\end{mathpar}
between continuous maps and their nondeterministic counterparts \cite{vickersdoublepowerlocale}. By this correspondence, we can think of the space $\PLowerP(B)$ as the space of nondeterministic values. Spaces of the form 
$\PLowerP(B)$ are known as \emph{positive lower powerspaces} (also known as positive or inhabited Hoare powerlocales).

There is a map $\possibly : \Open{A} \to \Open{\PLowerP(A)}$ (read ``possibly'') that can be intuitively understood in the following sense: if $U : \Open{A}$ is interpreted as a property of points of $A$, and a point $s : \One \cto \PLowerP(A)$ (where $\One$ is the one-point space) is interpreted as a subspace of $A$, $\possibly U : \Open{\PLowerP(A)}$ holds of $s$ if $U$ holds of some point in $s$.

Just like $\Lift$ for lifted spaces, the $\possibly$ operator preserves all the structure in the lattice $\Open{A}$ that ought to be preserved in $\Open{\PLowerP(A)}$: that is, it preserves joins and $\top$.
\grammar{This} implies a particular fact that is critical in the computational properties of the nondeterministic powerspaces:
\begin{theorem}
For every cover $\top \le \bigvee_{i : I} V_i$ in $A$ there is a cover $\top \le \bigvee_{i : I} \possibly V_i$ in $\PLowerP(A)$.
\end{theorem}
\begin{proof}
Suppose $\top \le \bigvee_{i : I} V_i$ in $A$. Since $\possibly$ preserves $\top$ and joins,
\[
\bigvee_{i : I} \possibly V_i
= \possibly \left( \bigvee_{i : I} V_i \right)
= \possibly \top
= \top.
\]
\end{proof}
The previous theorem can be understood computationally as allowing simulation of some nondeterministic result.

Classically, the points of $\PLowerP(B)$ are in correspondence with closed nonempty subspaces of $B$. Constructively, these subspaces also are \emph{overt} \cite{vickersdoublepowerlocale}, which is helpful for some computational tasks.

The specialization order on $\PLowerP(A)$ corresponds to subspace inclusion of possible values. 

\citet{vickersdoublepowerlocale} describes how to construct this powerspace, as well the powerspace for $\cat{FSpc}_{nd, p}$, predicatively within \textbf{FSpc}.

\citet{vickersdoublepowerlocale} shows that $\PLowerP$ defines a monad, but it remains to show that $\PLowerP$ in fact defines a \emph{strong} monad. Its tensorial strength
$s : A \times \PLowerP(B) \ndto A \times B$ is defined by the inverse image map
\begin{align*}
s^* &: \BOpen(A \times B) \to \Open{A \times \PLowerP(B)}
\\  s^*(a \times b) &\triangleq a \times \possibly b.
\end{align*}
The inverse image map $s^*$ preserves joins and $\top$ since $\possibly$ does.

\subsubsection*{Lower powerspaces}

As mentioned previously, the ``forgetful'' functor from $\cat{FSpc}$ to $\cat{FSpc}_{nd, p}$ has a right adjoint, such that there is a monad $\PLower : \cat{FSpc} \to \cat{FSpc}$ giving a correspondence
\begin{mathpar}
\mprset{fraction={===}}
\inferrule*{A \ndpto B}
  {A \cto \PLower(B)}
\end{mathpar}
between continuous maps and their nondeterministic and partial counterparts \cite{vickersdoublepowerlocale}. By this correspondence, we can think of the space $\PLower(B)$ as the space of nondeterministic and partial values.
Spaces of the form $\PLower(B)$ are known as \emph{lower powerspaces} (also known as Hoare powerlocales).

There is a map $\possibly : \Open{B} \to \Open{\PLower(B)}$ (read ``possibly'') which distributes over joins but not necessarily meets. In particular,
\[
\possibly U \le \bigvee_{i : I} \possibly V_i
\]
holds in $\Open{\PLower(B)}$ whenever $U \le \bigvee_{i : I} V_i$ in $\Open{B}$.

The monad $\PLower$ is, like $\PLowerP$, a strong monad, and its tensorial strength is almost the same: $s : A \times \PLower(B) \ndpto A \times B$ is defined by the inverse image map
\begin{align*}
s^* &: \BOpen(A \times B) \to \Open{A \times \PLower(B)}
\\  s^*(a \times b) &\triangleq a \times \possibly b.
\end{align*}
The inverse image map $s^*$ preserves joins since $\possibly$ does.


\section{Marshall language details}
\label{a:marshall}

In this appendix, we describe the Marshall language and our modifications to it: we describe its syntax, denotational semantics based on constructive topology, and operational semantics, and explain the connection between them.

\subsection{Syntax}

The Marshall language's syntax is described in Fig. \ref{marshall-syntax}.
Our syntax extensions are presented by Fig. \ref{marshall-syntax-ext}.
Fig. \ref{marshall-types} presents its typing rules,
and Fig. \ref{marshall-types-ext} presents the additional typing rules for the constructs that we added.

\begin{figure}
\small
\begin{align*}
\text{expression}\ e ::=\ &x
\\
\gor &\lstinline!True!
\gor \lstinline!False!
\\
\gor &q
\\
\gor & \lstinline!cut!\ x : r\ \lstinline!left!\ e\ \lstinline!right!\ e 
\\
\gor & e\ \lstinline!/\\!\ \ldots\ \lstinline!/\\!\ e
\\
\gor & e\ \lstinline!\\/!\ \ldots\ \lstinline!\\/!\ e
\\
\gor & e < e
\\ \gor & \lstinline!exists!\ x : r \lstinline!,!\ e
\gor \lstinline!forall!\ x : r \lstinline!,!\ e
\\ \gor & \lstinline!(! e_1\ \lstinline!,!\ \ldots\ \lstinline!,!\ e_n  \lstinline!)! 
\\
\gor &e \lstinline!#! k
\\ \gor & \lstinline!fun!\ x\ \lstinline!:!\ t\ \lstinline!=>!\ e
\\ \gor & e \ e
\\ \gor & e\otimes e
\gor e\ \lstinline!^!\ k
\\ \gor & \lstinline!let!\ x\ \lstinline!=!\ e\ \lstinline!in!\ e
\\
\\ \text{type}\ t ::=\ &
\lstinline!real!
\gor \lstinline!prop!
\gor t\ \lstinline!*!\ \ldots\ \lstinline!*!\ t
\gor t\,\,\, \lstinline!->!\ t
\\ \text{arith op}\ \otimes ::=\ &
\lstinline!+!\
\gor \lstinline!-!\
\gor \lstinline!*!\
\gor \lstinline!/!\
\\ \text{variable}\ x \phantom{::=\ }&
\\ \text{rational}\ q \phantom{::=\ }&
\\ \text{natural}\ k\phantom{::=\ }&
\\ \text{range limit}\ l ::=\ &
\lstinline!-inf!
\gor \lstinline!inf!
\gor \lstinline!q!
\\ \text{range}\ r ::=\ &
\lstinline!(!\ l\ \lstinline!,!\ l\ \lstinline!)!
\gor \lstinline![!\ l\ \lstinline!,!\ l\ \lstinline!]!
\end{align*}
\caption{Marshall syntax}
\label{marshall-syntax}
\end{figure}

\begin{figure}
\small
\begin{align*}
\text{expression}\ e ::=\
 &e\ \lstinline!\~>!\ e
\\ \gor &e \lstinline!||!\ \ldots\ \lstinline!||!\ e
\\ \gor &\lstinline!mkbool!\ e\ e
\gor \lstinline!is_true!\ e
\gor \lstinline!is_false!\ e
\\
\\ \text{type}\ t ::=\ &
\lstinline!bool!
\end{align*}
\caption{Marshall syntax extensions}
\label{marshall-syntax-ext}
\end{figure}

\newcommand{\lcode}[1]{\texttt{#1}}

\begin{figure}
\begin{mathpar}
\inferrule*
  { }
  {\Gamma \vdash \lcode{True} : \lcode{prop}}

\inferrule*
  { }
  {\Gamma \vdash \lcode{False} : \lcode{prop}}
  
\inferrule*
  { }
  {\Gamma \vdash q : \lcode{real}}
  
\inferrule*
  { \Gamma, x : \lcode{real} \vdash e_1 : \lcode{prop}
  \\ \Gamma, x : \lcode{real} \vdash e_2 : \lcode{prop}}
  { \Gamma \vdash \lcode{cut}\ x : r\ \lcode{left}\ e_1\ \lcode{right}\ e_2 : \lcode{real}} 

\inferrule*
  {\Gamma \vdash e_1 : \lcode{prop}
  \\ \ldots
  \\ \Gamma \vdash e_n : \lcode{prop} }
  {\Gamma \vdash e_1\ \lcode{/\textbackslash} \ldots \lcode{/\textbackslash}\ e_n : \lcode{prop}}
  
\inferrule*
  {\Gamma \vdash e_1 : \lcode{prop}
  \\ \ldots
  \\ \Gamma \vdash e_n : \lcode{prop} }
  {\Gamma \vdash e_1\ \lcode{\textbackslash/} \ldots \lcode{\textbackslash/}\ e_n : \lcode{prop}}
  
\inferrule*
  {\Gamma \vdash e_1 : \lcode{real}
  \\ \Gamma \vdash e_2 : \lcode{real} }
  {\Gamma \vdash e_1 < e_2 : \lcode{prop}}
\\
\inferrule*
  {\Gamma, x : \lcode{real} \vdash e : \lcode{prop}}
  {\Gamma \vdash \lcode{exists}\ x : r\lcode{,}\ e : \lcode{prop}}
  
\inferrule*
  {\Gamma, x : \lcode{real} \vdash e : \lcode{prop}}
  {\Gamma \vdash \lcode{forall}\ x : r\lcode{,}\ e : \lcode{prop}}
  
\inferrule*
  {\Gamma \vdash e_1 : t_1
  \\ \ldots
  \\ \Gamma \vdash e_n : t_n}
  {\Gamma \vdash \lcode{(}e_1\lcode{,} \ldots\lcode{,}\ e_n\lcode{)} : t_1\ \lcode{*} \ldots \lcode{*}\ t_n}
  
\inferrule*
  {\Gamma \vdash e: t_1\ \lcode{*} \ldots \lcode{*}\ t_n}
  {\Gamma \vdash e\ \lcode{\#}\ k : t_k}
  
\inferrule*
  {\Gamma, x : t_1 \vdash e : t_2}
  {\Gamma \vdash \lcode{fun}\ x : t\ \lcode{=>}\ e : t_1\ \lcode{->}\ t_2}
  
\inferrule*
  {\Gamma \vdash e_1 : t_A\ \lcode{->}\ t_B
  \\ \Gamma \vdash e_2 : t_A}
  {\Gamma \vdash e_1\ e_2 : t_B}

\inferrule*
  {\Gamma \vdash e_1 : \lcode{real}
  \\ \Gamma \vdash e_2 : \lcode{real}}
  {\Gamma \vdash e_1 \otimes e_2 : \lcode{real}}

\inferrule*
  {\Gamma \vdash e : \lcode{real}}
  {\Gamma \vdash e\ \lcode{\^{}}\ k : \lcode{real}}
  
\inferrule*
  {\Gamma \vdash e_1 : t_1
  \\ \Gamma, x : t_1 \vdash e_2 : t_2}
  {\Gamma \vdash \lcode{let}\ x\ \lcode{=}\ e_1\ \lcode{in}\ e_2 : t_2}
\end{mathpar}
\caption{Marshall typing rules}
\label{marshall-types}
\end{figure}

\begin{figure}
\begin{mathpar}
\inferrule*
  { \mathsf{base}(t)
  \\ \Gamma \vdash e_1 : \lcode{prop}
  \\ \Gamma \vdash e_2 : t }
  {\Gamma \vdash e_1\ \lcode{\textasciitilde>}\ e_2 : t}
  
\inferrule*
  { \mathsf{base}(t)
  \\ \Gamma \vdash e_1 : t
  \\ \ldots
  \\ \Gamma \vdash e_n : t}
  {\Gamma \vdash e_1\lcode{||} \ldots\lcode{||}\ e_n : t}
  
\inferrule*
  {\Gamma \vdash e_1 : \lcode{prop}
  \\ \Gamma \vdash e_2 : \lcode{prop} }
  {\Gamma \vdash \lcode{mkbool}\ e_1\ e_2 : \lcode{bool}}
\\
\inferrule*
  {\Gamma \vdash e : \lcode{bool}}
  {\Gamma \vdash \lcode{is\_true}\ e : \lcode{prop}}  

\inferrule*
  {\Gamma \vdash e : \lcode{bool}}
  {\Gamma \vdash \lcode{is\_false}\ e : \lcode{prop}}  
\end{mathpar}
\caption{Marshall typing rule extensions}
\label{marshall-types-ext}
\end{figure}

\subsection{Denotational semantics}

We can ascribe the original version of Marshall language a categorical semantics in the category of presheaves over $\cat{FSpc}$. Types denote objects (Fig. \ref{marshall-type-denote}) and expressions denote maps (Fig. \ref{marshall-term-denote}) in this category. Presumably, it would be possible to restrict to some subcategory of sheaves, but we will not do so, as we are primarily interested in the behavior of the first-order fragment of the language. A first-order expression denotes a continuous map.

The Dedekind cut constructor is generally interpreted according to $\mathsf{from\_cut}$, whose inverse image map is defined as\footnote{
Because $\R$ is locally compact, the exponential object $\Sierp^{\R}$ exists, but even if it did not, this definition could be interpreted appropriately within the presheaf category.}
\begin{align*}
\mathsf{from\_cut} &: \Sierp^{\R} \times \Sierp^{\R} \ndpto \R
\\ \left(\mathsf{from\_cut}(L, U)\right)^*(B_\varepsilon(q)) &\triangleq
     L(q - \varepsilon) \wedge U(q + \varepsilon),
\end{align*}
recalling that it suffices to define an inverse image map by its behavior on basic opens. 
However, we only consider a Marshall Dedekind cut to be a valid expression (both in its original version and in our modified version) if it in fact denotes a total and deterministic real number.

Rational numbers are denoted in the following straightforward manner:
\begin{align*}
\mathsf{from\_rat} &: \rat \cto \R
\\ \left(\mathsf{from\_rat}(q')\right)^*(B_\varepsilon(q)) &\triangleq
     q - \varepsilon < q' < q + \varepsilon.
\end{align*}

In our modification to Marshall that enables partiality and nondeterminism, we denote into the category of presheaves over $\cat{FSpc}_{nd, p}$ rather than $\cat{FSpc}$. Accordingly, first-order expressions denote maps in $\cat{FSpc}_{nd, p}$. Our extended denotational semantics is presented for types in Fig. \ref{marshall-type-denote-ext} and for terms in Fig. \ref{marshall-term-denote-ext}.

We add nondeterministic joins \lstinline!||! as well as the partial restriction operator \lstinline!~>!, but only allow these operations to be applied on expressions of \emph{base types}, which are those types which don't involve any function arrows. Every base type corresponds to (the Yoneda embedding of) a space. We write $\mathsf{base}(t)$ to indicate that $t$ is a base type.

The restriction operator \lstinline!~>! that we added induces partiality. We define the restriction operator $\restrict{\cdot}{\cdot} : A \times \Sigma \pto A$ by the inverse image map
\begin{align*}
\left(\restrict{e}{U}\right)^*(V) \triangleq U \wedge e^*(V).
\end{align*}

\begin{figure}
\small
\begin{align*}
\denote{\lstinline!real!} &\triangleq \R
\\ \denote{\lstinline!prop!} &\triangleq \Sigma
\\ \denote{t_1\ \lstinline!*!\ \ldots\ \lstinline!*!\ t_n} &\triangleq
    \denote{t_1} \times \ldots \times \denote{t_n}
\\ \denote{t_1\,\,\, \lstinline!->!\ t_2} &\triangleq \denote{t_1} \to \denote{t_2}
\end{align*}
\caption{Marshall denotation of types}
\label{marshall-type-denote}
\end{figure}

\newcommand{\subst}[3]{ {[#2/#1]#3} }
\begin{figure}
\small
\begin{align*}
\denote{\lstinline!True!} &\triangleq \top_\Sigma
\\ \denote{\lstinline!False!} &\triangleq \bot_\Sigma
\\ \denote{q} &\triangleq \mathsf{from\_rat}(q)
\\ \denote{\lstinline!cut!\ x : r\ \lstinline!left!\ e_1\ \lstinline!right!\ e_2}
   &\triangleq \mathsf{from\_cut}(\lambda x.\ \denote{e_1}, \lambda x.\ \denote{e_2})
\\ \denote{e_1\ \lstinline!/\\!\ \ldots\ \lstinline!/\\!\ e_n}
  &\triangleq \denote{e_1} \wedge \ldots \wedge \denote{e_n}
\\ \denote{e_1\ \lstinline!\\/!\ \ldots\ \lstinline!\\/!\ e_n}
  &\triangleq \denote{e_1} \wedge \ldots \wedge \denote{e_n}
\\ \denote{e_1 < e_2} &\triangleq \denote{e_1} < \denote{e_2}
\\ \denote{\lstinline!exists!\ x : r \lstinline!,!\ e}
  &\triangleq \exists x : \denote{r}.\ \denote{e}
\\ \denote{\lstinline!forall!\ x : r \lstinline!,!\ e}
  &\triangleq \forall x : \denote{r}.\ \denote{e}
\\ \denote{e \lstinline!#! k}
  &\triangleq \pi_k \denote{e}
\\ \denote{\lstinline!fun!\ x\ \lstinline!:!\ t\ \lstinline!=>!\ e}
  &\triangleq \lambda x. \denote{e}
\\ \denote{e_1 \ e_2}
  &\triangleq \denote{e_1}\ \denote{e_2}
\\ \denote{e_1 \otimes e_2}
  &\triangleq \denote{\otimes}(\denote{e_1}, \denote{e_2})
\\ \denote{e_1\ \lstinline!^!\ k}
  &\triangleq \denote{e_1}^k
\\ \denote{\lstinline!let!\ x\ \lstinline!=!\ e_1\ \lstinline!in!\ e_2}
  &\triangleq{ \denote{\subst{x}{e_1}{e_2}}}
\end{align*}
\caption{Marshall denotation of terms}
\label{marshall-term-denote}
\end{figure}

\begin{figure}
\small
\begin{align*}
\denote{\lstinline!bool!} &\triangleq \bool
\end{align*}
\caption{Marshall extension denotation of types}
\label{marshall-type-denote-ext}
\end{figure}

\begin{figure}
\small
\begin{align*}
\denote{e_1\ \lstinline!\~>!\ e_2} &\triangleq \restrict{\denote{e_2}}{\denote{e_1}}
\\ \denote{e_1 \lstinline!||!\ \ldots\ \lstinline!||!\ e_n}
  &\triangleq \denote{e_1} \sqcup \ldots \sqcup \denote{e_n}
\\ \denote{\lstinline!mkbool!\ e_1\ e_2}
  &\triangleq \restrict{\btrue}{\denote{e_1}} \sqcup \restrict{\bfalse}{\denote{e_2}}
\\ \denote{\lstinline!is_true!\ e}
  &\triangleq \denote{e} = \btrue
\\ \denote{\lstinline!is_false!\ e}
  &\triangleq \denote{e} = \bfalse
\end{align*}
\caption{Marshall extension denotation of terms}
\label{marshall-term-denote-ext}
\end{figure}

\subsection{Operational semantics}

Fundamentally, Marshall expressions are executed in two steps. First, a Marshall term is normalized to a normal form. Some normalization is done under binders. The normalization procedure eliminates local definitions and redexes, reducing applications of lambda expressions and tuple projections of tuple values.

Our modified version of Marshall also handles nondeterministic joins during normalization. All expressions are reduced to an $n$-ary join of join-free expressions.

After normalization, Marshall expressions are successively refined until they are fine enough to return a result. Expressions of function types will be lambda expressions (since they have already been normalized) and return immediately. For expressions of base types, the evaluation relation $e \Downarrow_p U$, described in Fig. \ref{marshall-eval}, says when a Marshall expression $e$ is sufficiently refined to know that $\denote{e}$ lies in $U$ (formally, $\denote{e}^*U = \top$), where $p : \rat^+$ target precision that allows the user to control when a real number is sufficiently precise that it should no longer be refined.

The function $\mathsf{real\_approx}$ approximates the result of real-valued expressions using interval arithmetic and Kaucher multiplication \cite{marshall, intawi}. For a Dedekind cut, $\mathsf{real\_approx}$ simply uses the range over which the cut variable is bound as the interval approximation.

We extended Marshall's evaluation relation to handle Boolean expression and nondeterministic joins, as shown in Fig. \ref{marshall-eval-ext}. The rule for evaluation of nondeterministic joins says that if any expression in a join evaluates to some $U$, then the entire join does.

\newcommand{\done}[2]{#1 \Downarrow_p #2}

\begin{figure}
\begin{mathpar}
\inferrule*
  { }
  {\done{\lcode{True}}{\cdot = \strue}}
  
\inferrule*
  {\mathsf{real\_approx}(e) = (a, b) \\ b - a < p}
  {\done{e}{B_{(b-a)/2}\left(\frac{a + b}{2}\right)}}
  
\inferrule*
  {\done{e_1}{U_1}
  \\ \ldots
  \\ \done{e_n}{U_n} }
  {\done{\lcode{(}e_1\lcode{,} \ldots\lcode{,}\ e_n\lcode{)}}
             {(U_1, \ldots, U_n)}
  }
\end{mathpar}
\caption{Marshall evaluation relation}
\label{marshall-eval}
\end{figure}

\begin{figure}
\begin{mathpar}
\inferrule*
  {\done{e_i}{U}}
  {\done{e_1\lcode{||} \ldots\lcode{||}\ e_n}{U}}
\\
\inferrule*
  { }
  {\done{ \lcode{mkbool}\ \lcode{True}\ e_2 }{ \cdot = \btrue }}

\inferrule*
  { }
  {\done{ \lcode{mkbool}\ e_1\ \lcode{True} }{\cdot = \bfalse }}
\end{mathpar}
\caption{Extensions to Marshall evaluation relation}
\label{marshall-eval-ext}
\end{figure}

If a (normalized) Marshall expression is not yet fine enough to evaluate to a result, one step of refinement is performed, which we denote here by the function $r$. Refinements are performed repeatedly until the expression evaluates to a result. An expression may be refined repeatedly without ever evaluating, meaning that the expression diverges. Refinements preserve the meanings of expressions, i.e., $\denote{r(e)} = \denote{e}$.

Refinement of \lstinline!prop!-valued expressions proceeds as follows: there are procedures for computing lower and upper approximants of any \lstinline!prop!-valued expression. That is, given an expression $\phi$, Marshall computes expressions $\phi^-$ and $\phi^+$ such that $\denote{\phi^-} \le \denote{\phi} \le \denote{\phi^+}$ (where $\le$ is specialization order). \citet{marshall} explains how these approximants are computed. If $\phi^- = \lstinline!True!$, then $r(\phi) = \lstinline!True!$, and if $\phi^+ = \lstinline!False!$, then $r(\phi) = \lstinline!False!$. If neither of these is true, then refinement proceeds as follows.

Many constructs do not directly participate in refinement; refinement simply proceeds into their arguments. For instance, $r(e_1 \otimes e_2) = r(e_1) \otimes r(e_2)$. Refinement works similarly for exponentiation and \lstinline!prop!-valued conjunction and disjunction.

Refinement of Dedekind cuts, existential quantifiers, and universal quantifiers, is a more sophisticated process, detailed further by \citet{marshall}. Dedekind cuts are refined to cuts where, fundamentally, the range that the variable is cut over is more narrow. Existentially quantified statements are refined by splitting into disjunctions of existentially quantified statements over smaller ranges, and universally quantified statements are refined into conjunctions of universally quantified statements over smaller ranges.

We extend the refinement relation in a rather straightforward way: refinements of nondeterministic joins are joins of their refinements, and $r(\lstinline!mkbool!\ e_1\ e_2) =  \lstinline!mkbool!\ r(e_1)\ r(e_2)$. To refine the restriction operator, we have $r(\lstinline!True ~>! e_2) = r(e_2)$ and $r(e_1 \lstinline!~>! e_2) = r(e_1) \lstinline!~>! r(e_2)$ when $e_1 \ne \lstinline!True!$.

\subsection{Computational soundness and adequacy}

On an input expression $e$ (whose type has no function arrows) with target precision $p$, Marshall returns a response $U$ if, after normalizing $e$ to $e'$, after some number of refinements $n$ of $e'$, the result $r^n(e')$ evaluates to $U$, i.e., $\done{r^n(e')}{U}$.

Both normalization and refinement of expressions should preserve their denotations, such that $\denote{e} = \denote{r^n(e')}$. Together with the fact that the evaluation relation $\Downarrow_p$ should satisfy the proposition that if $\done{e}{U}$ then $\denote{e}^*(U) = \top$, we can conclude that if Marshall returns a response $U$ for an input expression $e$, then $\denote{e}^*(U) = \top$, i.e., $e$ lies in $U$.

Marshall should satisfy some notion of computational adequacy as well, whose statement is more complicated: given an input expression $e$ with base type $t$, then depending on the type $t$:

\begin{itemize}
\item \lstinline!real!: If $\denote{e}$ is total, i.e., it factors through a map $\One \ndto \R$, then for every $\varepsilon > 0$, setting the target precision to $\varepsilon$ returns a result $B_\varepsilon'(q)$, where $\varepsilon' < \varepsilon$.
\item \lstinline!prop!: If $\denote{e} = \strue$, then (regardless of the target precision $p$), Marshall returns the result $\cdot = \strue$.
\item \lstinline!bool!: If $\denote{e}$ is total, i.e., it factors through a map $\One \ndto \bool$, then (regardless of target precision $p$), Marshall returns either $\cdot = \btrue$ or $\cdot = \bfalse$.
\item $t_1 \lstinline!*! \ldots \lstinline!*! t_n$: Marshall returns a result if each of its projections returns a result.
\end{itemize}

\printglossary[title=List of symbols, toctitle=List of symbols]
\end{longversion}

\end{document}